\documentclass[twoside,11pt]{article}

\usepackage{jmlr2e}

\usepackage{amssymb}
\usepackage{amsmath}
\usepackage{cleveref}
\usepackage{enumerate}
\usepackage{cmll}
\usepackage{amsmath}

\allowdisplaybreaks[3]
\usepackage{bm}
\usepackage{natbib}
\usepackage{color}
\usepackage{float}
\usepackage{graphicx}
\usepackage{algorithm}
\usepackage{color}
\usepackage{algpseudocode}
\usepackage{epsfig}
\usepackage{threeparttable}
\usepackage{dcolumn}
\usepackage{multirow}
\usepackage{booktabs}
\usepackage{mathtools}
\usepackage{mathrsfs}
\usepackage{array}

\graphicspath{{figures/}} 	
\usepackage{subfigure} 	
\usepackage{geometry}

\makeatletter
\newenvironment{breakablealgorithm}
{
	\begin{center}
		\refstepcounter{algorithm}
		\hrule height.8pt depth0pt \kern2pt
		\renewcommand{\caption}[2][\relax]{
			{\raggedright\textbf{\ALG@name~\thealgorithm} ##2\par}%
			\ifx\relax##1\relax 
			\addcontentsline{loa}{algorithm}{\protect\numberline{\thealgorithm}##2}%
			\else 
			\addcontentsline{loa}{algorithm}{\protect\numberline{\thealgorithm}##1}%
			\fi
			\kern2pt\hrule\kern2pt
		}
	}{
		\kern2pt\hrule\relax
	\end{center}
}
\makeatother

\def\argmin{\mathop{\rm argmin}}

\newcommand{\E}{{\rm E}}
\newcommand{\var}{{\rm Var}}
\renewcommand{\P}{{\rm P}}
\newcommand{\T}{{\rm T}}

\jmlrheading{}{}{}{}{}{}{Ruoyu Wang and Miaomiao Su and Qihua Wang}

\ShortHeadings{Distributed nonparametric regression imputation}{Wang and Su and Wang}
\firstpageno{1}

\begin{document}

\title{Distributed nonparametric regression imputation for missing response problems with large-scale data}

\author{\name Ruoyu Wang \email wangruoyu17@mails.ucas.edu.cn \\
       \addr Academy of Mathematics and Systems Science\\
       Chinese Academy of Sciences\\
       Beijing, 100190, China
       \AND
       \name Miaomiao Su \email smm@amss.ac.cn \\
       \addr Academy of Mathematics and Systems Science\\
       Chinese Academy of Sciences\\
       Beijing, 100190, China
       \AND
       \name Qihua Wang \email qhwang@amss.ac.cn\\
       \addr Academy of Mathematics and Systems Science\\
       Chinese Academy of Sciences\\
       Beijing, 100190, China}

\editor{}

\maketitle

\begin{abstract}
	Nonparametric regression imputation is commonly used in missing data analysis.
	However, it suffers from the ``curse of dimension".  
	The problem can be alleviated by the explosive sample size in the era of big data, while the large-scale data size presents some challenges on the storage of data and the calculation of estimators. 
	These challenges make the classical nonparametric regression imputation methods no longer applicable. 
	This motivates us to develop two distributed nonparametric regression imputation methods. 
	One is based on kernel smoothing and the other on the sieve method. 
	The kernel-based distributed imputation method has extremely low communication cost and the sieve-based distributed imputation method can accommodate more local machines. To illustrate the proposed imputation methods,  response mean estimation is considered.
	Two distributed nonparametric regression imputation estimators are proposed for the response mean, which are proved to be asymptotically normal with asymptotic variances achieving the semiparametric efficiency bound. 
	The proposed methods are evaluated through simulation studies and are illustrated by a real data analysis.
\end{abstract}


\section{Introduction}

Missing data is a common issue that practitioners may face in data analysis. 
A typical example of missing data is the missing response problem such as nonresponse in sample surveys and dropout in clinical trials. 
For missing response problems, the nonparametric regression imputation methods, which can often produce robust and efficient estimates \citep{Cheng1994nonparametric, hahn1998role}, are commonly used to deal with missing values.
However, the nonparametric regression imputation methods suffer from the ``curse of dimension",  which deteriorates their finite sample performance since the sample size for achieving a given estimation accuracy needs to increase exponentially as the dimension of the covariate vector.
This problem restricts the application of the nonparametric regression imputation methods to high dimension regression problems because of the limitation of sample size. We conduct a detailed discussion of the ``curse of dimension" in nonparametric regression imputation estimation in Section \ref{subsec: kernel method}.
In the big data era, however, the sample size is extremely large in some cases.
On the one hand, the large sample size makes it possible to achieve the desired estimation accuracy using the nonparametric regression imputation methods when the dimension of the covariate vector is large.
On the other hand, it may be infeasible to keep the large-scale data set in memory or even store all the data on a single computer when the data size is too large \citep{fan2014challenges, wang2016statistical}. And computations of the nonparametric regression imputation estimations are usually time-consuming or even infeasible. 
These motivate us to develop distributed nonparametric regression imputation methods to solve the computation problem while retaining the good theoretical properties of the classical nonparametric methods and mitigating the so-called curse of dimension. 
To our knowledge, this problem has not been investigated in the literature. 
In this paper, we consider the estimation problem in the presence of missing responses when data are stored distributionally on different machines, and develop two distributed nonparametric regression imputation methods by focusing on the estimation of response mean. 

To reduce the computational burden, we adopt the divide and conquer strategy.
The main idea of this strategy is to calculate some summary statistics on each local machine and then aggregate the results from local machines to get a final estimate.
By doing so, we can substantially reduce the computing time and alleviate the computer memory requirements. 
The divide and conquer approach was first studied by \cite{Mcdonald2009dist} for multinomial regression.  
\cite{Zhang2013dist} rigorously showed that in parametric estimation problems with fully observed data, the divide and conquer approach generally has greater efficiency than the naive approach that uses only the sample on a single machine. 
\cite{Lee2017dist} and \cite{battey2018distributed} extended this approach to sparse linear model in the high-dimensional setting. \cite{tang2020distributed} further investigated the application of the divide and conquer strategy in sparse generalized linear models. Distributed inference problem with non-smooth loss functions, e.g., quantile regression, was studied via this strategy by \cite{Volgushev2019quantile}. 

In this paper, we extend the divide and conquer strategy to missing response problems. We first develop a kernel-based distributed imputation (KDI) approach to estimate the response mean.
Under some mild conditions, it is shown that the resulting estimator is asymptotically normal and has the same asymptotic variance as the classical nonparametric kernel regression imputation estimator \citep{Cheng1994nonparametric}, which achieves the semiparametric efficiency bound (see Appendix A for a brief introduction of the bound). 
The KDI method needs very little communication between machines because only a real value is required to transmit from each local machine to the central machine. 
As discussed in \cite{Zhang2013dist}, communication between different machines may be prohibitively expensive, and the difference in communication complexity between different algorithms can be significant.
Thus distributed estimation methods that require fairly limited communication are of more interest. 
Clearly, our KDI method is just this type.

However, the KDI method introduces additional bias if the number of machines is too large, which is a common drawback of one-shot communication approaches \citep{jordan2019communication}. 
To overcome this problem, we propose an alternative multi-round imputation method based on another nonparametric method, sieve method \citep{newey1994asymptotic,ai2003efficient,chen2007large,chen2012estimation, belloni2015some}, and call this method the sieve-based distributed imputation (SDI) method. 
Compared to the KDI method, the SDI method can accommodate distributed systems with more machines and thus has the potential of further reducing computing time. 
Under certain conditions, the SDI estimator of the response mean is also asymptotically normal and the asymptotic variance achieves the semiparametric efficiency bound. 
However, the SDI method needs more communication than the KDI.
If the communication cost is high, the KDI method is recommended and one should limit the number of local machines to avoid the additional bias of this method.
If the communication cost is low, the  SDI method is recommended and one can use a lot of local machines to reduce computing time.

The rest of this paper is organized as follows. In Section \ref{sec: kernel}, we present the KDI estimator and derive its theoretical properties.
In Section \ref{sec: sieve}, we propose the SDI method and present its theoretical properties. 
Simulation results are given in Section \ref{sec: sim}. 
As an illustration, a real data analysis is provided in Section 5. 
All proofs are relegated to the Appendix.

\section{Kernel-Based One-shot Method}\label{sec: kernel} 
\subsection{Methodology}\label{subsec: kernel method}

Let $Y$ be the response variable and $X$ the $d$-dimension completely observed covariate vector. Suppose we have i.i.d. incomplete observations $\{(\delta_{i}, Y_{i}, X_{i}): i=1,\dots,N\}$, 
where $\delta_{i}=1$ if $Y_{i}$ is observed and $\delta_{i} = 0$ otherwise. 
A classical nonparametric method to estimate the response mean $\mu$ is the nonparametric kernel regression imputation method due to \cite{Cheng1994nonparametric}. The main idea of regression imputation is to impute the missing response $Y$ by its conditional mean $m(X)=\E[Y\mid X]$. Throughout this paper, we assume that response is missing at random (MAR).
That is, $Y\Perp\delta\mid X$. Under MAR,  $\E[Y\mid X]=\E[Y\mid X,\delta=1]$. Thus to estimate $\mu$, \cite{Cheng1994nonparametric} first estimate $m(x)=\E[Y\mid X=x]$ by
\begin{equation*}\label{est: m1}
	\hat{m}_{\mathbb{K}}(x)=\frac{\sum^N_{i=1}K_{h}(X_i-x)\delta_iY_i}{\sum^N_{i=1}K_{h}(X_i-x)\delta_i}
\end{equation*} 
and then the final estimator of $\mu$ is given by 
\begin{equation}\label{est: no truct}
	\hat{\mu}_{\mathbb{K}}= N^{-1}\sum^N_{i=1}\{\delta_iY_i+(1-\delta_i)\hat{m}_{\mathbb{K}}(X_i)\},
\end{equation}
where $K_{h}(\cdot) = h^{-d}K(\cdot/h)$, $K(\cdot)$ is some kernel function and $h$ is a bandwidth sequence that decreases to zero as $N\to\infty$. 

As pointed out previously, the estimator $\hat{\mu}_{\mathbb{K}}$ suffers from the ``curse of dimension" problem when the dimension of $X$ is high. The nonparametric regression imputation is regarded reliable only when the number of covariates is very small in some existing works \citep{Hu2012core,chen2017finite} due to this problem. Fortunately, this problem can be mitigated by the large sample size 
in the big data era. Next, we give some heuristic discussions. 
Under some mild conditions, $\hat{\mu}_{\mathbb{K}} - \mu$ admits the following decomposition
\begin{equation}\label{eq: N}
	\hat{\mu}_{\mathbb{K}} - \mu = \psi_{N} + R_{N}
\end{equation}
where $\psi_{N}=N^{-1}\sum_{i=1}^{N}\left\{\delta_{i}Y_{i}/\pi(X_{i}) + (\pi(X_{i}) -\delta_{i})m(X_{i})/\pi(X_{i})- \mu\right\}$ and $\pi(x) = P(\delta=1\mid X=x)$. The term $R_{N} = O_{P}\left(1 /(Nh^{d}) + h^{q}\right)$ where $q$ is a quantity that indicates the smoothness of some specific functions determined by the underlying data generation process. Convergence rate of the term $\psi_{N}$ is dimension free while that of the term $R_{N}$ does depend on $d$. To discuss the impact of dimension, we focus on $R_{N}$. With the optimal choice of the bandwidth that minimizes the convergence rate of $R_{N}$, we have 
$R_{N} = O_{P}\left(a_{N}\right)$ where $a_{N} = (1 / N)^{q/(d+q)}$. Let $\epsilon > 0$ be the required accuracy. For any given sample size $N$ and $q$, we investigate what values can be taken for $d$ such that $a_{N} \leq \epsilon$.
By straightforward calculations, $a_{N} \leq \epsilon$ is equivalent to 
\begin{equation}\label{req: dim}
	d \leq q\left(\frac{\log N}{\log \epsilon^{-1}} - 1\right),
\end{equation} which establishes the upper bound of $d$ given $q$, $\epsilon$ and $N$. 
For example, if $q = 10$, $\epsilon = 0.01$ and $N = 200$, \eqref{req: dim} restricts $d$ to take $1$ only. 
This implies that we can control the magnitude of $R_{N}$ to the given accuracy $\epsilon = 0.01$ only in the case of one-dimensional covariate when the sample size is $200$.   
However, if $q=10$, $\epsilon=0.01$ as before and the sample size $N = 200000$, the inequality $\eqref{req: dim}$ allows $d \leq 16$. 
This demonstrates the crucial role of sample size in applying the nonparametric regression imputation method to problems with high dimensional covariate vectors. The larger the sample size is, the larger $d$ is allowed to take.

However, when $N$ is extremely large, the calculation of $\hat{\mu}_{\mathbb{K}}$ is problematic. The computing time of the estimation process in \eqref{est: no truct} is $\Theta(N^2)$, which is extremely long when $N$ is large. In this paper, we say the computing time is $\Theta(a_{N})$ for some positive sequence $a_{N}$ if the computing time belongs to $[C^{-1}a_{N}, Ca_{N}]$ for some constant $C > 1$. A similar notation is used for the communication complexity. 

Besides the computational issue, it may be infeasible to keep the whole data set in memory or even store all the data on a single computer when $N$ is extremely large. Throughout this paper, we assume the samples are evenly distributed on $L$ machines. Here we assume $N$ is divisible by $L$ for simplicity and denote $n = N/L$.
Suppose $\{(\delta_{i}, Y_{i}, X_{i}): i=n(l-1) + 1,\dots,nl\}$ is stored on the $l$-th machine for $l=1,\dots,L$. 

To reduce the computing time and accommodate the distributed data set, we propose the KDI method to estimate $\mu$. 
The algorithm for 
the kernel-based distributed imputation method is then presented as follows.
\begin{breakablealgorithm}
	\caption{Algorithm for the KDI method}\label{al:kernel}
	\begin{algorithmic}[1]
		\For{$l = 1,\dots,L$}
		\State Calculate 
		\[\hat{\mu}^{(l)}_{\mathbb{K}}
		=  n^{-1}\sum^{nl}_{i=n(l - 1)+1}\{\delta_iY_i+(1-\delta_i)\hat{m}^{(l)}_{\mathbb{K}}(X_i)\}\]
		on the $l$-th machine for $l=1,\dots,L$ in parallel, where
		\[\hat{m}^{(l)}_{\mathbb{K}}(x)=\frac{\sum^{nl}_{i= n(l-1) + 1}K_{h}(X_i-x)\delta_iY_i}{\sum^{nl}_{i=n(l-1)+1}K_{h}(X_i-x)\delta_i};\]
		\State Transmit $\hat{\mu}^{(l)}_{\mathbb{K}}$ to the first machine;
		
		\EndFor
		\State Calculate $\tilde{\mu}_{\mathbb{K}} = L^{-1}\sum_{l=1}^{L}\hat{\mu}^{(l)}_{\mathbb{K}}$ on the first machine;
		\State
		\Return $\tilde{\mu}_{\mathbb{K}}$.
	\end{algorithmic}
	\label{algo: kernel}
\end{breakablealgorithm}

The computing time of the KDI estimation is $\Theta(n^{2}) = \Theta(N^{2} / L^{2})$, which is significantly faster than the conventional kernel regression imputation estimation in \eqref{est: no truct}. Moreover, the computation of the estimator $\tilde{\mu}_{\mathbb{K}}$ has a low communication complexity of order $\Theta(L)$ because we only need to transmit a real number from each local machine to the first machine in the KDI method.

\subsection{Theoretical Properties}\label{subsec: kernel theory}
Next, we establish the asymptotic properties of the KDI estimator. For convenience, we next let $C$ be a generic positive constant that may be different in different places. For any positive sequences $a_{N}$ and $b_{N}$, let $a_{N} \asymp b_{N}$ denote $C^{-1}b_{N} \leq a_{N} \leq Cb_{N} $ for some $C > 1$.
The establishment of the asymptotic normality and efficiency of $\tilde{\mu}_{\mathbb{K}}$ is nontrivial since we allow the number of machines $L$ to diverge as $N\to \infty$. A careful analysis of the error term is needed to make sure that the summation of many negligible terms is still negligible. 
Let $f(\cdot)$ be the probability density of $X$ and $\sigma^{2}(x)$ the variance of $Y$ conditional on $X = x$. We need the following conditions to establish asymptotic results.
\begin{itemize}
	\item [(C.1)] $\pi(x)$, $f(x)$ and $m(x)$ have bounded partial derivatives up to order $q>0$.
	\item [(C.2)] $\inf_x\pi(x)>0$.
	\item [(C.3)] $\inf_x f(x)>0$. 
	\item [(C.4)] $\sup_x\sigma^{2}(x) < \infty$.
	\item [(C.5)] $K(\cdot)$ is a Lipschitz continuous kernel function of order no smaller than $q$ with compact support.
	\item [(C.6)] $h^{q}\log N$ and $L(\log N)^{2}/(Nh^{d})$ are bounded  as $N\to \infty$.
	\item [(C.7)] $q > d$, $\sqrt{N}h^{q} \to 0$ and $L/(\sqrt{N}h^{d}) \to 0$ as $N\to \infty$.
\end{itemize}
Conditions (C.1),(C.3),(C.4) and (C.5) are all standard conditions in the literature of nonparametric regression \citep{hansen2008uniform,Li2011dim,Li2017mean,Ma2019efficient}. Conditions (C.1), (C.3) and (C.4) are required to establish the convergence rate of the kernel estimators for $m(\cdot)$ and $\pi(\cdot)$. Condition (C.1) is a general smoothness condition. The quantity $q$ determines how restrictive (C.1) is and a detailed discussion for $q$ will be conducted later. The infimum and supremum in Conditions (C.2), (C.3), and (C.4) are taken over the support of $X$. Condition (C.2) requires that the response of units with any covariate values can be observed with positive probability. It is crucial for the identification and $\sqrt{N}$-consistent estimation of $\mu$ \citep{Khan2010irregular}. Condition (C.2) can be satisfied in many problems and is widely adopted in the literature of nonparametric missing data methods \citep{Wang2002ELMissing,Hirano2003efficient,Hu2012core, chan2016globally}. Condition (C.3) can be easily satisfied when the covariate has bounded support. 
By truncating the denominators in the estimator, we may relax Condition (C.3) \citep{Cheng1994nonparametric, Wang2002ELMissing}. However, the truncation procedure introduces some extra tuning parameters. For this reason, we do not adopt the truncation strategy. 
Condition (C.4) requires the conditional variance of the response to be bounded, which is a mild condition on the data distribution. Regularity conditions on the kernel imposed in (C.5) are for the convenience of establishing probability and moment bounds and can be satisfied by many kernels such as the Epanechnikov kernel and the tricube kernel. Condition (C.6) and (C.7) are restrictions imposed on the bandwidth $h$ and the number of machines $L$. Condition (C.6) is used to establish the asymptotic expansion of $\tilde{\mu}_{\mathbb{K}}$. Condition (C.7) further imposes restrictions on $L$ and $h$ to ensure the asymptotic normality. Sufficient conditions for Conditions (C.6) and (C.7) will be discussed later.
\begin{theorem}\label{thm:kernel}
	Under conditions (C.1)-(C.6), if $N/L\to \infty$, we have 
	\begin{equation}\label{res: kernel rate}
		\tilde{\mu}_{\mathbb{K}}-\mu=\psi_{N} + O_{P}\left(\frac{L}{Nh^{d}}+ h^{q}\right),
	\end{equation}
	where $\psi_{N}=N^{-1}\sum_{i=1}^{N}\left\{\delta_{i}Y_{i}/\pi(X_{i}) + (\pi(X_{i}) -\delta_{i})m(X_{i})/\pi(X_{i})- \mu\right\}$. If we further assume (C.7), we have
	\begin{equation}\label{res: KDI AN}
		\sqrt{N}(\tilde{\mu}_{\mathbb{K}}-\mu)\stackrel{d}{\rightarrow}N(0,\var[\psi]),
	\end{equation}
	where
	\begin{equation}
		\psi=\frac{\delta Y}{\pi(X)} + \frac{\pi(X) -\delta}{\pi(X)}m(X).
	\end{equation}
\end{theorem}

Equality \eqref{res: kernel rate} establishes the asymptotic expansion of $\tilde{\mu}_{\mathbb{K}}$. 
The first term $\psi_{N}$ is $\sqrt N$-asymptotically normal and has the same asymptotic variance as $\hat{\mu}_{\mathbb{K}}$. Hence we have $\psi_{N} = O_{P}(1/\sqrt{N})$ in terms of the convergence rate, which is irrelevant to $h$ and $L$. 
Some calculations can show that the convergence rate of the second term at the right-hand-side of \eqref{res: kernel rate} is minimized if 
\begin{equation}\label{eq: opt bandwidth}
	h \asymp \left(\frac{L}{N}\right)^{\frac{1}{d+q}}.
\end{equation} 
With the bandwidth rate given in \eqref{eq: opt bandwidth}, the convergence rate of $\tilde{\mu}_{\mathbb{K}}$ equals to 
\[O_{P}\left(\frac{1}{\sqrt{N}} + \left(\frac{L}{N}\right)^{\frac{q}{d+q}}\right).\]
It can be seen that the number of machines $L$ contributes to this convergence rate. The minimax rate $1/\sqrt{N}$ provided in Appendix A can be achieved if $q\geq d$ and $L$ is not too large, specifically, if $L\leq C N^{1/2-d/(2q)}$ for some constant $C$.

The asymptotic normality result \eqref{res: KDI AN} follows from Equation \eqref{res: kernel rate} provided the second term at right-hand-side of \eqref{res: kernel rate} is of order $o_{P}(1/\sqrt{N})$. It is straightforward to verify $\var[\psi]$ equals to the semiparametric efficiency bound provided in Appendix A.
When $L$ is too large, the second term at the right-hand-side of \eqref{res: kernel rate} may not get the required convergence rate. 
Thus some extra restrictions on $L$ are imposed in Condition (C.7) to ensure the asymptotic efficiency of $\tilde{\mu}_{\mathbb{K}}$.  The bandwidth satisfying (C.7) exists as long as $L$ satisfies $L / N^{1/2-d/(2q)} \to 0$. If $q$ in Condition (C.1) is sufficiently large, the requirement on $L$ is close to $L / \sqrt{N} \to 0$, which are also required by many existing divide and conquer methods to ensure their theoretical properties \citep{Zhang2013dist,Lee2017dist,battey2018distributed}.

\section{Sieve-Based Multi--round Method}\label{sec: sieve}
\subsection{Methodology}
In the previous section, we discuss the KDI method. 
The method has good asymptotic properties with very low communication cost. However,
according to the discussion in the last section, the number of machines cannot be too large to ensure the asymptotic efficiency of $\tilde{\mu}_{\mathbb{K}}$.
This may restrict the application of this method and limit the role of the KDI method in solving storage and computing problems with large-scale data. In fact, when we average over multiple nonlinear estimators to obtain the aggregated estimator, strict restrictions on the machine number are almost inevitable to ensure the $\sqrt{N}$-consistency \citep{Zhang2013dist,Lee2017dist,battey2018distributed,jordan2019communication}.
To relax the restriction on the machine number, we discard the one-shot averaging strategy and propose a multi-round SDI method in this section. Next, we propose the multi-round method based on the sieve method.

Sieve method is a widely used nonparametric method and has been used to impute the missing counterfactual value in the causal inference literature \citep{hahn1998role,imbens2005mean}. 
Nevertheless, the investigation on the application of the sieve method to data stored in a distributed manner is still rare in the literature.

Let $\{\mathcal{V}_{k}\}_{k=1}^{\infty}$ be a nested sequence of finite-dimensional function classes such that 
\[\mathcal{V}^2(\mathbf{R}^d) = \{f(x): E[f^2(X)] < \infty\} = \bigcup_{k=1}^{\infty} \mathcal{V}_{k}.\]
For any positive integer $K$, let $\{v_1(x), \dots, v_{K}(x)\}$ be a set of bases of $\mathcal{V}_{K}$ with $v_{1}(x)=1$ and let $V_{K}(x) = (v_1(x),\dots,v_{K}(x))^{\T}$. 
The linear combination of $\{v_1(x), \dots, v_{K}(x)\}$ can approximate any function in $\mathcal{V}^{2}(\mathbf{R}^d)$ if $K$ is large. To circumvent the theoretical difficulty brought by the approximation error, we let $K$ increase to infinity as $N\to \infty$. We use $\hat{m}_{\mathbb{S}}(x) = V_{K}(x)^{\T}\hat{\beta}$ to estimate $m(x)$, where
\begin{equation}\label{prob:sieve}
	\hat{\beta} = \arg\min_{\beta}\frac{1}{N}\sum_{i=1}^{N}\delta_{i}(Y_{i} - V_{K}(X_{i})^{\T}\beta)^2.
\end{equation}
Then $\mu$ can be estimated by
\[\hat{\mu}_{\mathbb{S}} = \frac{1}{N}\sum^N_{i=1}\{\delta_iY_i+(1-\delta_i)\hat{m}_{\mathbb{S}}(X_i)\}.\]
$\hat{\beta}$ has the explicit form 
\[\hat{\beta} = \hat{\Sigma}^{-1} \hat{\Gamma},\]
where $\hat{\Gamma} = \sum_{i=1}^{N}\delta_{i}Y_{i}V_{K}(X_{i})/N$ and $\hat{\Sigma} = \sum_{i=1}^{N}\delta_{i}V_{K}(X_{i})V_{K}(X_{i})^{\T} / N$. The computing time of $\hat{\mu}_{\mathbb{S}}$ is $\Theta(NK^{2} + K^{3})$ when processing all data on a single machine. When data are stored in a distributed manner, we can first calculate 
\[\sum_{i=n(l - 1) + 1}^{nl}\delta_{i}Y_{i}V_{K}(X_{i}),\]
on the $l$-th machine and transmit the resulting $K$-dimensional vector to the first machine, where $n = N/L$ is the number of observations in each machine. 
Then $\hat{\Gamma}$ can be gotten by summing up these vectors and dividing the summations by $N$. However, to calculate $\hat{\Sigma}$, we need to transmit $K\times K$ matrices and the transmission of $K\times K$ matrices has a high communication cost if $K$ is large. Most of the time, $K$ is taken to be of the polynomial order of $N$ to ensure the approximation accuracy of the basis functions. Thus $K$ is usually large when $N$ is large. Motivated by the communication efficient algorithms in the parametric model literature \citep{shamir2014communication,jordan2019communication,fan2021communication}, we propose a multi-round algorithm with a low communication cost to approximate $\hat{\beta}$.
For arbitrary $\beta^{\dag}$, we have by some simple algebra
\begin{align*}
	&\frac{1}{N}\sum_{i=1}^{N}\delta_{i}(Y_{i} - V_{K}(X_{i})^{\T}\beta)^2\\ 
	&= (\beta - \beta^{\dag})^{\T}\hat{\Sigma}(\beta - \beta^{\dag}) - 2(\beta - \beta^{\dag})^{\T}(\hat{\Gamma} - \hat{\Sigma}\beta^{\dag}) + \beta^{\dag\T}\hat{\Sigma}\beta^{\dag} -2\beta^{\dag\T}\hat{\Gamma} + \frac{1}{N}\sum_{i=1}^{N}\delta_{i}Y_{i}^2.
\end{align*}
Thus
\[
\hat{\beta} = \arg\min_{\beta} \hat{D}(\beta,\beta^{\dag}),
\]
where 
\[\hat{D}(\beta,\beta^{\dag}) = (\beta - \beta^{\dag})^{\T}\hat{\Sigma}(\beta - \beta^{\dag}) - 2(\beta - \beta^{\dag})^{\T}(\hat{\Gamma} - \hat{\Sigma}\beta^{\dag}).\]

Let $\tilde{\Sigma} = n^{-1}\sum_{i=1}^{n}\delta_{i}V_{K}(X_{i})V_{K}(X_{i})^{\T}$. 
Since samples are independent and identically distributed across different machines, $\tilde{\Sigma}$ is expected to be close to $\hat{\Sigma}$. 
We hence use 
\[
\tilde{D}(\beta,\beta^{\dag}) = (\beta - \beta^{\dag})^{\T}\tilde{\Sigma}(\beta - \beta^{\dag}) - 2(\beta - \beta^{\dag})^{\T}(\hat{\Gamma} - \hat{\Sigma}\beta^{\dag})
\]
to approximate $\hat{D}(\beta,\beta^{\dag})$. Note that $\tilde{\Sigma}$ can be calculated with the data in the first machine solely. Moreover, we can see that $\hat{\Sigma}\beta^{\dag}$ is a vector and the calculation of $\hat{\Sigma}\beta^{\dag}$ only needs transmission of $K$-dimensional vectors by noting  $\hat{\Sigma}\beta^{\dag}=\sum_{l=1}^{L} \hat {\Sigma}_{l}\beta^{\dag}$ where $\hat{\Sigma}_{l}=\sum_{i=n(l-1) + 1}^{nl}\delta_{i}V_{K}(X_{i})V_{K}(X_{i})^{\T}$.
Thus to calculate $\tilde{D}(\beta, \beta^{\dag})$, we need not to transmit $K\times K$ matrices. Let $\|\cdot\|$ be the Euclid/spectral norm of a vector/matrix. 
Then $|\hat{D}(\beta, \beta^{\dag}) - \tilde{D}(\beta, \beta^{\dag})| \leq \|\tilde{\Sigma} - \hat{\Sigma}\|\|\beta - \beta^{\dag}\|^2$. Thus if $\beta$ is close to $\beta^{\dag}$, the approximation performs well even if $\tilde{\Sigma}$ is not that close to $\hat{\Sigma}$. 
In order to get the desired approximation rate, we start at an initial value $\beta_{0}$, and minimize $\tilde{D}(\beta, \beta_{0})$ over a small ball around $\beta_{0}$. 
By the equivalence of the constrained optimization and the penalized optimization, we minimize the following loss function
\begin{equation}\label{eq:iter}
	\tilde{D}(\beta,\beta_{0}) + \alpha\|\beta - \beta_{0}\|^2,
\end{equation}
where $\alpha$ is a tuning parameter. Let $\beta_{1}$ be the minimum point of \eqref{eq:iter}, $\beta_{2}$ the minimum point of $\tilde{D}(\beta,\beta_{1}) + \alpha\|\beta - \beta_{1}\|^2$ and let $\beta_{t}$ be defined in the same way for $t > 2$.
It is not hard to show that $\beta_{t}$ has the recursive relation
\begin{align*}
	\beta_{t}& = \beta_{t-1} + (\tilde{\Sigma} + \alpha I)^{-1}(\hat{\Gamma} - \hat{\Sigma}\beta_{t-1})\\
	& =  \beta_{t-1} + (\tilde{\Sigma} + \alpha I)^{-1}\left(
	\hat{\Gamma} - \frac{1}{N}\sum_{i=1}^{N}\delta_{i}V_{K}(X_{i})V_{K}(X_{i})^{\T}\beta_{t-1}
	\right).
\end{align*}
We use $\tilde{\beta} = \beta_{T}$ as the approximation of $\hat{\beta}$ with $T$ being a properly chosen integer and propose the following estimator for $\mu$ 
\[
\tilde{\mu}_{\mathbb{S}} = \frac{1}{N}\sum^N_{i=1}\{\delta_iY_i+(1-\delta_i)\tilde{m}_{\mathbb{S}}(X_i)\},
\]
where $\tilde{m}_{\mathbb{S}}(x) = V_{K}(x)^{\T}\tilde{\beta}$.
These procedures are summarized in Algorithm \ref{al:sieve}.
\begin{breakablealgorithm}
	\caption{Algorithm for the SDI method}\label{al:sieve}
	\begin{algorithmic}[1]
		\State Initialize $\beta_{0} = (0,\dots,0)^{\T}$;
		\State Calculate 
		\[z_{0,l} = \sum_{i=n(l-1) + 1}^{nl}\delta_{i}Y_{i}V_{K}(X_{i})\]
		on the $l$-th machine for $l=1,\dots,L$ in parallel and transmit the result to the first machine;
		\State Calculate 
		\[\Lambda_{l} = \sum_{i=n(l-1) + 1}^{nl}\delta_{i}V_{K}(X_{i})V_{K}(X_{i})^{\T}\]
		on the $l$-th machine for $l=1,\dots,L$ in parallel;
		\State Calculate $\Psi = (n^{-1}\Lambda_{1} + \alpha I)^{-1}$ and $\hat{\Gamma} = N^{-1}\sum_{l=1}^{L}z_{0,l}$ on the first machine;
		\For{$t = 1,\dots,T$}
		\State Calculate 
		\[z_{t,l} = \Lambda_{l}\beta_{t-1}\]
		on the $l$-th machine for $l=1,\dots,L$ in parallel and transmit the result to the first machine;
		\State	Update $\beta_{t} = \beta_{t-1} + \Psi(\hat{\Gamma} - N^{-1}\sum_{l=1}^{L}z_{t,l})$ and transmit $\beta_{t}$ to each local machine;
		\EndFor
		\State Calculate 
		\[\hat{\mu}_{\mathbb{S}}^{(l)} = \frac{1}{n} \sum_{i=n(l-1) + 1}^{nl}\{\delta_iY_i+(1-\delta_i)V_{K}(X_{i})^{\T}\tilde{\beta}\}\]
		on the $l$-th machine for $l=1,\dots,L$ in parallel and transmit the result to the first machine;
		\State Calculate $\tilde{\mu}_{\mathbb{S}} = L^{-1} \sum_{l=1}^{L}\hat{\mu}_{\mathbb{S}}^{(l)}$ on the first machine;
		\State
		\Return $\tilde{\mu}_{\mathbb{S}}$.
	\end{algorithmic}
\end{breakablealgorithm}
The proposed SDI method has a computing time of order $\Theta(NK^{2} / L  + K^{3} + TNK/L)$ and a communication complexity of order $\Theta(TLK)$. In the next section, we show that taking  $T\asymp \log N$ is often sufficient to ensure the statistical accuracy.  In this case, the computing time of the SDI estimator $\tilde{\mu}_{\mathbb{S}}$ is much shorter than that of $\hat{\mu}_{\mathbb{S}}$ if $L$ is large. Moreover, the communication complexity of the SDI method is also relatively low as we do not require to transmit $K\times K$ matrices.

\subsection{Theoretical Properties}\label{subsec: sieve theory}
Next, we establish the theoretical properties of the SDI estimator and show that the SDI method can accommodate more machines compared to the KDI method while keeping the theoretical properties.

To show the large sample properties, we need to introduce some conditions.
We first define $\zeta_{K} = \sup_{x}\|V_{K}(x)\|$. The quantity $\zeta_{K}$ is important in theoretical development and can be determined by the basis functions. 
We have $\zeta_{K} \leq C \sqrt{K}$ if the basis functions are tensor products of univariate B-spline, Chebyshev polynomial, trigonometric polynomial or wavelet bases. Also we have $\zeta_{K} \leq C\sqrt{K}$ for tensor products of power series if the support of $X$ is contained in $[-1,1]^{d}$. 
See \cite{newey1994asymptotic} and \cite{chen2007large} for more results on this quantity. 
For any symmetric matrix $A$, denote its largest and smallest eigenvalues by $\sigma_{\rm max}(A)$ and $\sigma_{\rm min}(A)$, respectively. Let $\Sigma = \E [\delta V_{K}(X)V_{K}(X)^{\T}]$. 
Then we are ready to introduce the following technical conditions.

\begin{itemize}
	\item [(C.8)] There are some universal constants $C_{\rm L}$ and $C_{\rm H}$ such that
	$C_{\rm L} \leq \sigma_{\rm min}(\Sigma) \leq \sigma_{\rm max}(\Sigma) \leq C_{\rm H}$.
	\item [(C.9)] (i) $E[(m(X))^{2}]\leq \infty$; (ii) there exists some constant $r>0$ such that for any $K$ there is some $\bar{\beta}$ and $\bar{\gamma}$ satisfying $\E[(m(X) - V_{K}(X)^{\T}\bar{\beta})^2] \leq C K^{-2r}$ and $\E[(\pi(X)^{-1} - V_{K}(X)^{\T}\bar{\gamma})^2]\leq C K^{-2r}$.
	\item[(C.10)] $\zeta_{K}^{2}\log K/N \to 0$.
\end{itemize}
Let $C_{N,K} = 1 - (1 - \min\{\alpha,1/\log K\})/(1 + 2\alpha C_{\rm L})$ where $\alpha$ is the tuning parameter for the penalty in \eqref{eq:iter} and assume without loss of generality that $\log K > 1$. 
\begin{itemize}
	\item [(C.11)] (i)$\zeta_{K}^{4}/N \to 0$, $NK^{-4r} \to 0$; 
	
	(ii)$T\log C_{N,K}^{-1} - 0.5\log N - \log \zeta_{K} \to \infty$,
	where $T$ is the number of iterations which can be seen as a tuning parameter.
\end{itemize}
Condition (C.8) is a widely-used condition in the literature of sieve method \citep{newey1994asymptotic,ai2003efficient,belloni2015some}. The lower bound on the minimum eigenvalue in (C.8) requires the basis functions not to be too co-linear, which is important for establishing the optimization property of the iterative algorithm in the SDI method. The upper bound on the maximum eigenvalue is needed to establish the convergence rate of $\tilde{\Sigma}$ and $\hat{\Sigma}$ to their population counterpart $\Sigma$. Condition (C.9)(i) is a mild moment condition required to establish the probability bound in the proof. Condition (C.9)(ii) imposes some restrictions on the approximation error. The constant $r$ appearing in (C.9) and (C.11) depends on the dimension of the covariate vector, the smoothness of $m(\cdot)$ and $\pi(\cdot)$, and the basis functions.
Under Conditions (C.1), (C.2) and (C.3), Condition (C.9)(ii) is satisfied with $r = q/d$ if the basis functions are tensor products of polynomial functions, B-splines, trigonometric polynomial functions or wavelet bases \citep{lorentz1986approximation, chen2007large}. See \cite{chen2007large} for more results on some common basis functions. Condition (C.11)(i) is related to the number of basis functions $K$. Conditions (C.10) and (C.11)(i) are related to the number of basis functions $K$. As discussed above, we have $\zeta_{K}\leq C\sqrt{K}$ for many commonly used basis functions. Then Condition (C.10) only requires $K\log K /N \to 0$, which is a mild constraint on $K$. If $\zeta_{K}\leq C\sqrt{K}$, a sufficient condition for Condition (C.11)(i) is $K = o(\sqrt{N})$ and $K^{-4r} = o(1/N)$. The condition $K^{-4r} = o(1/N)$ can be easily satisfied as long as $K$ is not too small and $r$ is moderately large. The requirement on upper bound of the number of basis functions $K = o(\sqrt{N})$ is considerably weaker than that in the existing literature on mean estimation problem using sieve method. For example, \cite{hahn1998role} requires $K = o(N^{1/7})$ and \cite{chan2016globally} requires $K = o(N^{1/11})$. Condition (C.11)(ii) require the number of iterations not to be too small. If $\alpha$ is bounded as $N$ increases, it is not hard to verify that there is some constant $C > 0$ such that  $\log C_{N,K}^{-1} \geq C$ for all $N$. Then Condition (C.11)(ii) can be satisfied if we take $T = C^{-1}\log N + 2C^{-1}\log \zeta_{K}$. This combined with Condition (C.11)(i) implies that it is sufficient to meet Condition (C.11)(ii) if we increase the number of iterations with the sample size at a logarithm rate. Then we are ready to state the theoretical result for the SDI method. For two positive sequences $a_{N}$ and $b_{N}$, we denote $a_{N}\asymp b_{N}$ if $C^{-1}b_{N} \leq a_{N} \leq C b_{N}$ for any $N$ and some constant $C > 1$.
\begin{theorem}\label{thm: sieve}
	If $\alpha \asymp \log^2K  \sqrt{(L\zeta_{K}^{2})/N}$, then under Conditions (C.2), (C.4), (C.8), (C.9) and (C.10), we have 
	\begin{equation}\label{res: sieve rate}
		\tilde{\mu}_{\mathbb{S}}-\mu=\psi_{N} + O_{P}\left(\frac{\zeta_{K}^{2}}{N} + \frac{1}{K^{2r}}\right) + O_{P}\left(\zeta_{K}C_{N,K}^{T}\right),
	\end{equation}
	where $\psi_{N}=N^{-1}\sum_{i=1}^{N}\left\{\delta_{i}Y_{i}/\pi(X_{i}) + (\pi(X_{i}) -\delta_{i})m(X_{i})/\pi(X_{i})- \mu\right\}$. If we further assume that (C.11) holds, we have
	\begin{equation}\label{res: SDI AN}
		\sqrt{N}(\tilde{\mu}_{\mathbb{S}}-\mu)\stackrel{d}{\rightarrow}N(0,\var[\psi]),
	\end{equation}
	where
	\begin{equation}
		\psi=\frac{\delta Y}{\pi(X)} + \frac{\pi(X) -\delta}{\pi(X)}m(X).
	\end{equation}
\end{theorem} 
The equality \eqref{res: sieve rate} presents the asymptotic expansion of $\tilde{\mu}_{\mathbb{S}}$. We have $\psi_{N} = O_{P}(1/\sqrt{N})$ as in Theorem \ref{thm:kernel}. The second term at the right-hand-side of \eqref{res: sieve rate} depends on $K$, $N$, $\zeta_{K}$ and is independent of $L$.
Under Conditions (C.1), (C.2) and (C.3), we have $\zeta_{K} \leq C\sqrt{K}$ and Condition (C.9)(ii) is satisfied with $r = q/d$ for the widely used basis functions listed before Theorem \ref{thm: sieve}. Then the second term at the right-hand-side of \eqref{res: sieve rate} is of order $O_{P}(K/N + 1/K^{2q/d})$. This rate is minimized if $K\asymp N^{d/(d+2q)}$ and the corresponding rate is $O_{P}(N^{-2q/(d+2q)})$ which is no slower than $O_{P}(1/\sqrt{N})$ if $q \geq d/2$. The third term at the right-hand-side of \eqref{res: sieve rate} depends on $K$ and $L$ through $\zeta_{K}$ and $C_{N,K}$ (Note that $C_{N,K}$ depends on $\alpha \asymp \log^2(K)  \sqrt{(L\zeta_{K}^{2})/N}$). However, for any $K$ and $L$, we always have $C_{N,K} < 1$ which implies the quantity $\zeta_{K}C_{N,K}^{T}$ can be made arbitrarily small by choosing $T$ sufficiently large. If we take
\begin{equation}\label{eq: number of iterations}
	T \geq \frac{0.5 \log N + \log \zeta_{K}}{\log C_{N,K}^{-1}},
\end{equation} 
then we have $\zeta_{K}C_{N,K}^{T} \leq 1/\sqrt{N}$ and hence the third term is of order $O_{P}(1/\sqrt{N})$. 

In summary, the minimax rate $1/\sqrt{N}$ provided in Appendix A can be achieved by $\tilde{\mu}_{\mathbb{S}}$ if (a) $\zeta_{K} \leq C\sqrt{K}$ and Condition (C.9)(ii) is satisfied with $r = q/d$; (b) $q \geq d/2$; (c) $K\asymp N^{d/(d+2q)}$ and the number of iterations $T$ satisfies \eqref{eq: number of iterations}. 

Because $\psi_{N}$ is $\sqrt N$-asymptotically normal with asymptotic variance achieving the semiparametric efficiency bound, $\tilde{\mu}_{\mathbb{S}}$ is asymptotically efficient when the second and third terms at the right-hand-side of \eqref{res: sieve rate} are both of order $o_{P}(1/\sqrt{N})$.
The second term is $o_{P}(1/\sqrt{N})$ under (C.11)(i) and the third term is $o_{P}(1/\sqrt{N})$ provided the restriction on $K$ and $T$ presented in Condition (C.11)(ii) is satisfied. If $\zeta_{K} \leq C\sqrt{K}$, Condition (C.9)(ii) is satisfied with $r = q/d$ and $K\asymp N^{d/(d+2q)}$, then Condition (C.11)(i) is satisfied as long as $q > d/2$. 
If the number of machine $L$ satisfies
\begin{equation}\label{eq: L restriction SDI}
	L \leq CN^{\frac{2q}{d+2q + 1}},
\end{equation} 
then $\alpha \asymp \log^2(K)  \sqrt{(L\zeta_{K}^{2})/N}$ converges to zero and hence is bounded.
According to the discussion before Theorem \ref{thm: sieve}, 
it suffices to take $T$ proportionally to $\log N$ to fulfill Condition (C.11)(ii) if $\zeta_{K} \leq C\sqrt{K}$, $K\asymp N^{d/(d+2q)}$ and $L$ satisfies \eqref{eq: L restriction SDI}. When $q$ is large, the restriction \eqref{eq: L restriction SDI} is close to $L \leq C N $ which is a mild condition on $L$.

Notice that both the KDI method and SDI method can achieve the minimax rate and are asymptotically efficient under certain conditions. However, the KDI method requires much stronger restrictions on $L$ compared to the SDI method. The reason may be 
that  $\hat{\mu}_{\mathbb{K}}^{(l)}$ for $l=1,2,\cdots,L$ are not linearly additive. The bias of the resulting estimator depends on the biases of $\hat{\mu}_{\mathbb{K}}^{(l)} $ for $=1,2,\cdots, L$ by their average. 
The bias of  $\hat{\mu}_{\mathbb{K}}^{(l)} $ and $\tilde{\mu}_{\mathbb{K}}$ is large if $L$ is too large because in this case the sample size on every machine is too small. 

In addition, if $\zeta_{K} \leq C\sqrt{K}$, Condition (C.9)(ii) is satisfied with $r = q/d$ and $K\asymp N^{d/(d+2q)}$, then the SDI can achieve the semiparametric efficiency bound under weaker smoothness conditions compared to the KDI method ($q > d/2$ for SDI v.s. $q > d$ for KDI). This is thanks to the moment condition satisfied by the limitation of $\tilde{\beta}$. Next, we explain this phenomenon in detail. Let $\beta^{*} = \mathop{\arg\min}_{\beta}\E[\delta(Y - V_{K}(X)^{\T}\beta)^{2}]$. The proof of Theorem \ref{thm: sieve} can show that $\tilde{\beta}$ converges to $\beta^{*}$ and $\tilde{\mu}_{\mathbb{S}} - \mu$ has the following expansion
\begin{equation}\label{eq: SDI bias expansion}
	\tilde{\mu}_{\mathbb{S}} - \mu = \psi_{N} + O_{P}\left(\frac{K}{N} + \frac{1}{K^{\frac{q}{d}}}\sqrt{\frac{K}{N}}\right) + E[V_{K}(X)^{\T}\beta^{*} - m(X)]
\end{equation}
provided $\zeta_{K} \leq C\sqrt{K}$ and Condition (C.9)(ii) is satisfied with $r = q/d$. Here $E[V_{K}(X)^{\T}\beta^{*} - m(X)]$ is the bias term caused by using basis functions to approximate the true conditional mean function. A natural idea to bound the approximation error is to invoke Jensen's inequality, i.e.,
\[
E[V_{K}(X)^{\T}\beta^{*} - m(X)] \leq \sqrt{E[(V_{K}(X)^{\T}\beta^{*} - m(X))^{2}]}.
\]
Note that, under Condition (C.2), we have 
\begin{equation}\label{eq: approx bound}
	\begin{aligned}
		E[(V_{K}(X)^{\T}\beta^{*} - m(X))^{2}] &= E\left[\frac{\delta}{\pi(X)} (V_{K}(X)^{\T}\beta^{*} - m(X))^{2}\right]\\
		& \leq C E[\delta (V_{K}(X)^{\T}\beta^{*} - m(X))^{2}]\\ &\leq CE[\delta (V_{K}(X)^{\T}\bar{\beta} - m(X))^{2}]\\ &\leq CE[ (V_{K}(X)^{\T}\bar{\beta} - m(X))^{2}]
	\end{aligned}
\end{equation}
for the $\bar{\beta}$ in Condition (C.9)(ii). This implies $E[V_{K}(X)^{\T}\beta^{*} - m(X)] \leq C K^{-q/d}$
if Condition (C.2) holds and Condition (C.9)(ii) is satisfied with $r = q/d$.
Combining this bound with \eqref{eq: SDI bias expansion}, it will conclude that $q > d$ is required to assure the asymptotic efficiency. This requirement is the same as that of the KDI method. However, we find that $\beta^{*}$ satisfies the moment condition $E[\delta V_{K}(X)(Y - V_{K}(X)^{\T}\beta^{*})] = 0$ which implies 
\begin{equation}\label{eq: SDI moment cond}
	E[\delta V_{K}(X)(m(X) - V_{K}(X)^{\T}\beta^{*})] = 0.
\end{equation}
Next, we show how this moment condition can sharpen the bound on the approximation error and weaken the smoothness condition required for asymptotic efficiency. According to \eqref{eq: SDI moment cond}, we have
\[
E[\delta\gamma^{\T}V_{K}(X)(m(X) - V_{K}(X)^{\T}\beta^{*})] = 0
\]
for any $\gamma$. Thus we have
\begin{equation}\label{eq: SDI bias bound}
	\begin{aligned}
		E[V_{K}(X)^{\T}\beta^{*} - m(X)] & =    E\left[\frac{\delta}{\pi(X)}(V_{K}(X)^{\T}\beta^{*} - m(X))\right]  \\
		& = E\left[\delta\left(\frac{1}{\pi(X)} - V_{K}(X)^{\T}\gamma\right)(V_{K}(X)^{\T}\beta^{*} - m(X))\right]\\
		&\quad - E[\delta\gamma^{\T}V_{K}(X)(m(X) - V_{K}(X)^{\T}\beta^{*})] \\
		& = E\left[\delta\left(\frac{1}{\pi(X)} - V_{K}(X)^{\T}\gamma\right)(V_{K}(X)^{\T}\beta^{*} - m(X))\right] \\
		& \leq \sqrt{E\left[\delta\left(\frac{1}{\pi(X)} - V_{K}(X)^{\T}\gamma\right)^{2}\right]}\sqrt{E\left[(V_{K}(X)^{\T}\beta^{*} - m(X))^{2}\right]}\\
		& \leq C\sqrt{E\left[\left(\frac{1}{\pi(X)} - V_{K}(X)^{\T}\gamma\right)^{2}\right]}\sqrt{E\left[(V_{K}(X)^{\T}\bar{\beta} - m(X))^{2}\right]}
	\end{aligned}    
\end{equation}
for any $\gamma$ under Condition (C.2) according to \eqref{eq: approx bound}. By the arbitrariness of $\gamma$, we have  
\[ 
\begin{aligned}
	E[V_{K}(X)^{\T}\beta^{*} - m(X)] & \leq \sqrt{E\left[\left(\frac{\delta}{\pi(X)} - V_{K}(X)^{\T}\bar{\gamma}\right)^{2}\right]}\sqrt{E\left[(V_{K}(X)^{\T}\bar{\beta} - m(X))^{2}\right]} \\
	&\leq CK^{-\frac{2q}{d}}
\end{aligned}
\]
if Conditions (C.2) holds and Condition (C.9)(ii) is satisfied with $r = q/d$. Combining this with \eqref{eq: SDI bias expansion}, we conclude that $q > d/2$ is sufficient to assure the asymptotic efficiency of the SDI method. 

Inequality \eqref{eq: SDI bias bound} also implies that the bias caused by the approximation error is small if one of $m(x)$ and $1/\pi(x)$ can be approximated well by the basis functions used. The phenomenon can benefit the finite sample bias of $\tilde{\mu}_{\mathbb{S}}$ and may provide some insights for practitioners on which basis functions should be used. See Section \ref{sec: SDI approximation error} for more simulation evidence.

\section{A Distributed Tuning Parameter Selection Procedure}\label{sec: tuning selection}
For the implementation of the proposed methods, we recommend taking the bandwidth $h=c(L/N)^{1/(2d+1)}$ in the KDI method and the number of basis functions $K=\lceil cN^{d/(2d+1)}\rceil$ in the SDI method. Here $\lceil\cdot\rceil$ means round up. We choose these rates because the conditions of Theorems \ref{thm:kernel} and \ref{thm: sieve} can be satisfied under mild conditions on $L$ and the smoothness of $m(\cdot)$, $\pi(\cdot)$ if such rates are used according to the discussion behind the aforementioned theorems.

The asymptotic normality results of the estimators $\tilde{\mu}_{\mathbb{K}}$ and $\tilde{\mu}_{\mathbb{S}}$ imply that the tuning parameters $h$ and $K$ affect only the second order term in the asymptotic expansion of the estimators. Therefore, the value of tuning parameters has no effect on the convergence rate and asymptotic distribution of the estimators as long as conditions such as (C.7) and (C.11) are satisfied. 
This phenomenon is also discussed in \cite{Cheng1994nonparametric,Wang2002ELMissing} and other papers. 
Thus, the choice of the constant $c$ is not that crucial from the asymptotic point of view.
However, in practice, the performance of the proposed estimators may be sensitive to the constant $c$ when the sample size is not sufficiently large and the dimension of covariates is relatively high, as can be seen in the simulation in Section \ref{subsec: sim tuning}.
Thus we next provide a selection procedure for $h$ and $K$. In the following, we use $c$ to denote a generic tuning parameter whose value may be different in different places and let $\mathcal{C}$ be the candidate set of $c$.
Cross validation (CV) is one of the most commonly used methods to select tuning parameters. But it suffers from several problems when directly applying the CV method to the problem considered here.

The first problem is caused by missing data. When applying the CV method, we first divide the data set into two parts which is called the ``training data" and the ``test data", respectively. Let $I_{tr}$ and $I_{te}$ be the index set of training data and test data, respectively. For each $c\in\mathcal{C}$, suppose $\hat{m}_{tr,c}(x)$ is the imputation function based on the training data using the tuning parameter $c$. For the KDI method,  $\hat{m}_{tr,c}(x)$ is obtained by the Nadaraya-Watson kernel regression similar to $\hat{m}_{\mathbb{K}}(x)$. For the SDI method, $\hat{m}_{tr,c}(x)$ is obtained by the sieve least squares regression similar to $\hat{m}_{\mathbb{S}}(x)$. Then we calculate the criterion based on the squared prediction error
\begin{equation}\label{eq: classical CV}
	|I_{te}|^{-1}\sum_{i\in I_{te}}(Y_{i}-\hat{m}_{tr,c}(X_{i}))^{2},
\end{equation} 
where we use $|\cdot|$ to denote the cardinal of a set. The classical CV method repeats the data splitting procedure for $\kappa$ times and chooses the $c\in \mathcal{C}$ that minimizes \eqref{eq: classical CV} on average.

For missing responses problems, the values of $Y$ corresponding to $\delta=0$ are missing. Hence the criterion function \eqref{eq: classical CV} is unavailable in practice. The most direct idea is to work with the observed version of \eqref{eq: classical CV}
\begin{equation}\label{eq: classical CV 1}
	|I_{te,1}|^{-1}\sum_{i\in I_{te,1}}(Y_{i}-\hat{m}_{tr,c}(X_{i}))^{2},
\end{equation} 
where $I_{te,1} = \{i: i \in I_{te}, \delta_{i} = 1\}$ is the index set of complete cases in the test data.
Nevertheless, the accuracy of missing data imputation is what we really concern. Similar to \eqref{eq: classical CV}, the imputation accuracy can be characterized by
\begin{equation}\label{eq: classical CV 0}
	|I_{te,0}|^{-1}\sum_{i\in I_{te,0}}(Y_{i}-\hat{m}_{tr,c}(X_{i}))^{2},
\end{equation} 
where $I_{te,0} = \{i: i \in I_{te}, \delta_{i} = 0\}$.
Clearly, the criterion \eqref{eq: classical CV 0} is unobservable because $Y_{i}$ is unavailable for $i\in I_{te,0}$. To resolve this problem, we integrate out the unobserved responses by taking conditional expectation and try to approximate \eqref{eq: classical CV 0} with observed data by approximating its conditional expectation. 
Conditional on the training data and $\{(\delta_{i}, X_{i})\}_{i \in I_{te}}$, the expectation of \eqref{eq: classical CV 0} is 
\begin{equation}\label{eq: CV expe 0}
	\begin{aligned}
		&|I_{te,0}|^{-1}\sum_{i\in I_{te,0}}\left[(m(X_{i})-\hat{m}_{tr,c}(X_{i}))^{2} + \sigma^{2}(X_{i})\right] \\
		&=\int \left[(m(x)-\hat{m}_{tr,c}(x))^{2} + \sigma^{2}(x)\right]d\hat{F}_{te,0}(x).
	\end{aligned}
\end{equation}
where $\hat{F}_{te,0}(x)=|I_{te,0}|^{-1}\sum_{i\in I_{te,0}} 1\{X_{i}\leq x\}$ is the empirical covariate distribution function for $i \in I_{te,0}$.
This conditional expectation usually differs from that of \eqref{eq: classical CV 1}, which is 
\begin{equation*}
	\begin{aligned}
		\int \left[(m(x)-\hat{m}_{tr,c}(x))^{2} + \sigma^{2}(x)\right]d\hat{F}_{te,1}(x)
	\end{aligned}
\end{equation*}
with $\hat{F}_{te,1}(x)= |I_{te,1}|^{-1}\sum_{i\in I_{te,1}} 1\{X_{i}\leq x\}$, because $\hat{F}_{te,0} \neq \hat{F}_{te,1}$ in general. 
To get a better approximation, we consider a weighted version of \eqref{eq: classical CV 1} 
\begin{equation}\label{eq: classical CV weight}
	|I_{te,1}|^{-1}\sum_{i\in I_{te,1}}w_{i}(Y_{i}-\hat{m}_{tr,c}(X_{i}))^{2},
\end{equation}  
where $w_{i}$ is the weight for the $i$-th observation with $i\in I_{te,1}$. 
Conditional on the training data and $\{(\delta_{i}, X_{i})\}_{i \in I_{te}}$, the expectation of \eqref{eq: classical CV weight} is 
\begin{equation}\label{eq: CV weight expe}
	\begin{aligned}
		\int \left[(m(x)-\hat{m}_{tr,c}(x))^{2} + \sigma^{2}(x)\right]d\hat{F}_{te,1}^{w}(x)
	\end{aligned}
\end{equation}
where $\hat{F}_{te,1}^{w}(x)= |I_{te,1}|^{-1}\sum_{i\in I_{te,1}} w_{i}1\{X_{i}\leq x\}$ is the weighted empirical covariate distribution function for $i \in I_{te,1}$. \eqref{eq: CV weight expe} is close to \eqref{eq: CV expe 0} if $\hat{F}_{te,1}^{w}$ is close to $\hat{F}_{te,0}$.
This motivates us to choose the weights such that $\hat{F}_{te,1}^{w}(x)$ has moments identical to those of $\hat{F}_{te,0}(x)$. We use the first and second moments in practice for computation consideration. For $s=0,1$, let $i_{1}^{(s)},\dots,i_{|I_{te,s}|}^{(s)}$ be the indices in $I_{te,s}$ arranged in an ascending order.
Let $U_{s}$ be the matrix whose $k$-th row is 
$$(1, X_{i_{k}^{(s)}1}, \dots, X_{i_{k}^{(s)}d}, X_{i_{k}^{(s)}1}^{2}, X_{i_{k}^{(s)}1}X_{i_{k}^{(s)}2}, \dots,X_{i_{k}^{(s)}1}X_{i_{k}^{(s)}d},X_{i_{k}^{(s)}2}^{2},X_{i_{k}^{(s)}2}X_{i_{k}^{(s)}3},\dots,X_{i_{k}^{(s)}d}^{2})$$
where $X_{i_{k}^{(s)}j}$ is the $j$-th element of $X_{i_{k}^{(s)}}$ for $k = 1,\dots, |I_{te,s}|$ and $s = 0,1$. Let $\bar{U}_{0}=|I_{te,0}|^{-1}U_{0}^{\T} 1_{te,0}$ where $1_{te,0}$ is a $|I_{te,0}|$-dimensional vector of $1$'s. 
Then the first and second moments of $\hat{F}_{te,1}^{w}$ and $\hat{F}_{te,0}$ are identical if and only if $U_{1}^{\T}w = \bar{U}_{0}$ where $w = (w_{i_{1}}, \dots, w_{i_{|I_{te,1}|}})^{\T}$ is the vector of weights. The linear system $U_{1}^{\T}w = \bar{U}_{0}$ might have multiple solutions. In practice, we do not want the weights to be too large or too small. Thus among the weight vectors that solve $U_{1}^{\T}w = \bar{U}_{0}$, we prefer the one that is closest to the uniform weight $w_{\rm uni} = |I_{te,1}|^{-1}1_{te,1}$. Here $1_{te,1}$ is a $|I_{te,1}|$-dimensional vector of $1$'s. This motivates us to  obtain the weights via solving the following optimizing problem
\begin{equation}\label{eq: weight origin}
	\begin{aligned}
		&\min_{w}\|w-w_{\rm uni}\|^{2}\\
		&s.t.\quad U_{1}^{\T}w = \bar{U}_{0}.
	\end{aligned}
\end{equation}
On the other hand, $ U_{1}^{\T}w = \bar{U}_{0}$ may have no solution and the feasible set of \eqref{eq: weight origin} may be empty. To mitigate, we relax the constrain in \eqref{eq: weight origin} to $w \in \argmin\|U_{1}^{\T}w-\bar{U}_{0}\|^{2}$. It is straightforward to verify that $w$ is a minimum point of $\|U_{1}^{\T}w-\bar{U}_{0}\|^{2}$ if and only if $U_{1}U_{1}^{\T}w = U_{1}\bar{U}_{0}$ and this leads to the problem 
\begin{equation}\label{eq: weight general}
	\begin{split}
		&\min_{w}\|w-w_{\rm uni}\|^{2}\\
		&s.t. \quad U_{1}U_{1}^{\T}w = U_{1}\bar{U}_{0}.
	\end{split}
\end{equation}
It is easy to verify that \eqref{eq: weight general} is equivalent to \eqref{eq: weight origin} if the feasible set of \eqref{eq: weight origin} is non-empty.
By some algebra, the solution of \eqref{eq: weight general} has the following closed form
\begin{equation}\label{eq: solution weight general}
	w_{\rm uni}-U_{1}U_{1}^{\T}(U_{1}U_{1}^{\T}U_{1}U_{1}^{\T})^{-}(U_{1}U_{1}^{\T}w_{\rm uni}-U_{1}\bar{U}_{0}),
\end{equation}
where $A^{-}$ denotes the generalized inverse of $A$ for any matrix $A$. In practice, $U_{1}^{\T}U_{1}$ is usually invertible as long as $|I_{te,1}| > d(d+2)/2$, i.e., $U_{1}$ has more rows than columns. If $U_{1}^{\T}U_{1}$ is invertible, some algebra can show that the expression \eqref{eq: solution weight general} can be further simplified to $w_{\rm uni}-U_{1}(U_{1}^{\T}U_{1})^{-1}(U_{1}^{\T}w_{\rm uni}-\bar{U}_{0})$. In this case, calculation of the weights is fast with a computing time of order $\Theta(N)$. 

Besides the problem caused by missing data, both \eqref{eq: classical CV 1} and \eqref{eq: classical CV weight} are computationally cumbersome in the distributed data setting. 
To show this fact, we consider \eqref{eq: classical CV 1} for illustration. 
Suppose $\hat{m}_{tr,c}(x)$
is obtained by the kernel regression based on the training data.
Then for each $c\in\mathcal{C}$, the computing time and communication complexity for obtaining \eqref{eq: classical CV 1} are $\Theta(\kappa N^2/L)$ and $\Theta(\kappa LN)$, respectively, where $\kappa$ is the number of repeated splittings in cross validation. If $\hat{m}_{tr,c}(x)$
is obtained by the sieve least squares regression based on the training data, the computing time and communication complexity  become $\Theta(\kappa NK/L + \kappa K^{3})$ and $\Theta(\kappa LK^{2})$, respectively. By plugging in $K = cN^{d/(2d + 1)}$, the computing time and communication complexity are $\Theta(\kappa N^{(3d+1)/(2d+1)}/L + \kappa N^{3d/(2d+1)})$ and $\Theta(\kappa LN^{2d/(2d+1)})$, respectively. These computing and communication  costs are all high if $N$ and $\kappa$ are large.

In this section, we propose a distributed cross validation method to reduce the communication cost and the computing time. 
Specifically, we split the data on the $l$-th machine into training and test data for $l = 1,\dots, L$. Let $I_{tr}^{(l)}$ and $I_{te}^{(l)}$ be the index set of training data and test data on the $l$-th machine. $I_{te,1}^{(l)}$, $I_{te,0}^{(l)}$, $U_{1}^{(l)}$ , $U_{0}^{(l)}$ and $\bar{U}_{0}^{(l)}$ are defined in the same way as $I_{te,1}$, $I_{te,0}$, $U_{1}$ , $U_{0}$ and $\bar{U}_{0}$ with the whole data set replaced by the data on the $l$-th machine. Then for each $c\in \mathcal{C}$,
we calculate the criterion
\begin{equation*}\label{eq: dist CV prop}
	Q_{l}(c) = |I_{te, 1}^{(l)}|^{-1}\sum_{i\in I_{te,1}^{(l)}}w_{i}^{(l)}\left(Y_{i}-\hat{m}_{tr,c}^{(l)}(X_{i})\right)^{2},
\end{equation*}
and choose the constant $c$ which minimizes $L^{-1}\sum_{l = 1}^{L}Q_{l}(c)$.
Here the weights $w^{(l)}_{i}$ are obtained via equation \eqref{eq: solution weight general} with 
$w_{\rm uni}$, $U_{1}$ , $U_{0}$ and $\bar{U}_{0}$ replaced by $w_{\rm uni}^{(l)}$, $U_{1}^{(l)}$, $U_{0}^{(l)}$ and $\bar{U}_{0}^{(l)}$, respectively, where $w_{\rm uni}^{(l)} = |I_{te,1}^{(l)}|^{-1}1_{te,1}^{(l)}$ and $1_{te,1}^{(l)}$ is a $|I_{te,1}|$-dimensional vector of $1$'s.
For the KDI method,
$\hat{m}_{tr,c}^{(l)}(x)$ is obtained via the kernel regression based on the training data on the $l$-th machine with bandwidth $c|I_{tr}^{(l)}|^{-1/(2d + 1)}$. For the SDI method,   $\hat{m}_{tr,c}^{(l)}(x)$ is obtained by the sieve least squares regression based on the training data on the $l$-th machine with $\lceil c|I_{tr}^{(l)}|^{d/(2d + 1)}\rceil$ basis functions. At the beginning of this section, we recommend the rates for $h$ and $K$ according to some theoretical consideration. Thus we determine the rate of the tuning parameters used in $\hat{m}_{tr,c}^{(l)}(x)$ according to the size of the training data and select the constant in front of the rate here.

In classical CV, the data splitting procedure is repeated many times to make the selection more stable. In the proposed distributed cross validation procedure, we split the data and calculate the criterion on each machine in parallel. The criteria calculated on different machines are independent and identically distributed. Thus by taking the average over different machines, we obtain a stable criterion and we do not need to repeat the data splitting procedure. In our exploratory simulation study, little benefit is observed if we repeat the data splitting procedure on each machine for several times.  

Next, we summarize the above procedures and propose a CV method for the selection of $c$ for the KDI method in Algorithm \ref{al:dist cv KDI}. We refer to the proposed novel CV method as the {\itshape distributed weighted cross validation} (DWCV) algorithm. 
In the following algorithm, we use half of the data on each machine as the training data and half as the test data.
\begin{breakablealgorithm}
	\caption{DWCV algorithm for the KDI method}\label{al:dist cv KDI}
	\begin{algorithmic}[1]
		\State \textbf{Initialization}: a candidate set $\mathcal{C} = \{c_{1},\dots,c_{\mathcal{J}}\}$;
		\State For $l = 1,\dots, L$, split data on the $l$-th machine into training data of size $n/2$ with index set $I_{tr}^{(l)}$ and test data of size $n/2$ with index set $I_{te}^{(l)}$;
		\For{$j=1,\dots,\mathcal{J}$};
		\State  For each machine,
		\If{$U_{1}^{(l)\T}U_{1}^{(l)}$ is invertible}
		\State calculate the weight vector
		\[w^{(l)} = w_{\rm uni}^{(l)}-U_{1}^{(l)}(U_{1}^{(l)\T}U_{1}^{(l)})^{-1}(U_{1}^{(l)\T}w_{\rm uni}^{(l)} - \bar{U}_{0}^{(l)}),\]
		\newpage
		\Else
		\State calculate the weight vector
		\[w^{(l)} = 
		w_{\rm uni}^{(l)}-U_{1}^{(l)}U_{1}^{(l)\T}(U_{1}^{(l)}U_{1}^{(l)\T}U_{1}^{(l)}U_{1}^{(l)\T})^{-}(U_{1}^{(l)}U_{1}^{(l)\T}w_{\rm uni}^{(l)}-U_{1}^{(l)}\bar{U}_{0}^{(l)});\]
		\EndIf
		\State Take $h_{tr} = c_{j}(2/n)^{1/(2d+1)}$;
		\State For each machine, calculate the weighted predicting error
		\[E_{j}^{(l)} = \sum_{i\in I_{te,1}^{(l)}}w_{i}^{(l)}\left(Y_{i}-\hat{m}_{tr,c}^{(l)}(X_{i})\right)^{2};\]	
		where $w_{i}^{(l)}$ is the $|\{k: k\leq i,\ k\in I_{te,1}^{(l)}\}|$-th component of $w^{(l)}$ and \[\hat{m}_{tr,c}^{(l)}(x)=\frac{\sum_{i\in I_{tr}^{(l)}}K_{h_{tr}}(X_i-x)\delta_iY_i}{\sum_{i\in I_{tr}^{(l)}}K_{h_{tr}}(X_i-x)\delta_i};\]  
		\State Transmit $E_{j}^{(l)}$ to the first machine and calculate $E_{j} = L^{-1}\sum_{l=1}^{L}E_{j}^{(l)}$;
		\EndFor
		\State  Select the tuning parameter $c_{j^{*}}$ with $j^{*} = \argmin_{j=1,\dots,\mathcal{J}}E_{j}$.
	\end{algorithmic}
\end{breakablealgorithm}
The DWCV algorithm avoids the high communication cost and the repetition of data splitting in classical CV. It also adjusts the bias caused by using observed data to evaluate the missing data imputation accuracy via a computationally simple weighting procedure. The computing time of the weight $w^{(l)}$ is of order $\Theta(N/L)$ if $U_{1}^{(l)\T}U_{1}^{(l)}$ is invertible and of order $\Theta(N^{3}/L^{3})$ otherwise. We focus on the former case in the computing time analysis because $U_{1}^{(l)\T}U_{1}^{(l)}$ is usually invertible as long as $n > d(d+2)$ and the relationship $n > d(d+2)$ is likely to hold for large-scale data set. For each $c\in \mathcal{C}$, the computing time and the communication complexity of the DWCV algorithm for the KDI method are of order $\Theta(N^{2} /L^{2})$ and $\Theta(L)$ if $U_{1}^{(l)\T}U_{1}^{(l)}$ is invertible for $l = 1,\dots, L$.  Thus, the DWCV algorithm for the KDI method has a much shorter computing time and a much lower communication cost compared to those of the CV method based on \eqref{eq: classical CV 1} and the kernel regression.  

A similar procedure as Algorithm \ref{al:dist cv KDI} can be used to select the constant $c$ for the SDI method. Here we make a modification for computation considerations. For the SDI method, the computing time and communication complexity are relevant to the number of basis functions $K$ and hence the constant $c$. A larger $c$ results in a longer computing time and a higher communication complexity.
See Table \ref{table: time tuning parameter} in the simulation section for an illustration.
Hence we prefer a smaller $c$ if a larger $c$ does not bring about a significant accuracy improvement. Thus in the DWCV algorithm for the SDI method, we evaluate the candidate constants from the smallest to the largest. We stop the loop and select the current $c$ if the weighted prediction error does not change sufficiently from the current $c$ to the next. The procedure is summarized in Algorithm  \ref{al:dist cv SDI}.

\begin{breakablealgorithm}
	\caption{DWCV algorithm for the SDI method}\label{al:dist cv SDI}
	\begin{algorithmic}[1]
		\State \textbf{Initialization}: a candidate set $\mathcal{C} = \{c_{1},\dots,c_{\mathcal{J}}\}$ with the numbers arranged in an ascending order;
		\State For $l = 1,\dots, L$, split data on the $l$-th machine into training data of size $n/2$ with index set $I_{tr}^{(l)}$ and test data of size $n/2$ with index set $I_{te}^{(l)}$;
		\For{$j=1\dots,\mathcal{J}$};
		\If{$U_{1}^{(l)\T}U_{1}^{(l)}$ is invertible}
		\State calculate the weight vector
		\[w^{(l)} = w_{\rm uni}^{(l)}-U_{1}^{(l)}(U_{1}^{(l)\T}U_{1}^{(l)})^{-1}(U_{1}^{(l)\T}w_{\rm uni}^{(l)} - \bar{U}_{0}^{(l)}),\]
		\Else
		\State calculate the weight vector
		\[w^{(l)} = 
		w_{\rm uni}^{(l)}-U_{1}^{(l)}U_{1}^{(l)\T}(U_{1}^{(l)}U_{1}^{(l)\T}U_{1}^{(l)}U_{1}^{(l)\T})^{-}(U_{1}^{(l)}U_{1}^{(l)\T}w_{\rm uni}^{(l)}-U_{1}^{(l)}\bar{U}_{0}^{(l)});\]
		\EndIf
		\State Take $K_{tr}=\lceil c_{j}(n/2)^{d/(2d+1)} \rceil$;
		\State For each machine, calculate 
		\[\hat{\Sigma}^{(l)} = \frac{2}{n}\sum_{i\in I_{tr}^{(l)}}\delta_{i}V_{K_{tr}}(X_{i})V_{K_{tr}}(X_{i})^{\T},\]
		and 
		\[
		\hat{\beta}^{(l)} =(\hat{\Sigma}^{(l)})^{-1}\frac{2}{n}\sum_{i\in I_{tr}^{(l)}}\delta_{i}Y_{i}V_{K_{tr}}(X_{i});
		\]
		\State For each machine, calculate the weighted predicting error
		\[E_{j}^{(l)} = \sum_{i\in I_{te,1}^{(l)}}w_{i}^{(l)}\left(Y_{i}-\hat{m}_{tr}^{(l)}(X_{i})\right)^{2},\]
		where $w_{i}^{(l)}$ is the $|\{k: k\leq i,\ k\in I_{te,1}^{(l)}\}|$-th component of $w^{(l)}$ and $\hat{m}_{tr,c}^{(l)}(x) = V_{K_{tr}}(x)^{\T}\hat{\beta}^{(l)}$;
		\State Transmit $E_{j}^{(l)}$ to the first machine and calculate $E_{j} = L^{-1}\sum_{l=1}^{L}E_{j}^{(l)}$;
		\State If $(E_{j-1}-E_{j})/E_{j-1}<0.01$, set $E_{k} = \infty$ for $k = j\dots, \mathcal{J}$ and break;
		\EndFor
		\State  Select the tuning parameter $c_{j^{*}}$ with $j^{*} = \argmin_{j=1,\dots,\mathcal{J}}E_{j}$.
	\end{algorithmic}
\end{breakablealgorithm}
For each $c\in \mathcal{C}$, the DWCV algorithm for the SDI method has a computing time and a communication complexity of order $\Theta((N/L)^{(3d+1)/(2d+1)} + (N/L)^{3d/(2d+1)})$ and $\Theta(L)$ when $U_{1}^{(l)\T}U_{1}^{(l)}$ is invertible for $l = 1,\dots, L$. It can be seen that the DWCV algorithm for the SDI method makes a great improvement on both the computing time and the communication complexity compared to the CV method based on \eqref{eq: classical CV 1} and the sieve least squares regression.

\section{Simulation Study}\label{sec: sim}
\subsection{Performance under Different $d$ and $L$}\label{subsec: sim main}
In this section, some simulation studies were conducted to evaluate the performance of the two proposed distributed nonparametric methods. 

The $d$-dimensional covariate $X$ is used to impute the missing outcome $Y$ in the simulation.
We consider two scenarios with different values of $d$ to evaluate the impact of the covariate dimension. First, a random vector $Z=(Z_1,Z_2,Z_3,Z_4,Z_5)$ is generated where $Z_j\sim N(0,1)$ for $j=1,2,3$ and $Z_j\sim U(-1,1)$ for $j=4, 5$. We set $X = Z$ in the first scenario while $X = (Z,Z^{\prime}, Z^{\prime\prime})$ in the second scenario, where $Z^{\prime}$ and $Z^{\prime\prime}$ are two independent duplications of $Z$. Clearly, $d=5$ in the first scenario and $d=15$ in the second.
In both scenarios, the response is generated from 
\begin{equation*}
	Y = 2+0.5Z_{1}+0.5Z_{2}+0.5Z_{3}+0.5Z_{4}+Z_{5}+\epsilon
\end{equation*}
with $\epsilon\sim N(0,1)$, and the missing mechanism is set to be \[P(\delta=1\mid Z) = 0.5\times\frac{\exp(Z_1+Z_2+Z_3+Z_4+Z_5)}{1+\exp(Z_1+Z_2+Z_3+Z_4+Z_5)}+0.5.\]
We fix the total sample size $N=2\times 10^{5}$ and vary the number of machines $L = 10, 20, 50, 100$, $200$ and $500$ to evaluate the effect of  machine number. The proposed distributed imputation estimators are compared with other three estimators: 
the single machine (SGM) estimators calculated by the classical kernel and sieve methods using the data on a single machine; the complete-case estimator given by the sample mean of observed responses; and the oracle estimator given by the sample mean of all responses.  The oracle estimator is infeasible in practice and here we use it as the gold standard. We generate $200$ Monte Carlo random samples to evaluate the performance of the different estimators in the two scenarios. In the simulation, the bias and standard error (SE) of the oracle estimator are $0.000$ and $0.003$, respectively. The complete-case estimator has a large bias, $0.208$, which makes its small SE, $0.004$, meaningless.

A kernel function of order $20$ based on Legendre Polynomial \citep{Berlinet1993kernel} is used to implement the kernel regression imputation method. Motivated by the rate given at the beginning of Section \ref{sec: tuning selection}, we use the bandwidth $h = c(L\log L/N)^{1/(2d+1)}$. Here we add a factor of $\log L$ as it can improve the finite sample performance of the KDI estimators. The constant $c$ is taken to be $1.3$ when $d = 5$ and $1.7$ when $d=15$. Further simulations on the selection of tuning parameters are presented in the next section.  The following figure contains the box-plots of kernel-based SGM estimators and KDI estimators with different covariate dimensions and numbers of machines.

\centerline{[Insert Figure \ref{fig: Linear_boxplot_kernel} about here.]}

The box-plots show that the performance of the KDI estimators is more stable compared to the kernel-based SGM estimators.
The bias and SE of these estimators are calculated in Table \ref{table: Linear_kernel}.

\centerline{[Insert Table \ref{table: Linear_kernel} about here.]}

In the first scenario ($d=5$), it can be seen that the kernel-based SGM estimators and KDI estimators have a much smaller bias compared to the complete-case estimator. 
The kernel-based SGM estimators have large SE especially when the number of machines is large since only a small subset of all data is used.
In contrast, the KDI estimators have comparable performance as the oracle estimator when the number of machines is no more than $100$. The SE of the KDI estimators remains small even though the number of machines is large. The bias is a little large when $L=500$.

In the second scenario ($d=15$), the SE of the kernel-based SGM estimators and the KDI estimators is similar to those in the first scenario. The bias, however, is larger than that in the first scenario especially when $L$ is large. Unlike the case of $d=5$, the KDI estimators have comparable performance as the oracle estimator only when $L = 10$. This is consistent with the previous discussion on the theoretical results.

For the SDI method, we use the tensor product of polynomials as basis functions and let $K = \lceil cN^{d/(2d+1)}\rceil$. We take the constant $c$ to be $0.5$ when $d=5$ and $0.9$ when $d=15$. Figure~\ref{fig: Linear_boxplot_sieve} shows the box-plots of sieve-based SGM estimators and SDI estimators with different covariate dimensions and numbers of machines.

\centerline{[Insert Figure \ref{fig: Linear_boxplot_sieve} about here.]}

It can be seen that the SDI estimators perform much better than the sieve-based SGM estimators in both scenarios especially when $L$ is $100$, $200$, and $500$. The bias and SE of these estimators are summarized in the following table. 

\centerline{[Insert Table \ref{table: Linear_sieve} about here.]}

The SDI estimators perform similarly to the oracle estimator in terms of bias and SE regardless of how large $L$ is in the first scenario ($d=5$). In the second scenario ($d=15$), the bias and SE of the sieve-based SGM and SDI estimators are similar to those in the first scenario. This implies the larger covariate dimension has little effect on the accuracy of the SDI estimators. Moreover, the bias of the SDI estimator is much smaller than that of the KDI estimators when $L$ is large in the second scenario. This illustrates the capacity of the SDI method to accommodate a large number of machines.

To further explore the performance of the proposed methods under different data generation processes (DGPs), we consider another DGP in which the response is generated from a more complex model
\begin{equation}\label{eq: mod nonlinear}
	Y=2+\sin(Z_1+Z_2+Z_3)+2\Phi^{-1}(0.5Z_4+0.5)+2\Phi^{-1}(0.5Z_5+0.5)+\epsilon
\end{equation}
where $\epsilon\sim N(0,1)$ and other simulation settings are the same as before.
To distinguish the two settings, we call the previous DGP the linear setting and call the DGP considered here the nonlinear setting. 
In the simulation, the bias and standard error (SE) of the oracle estimator are $0.000$ and $0.007$, respectively. The complete-case estimator still has a large bias, $0.326$, and a small SE, $0.008$.
Figure~\ref{fig: Nonlinear_boxplot_kernel} shows the box-plots of kernel-based SGM estimators and SDI estimators with different covariate dimensions and numbers of machines.
Table~\ref{table: Nonlinear_kernel} presents the bias and SE of these estimators.

\centerline{[Insert Figure \ref{fig: Nonlinear_boxplot_kernel} about here.]}

\centerline{[Insert Table \ref{table: Nonlinear_kernel} about here.]}

The simulation results for the first scenario ($d=5$) are similar to those in the linear setting. For the second scenario ($d=15$), we find that the bias of the KDI estimators has drastic change as $L$ changes and is large for most $L$'s. Combining this with the performance of the KDI estimators in the linear setting with $d = 15$, we regard the KDI method as unreliable when $d =15$ under the sample size considered in the simulation.

The box-plots of sieve-based SGM estimators and SDI estimators with different covariate dimensions and numbers of machines are presented in Fig.~\ref{fig: Nonlinear_boxplot_sieve}. 
Table~\ref{table: Nonlinear_sieve} presents the bias and SE of these estimators. 

\centerline{[Insert Figure \ref{fig: Nonlinear_boxplot_sieve} about here.]}

\centerline{[Insert Table \ref{table: Nonlinear_sieve} about here.]}

Similar phenomenons as those under the linear setting can be observed in Fig.~\ref{fig: Nonlinear_boxplot_sieve} and Table~\ref{table: Nonlinear_sieve}. 
The SDI estimators outperform the SGM estimators for all $L$ and $d$.
By comparing the performance of the SDI method under two settings, we find that the bias of SDI estimators is smaller in the linear setting.
This is consistent with our discussion at the end of Section \ref{subsec: sieve theory}.
In the linear setting, the approximation error of the basis functions concerning the true conditional mean function is zero as we use polynomials as the basis functions, and this results in a small finite sample bias. 
Further explorations about the impact of the approximation error on SDI estimators can be found in Section \ref{sec: SDI approximation error}.

Besides the estimation accuracy, we also evaluated the computation efficiency of the proposed distributed estimators. 
We compare the computing time of the KDI and SDI estimators for different numbers of machines with that of the classical non-distributed estimators which are computed using the whole data on one machine.
All computations are performed in the R Programming \citep{Rcore2016} using a windows server with a 24-core processor and 128GB RAM.
Table \ref{table: simtime} presents CPU times (in second) required to obtain the proposed distributed estimators and the classical estimators under different settings with different $d$. 

\centerline{[Insert Table \ref{table: simtime} about here.]}

Table \ref{table: simtime} shows that the computing times of KDI and SDI methods decrease as the number of machines increases. Both the KDI and SDI methods have a significantly shorter computing time compared to the classical regression imputation method. The SDI method has a shorter computing time compared to the KDI method when $L$ is no larger than $50$. When $L$ is no smaller than $100$, the KDI method is faster. This might be because the computing time of the KDI method is theoretically proportional to $n^{2} = N^{2} / L^{2}$ which decreases at the rate of $1/L^{2}$ as $L$ increases and the computing time of the SDI method does not decrease at such a fast rate.

As can be seen from Table \ref{table: simtime}, the computing time of the classical sieve regression imputation estimator is not that long. Therefore, we consider a larger $N=2\times 10^{6}$ to further demonstrate the effectiveness of the SDI method. The computing times of the SDI estimators with $L=$ $10$, $20$, $50$, $100$, $200$, $500$ are $559.08$, $428.01$, $229.62$, $80.92$,
$47.60$, $26.34$ respectively, while the computing time of the classical sieve regression imputation estimator is $4003.38$. 
The proposed SDI method can reduce the computing time from over one hour to a few minutes when the number of machines is sufficiently large. Comparing this result to those in Table \ref{table: simtime}, we can see that the improvement in terms of the computational time is more significant for a larger sample size.

\subsection{Tuning Parameter Selection}\label{subsec: sim tuning}

In this subsection, we evaluate the performance of the proposed distributed cross validation strategy. 
We consider the candidate set $\mathcal{C} = \{0.1$, $0.5$, $0.9$, $1.3$, $1.7$, $2.1\}$ for $c$. We set $L = 100$ and examine the effect of different $c$'s and $d$'s in this subsection.
Table \ref{table: time tuning parameter} presents the computing time of the KDI/SDI estimators in the nonlinear setting with $d=5$, $L = 100$ and different $c$'s in $\mathcal{C}$.

\centerline{[Insert Table \ref{table: time tuning parameter} about here.]}

From Table \ref{table: time tuning parameter}, the computing time of the KDI estimator is similar for different $c$, while that of the SDI estimator is longer for larger $c$. 
This verifies the motivation of the preference in smaller $c$ in Algorithm \ref{al:dist cv SDI}.

Further, we compare the RMSE of the KDI/SDI estimator under the selected $c$ with that of the KDI/SDI estimators under different $c$'s in $\mathcal{C}$. 
Figures \ref{fig: Linear_bws} and \ref{fig: Nonlinear_bws} present the simulation results in the linear and nonlinear setting, respectively.

\centerline{[Insert Figure \ref{fig: Linear_bws} about here.]}

\centerline{[Insert Figure \ref{fig: Nonlinear_bws} about here.]}

As can be seen, the KDI estimator with constant selected via the proposed DWCV algorithm has the smallest RMSE among estimators with different candidate constants in both the linear and nonlinear setting when $d=5$. When $d = 15$, the estimator with CV selected constant still has a performance close to the optimal candidate constant. 

The performance of the SDI method is similar for different constants in the linear setting, which is also similar to the performance of the SDI method with the constant selected by the proposed DWCV algorithm. In the nonlinear setting, the proposed DWCV algorithm is able to select the smallest constant with a nearly optimal performance, which is in line with our design intention. 

In addition, we conducted an ablation study to examine the effectiveness of weighting and distributed calculation procedure in the proposed DWCV algorithm. We delegate it to the Appendix due to limited space.

\subsection{The Approximation Error of the SDI Method}\label{sec: SDI approximation error}
In this section, we provide some further simulation evidence on the phenomenon discussed at the end of Section \ref{subsec: sieve theory}. The simulation in Section \ref{subsec: sim main} has shown that the SDI method has a small bias if $m(x)$ can be approximated well by the basis functions used. In this section, we show that this is also true if $1/\pi(x)$ can be approximated well by the basis functions used. To this end, we consider another simulation setting which is the same as the nonlinear setting considered in Section \ref{subsec: sim main} except that the missing mechanism is
\[
P(\delta = 1\mid Z) = \frac{1}{1 + (Z_{1} + Z_{2} + Z_{3} + Z_{4} + Z_{5} - 0.5 )^{2}}.
\]
Here we still consider two scenarios with $X=Z$ ($d = 5$) and $X = (Z,Z^{\prime},Z^{\prime\prime})$ ($d=15$).
Under the mechanism considered here, $1/\pi(x)$ is a polynomial of $x$ and hence the SDI method has a small approximation error under this setting according to the analysis at the end of Section \ref{subsec: sieve theory}. We evaluate the performance of the SDI estimators with $d = 5$ and $15$ based on $200$ Monte Carlo random samples. In this simulation, the bias and standard error (SE) of the oracle estimator are $0.000$ and $0.007$, respectively. The complete-case estimator has a bias $0.364$ and a SE $0.009$, which are both slightly larger than those under the nonlinear setting in Section \ref{subsec: sim main}. Moreover, the missing rate is about $55.9\%$, which is much larger than the missing rate $25.0\%$ under the nonlinear setting in Section \ref{subsec: sim main}.
The bias and SE of SDI estimators under this setting are summarized in the following table. 

\centerline{[Insert Table \ref{table: br} about here.]}

Comparing Table \ref{table: br} with Table \ref{table: Nonlinear_sieve}, we find that the SDI has a much smaller bias here especially when $d=5$ or $L = 10,20,50$. This verifies the theoretical analysis at the end of Section \ref{subsec: sieve theory}. As illustrated by Tables \ref{table: Linear_sieve} and \ref{table: br}, a pretty small finite sample bias can be expected for the SDI method as long as the basis functions used can approximate either $m(x)$ or $1/\pi(x)$ well.
Moreover, the SDI method has a larger SE compared to  Table \ref{table: Nonlinear_sieve}, which is not surprising because the missing rate is much higher under the setting considered here compared to that in Section \ref{subsec: sim main}.

\section{Real Data Analysis}
GroupLens Research has collected and made available movie rating data sets on the MovieLens website (\texttt{https://movielens.org}). In this section, we apply our method to a large-scale movie rating dataset, the \emph{ml-25m} dataset which contains $25000095$ ratings created by $162541$ users and $1093360$ tag applications across $62423$ movies. 
The parameter of interest is the average rating of the film \emph{Pulp Fiction} among all users. 
Less than half of users ($79672$) create their rating for \emph{Pulp Fiction} and ratings of the other users are missing. 
The average rating of users who create their rating for \emph{Pulp Fiction} may be a biased estimator for the parameter of interest.  
We create a $10$-dimensional covariate vector to describe each user's characteristic based on their rating histories to adjust this bias and assume the MAR missing mechanism.  
We apply our two distributed methods to compute the regression imputation estimators. The rates of tuning parameters are the same as those in Section \ref{sec: sim} and the constant $c$ is taken to be $1.3$ and $0.9$ for the KDI and the SDI method, respectively. We also compute the complete-case estimator, the classical kernel regression imputation estimator ($\hat{\mu}_{\mathbb{K}}$), and the sieve regression imputation estimator ($\hat{\mu}_{\mathbb{S}}$) for comparison. The results are $4.19$ (complete-case), $4.12$ ($\hat{\mu}_{\mathbb{K}}$) and $4.04$ ($\hat{\mu}_{\mathbb{S}}$), respectively. We use $\hat{\mu}_{\mathbb{K}}$ as the benchmark for the kernel-based estimators and use $\hat{\mu}_{\mathbb{S}}$ as the benchmark for the sieve-based estimators. In Fig \ref{fig: movie}, we plot curves of the absolute value of the differences between the complete-case, KDI, SDI estimators and their benchmarks. We also plot the absolute difference curves for the kernel and sieve-based SGM estimators for comparison.

\centerline{[Insert Figure \ref{fig: movie} about here.]}

It can be seen that the performance of all the SGM estimators is unstable. The KDI performs well when $L$ is small but the performance deteriorates when the number of machines is large. In contrast, the performance of SDI estimators stays good even when $L$ is $500$. The SDI estimators are quite stable and always outperform the SGM estimators. This again validates the ability of the SDI method to accommodate a large number of machines.

The computing time of the KDI method is $110.58, 26.72, 4.28, 1.11, 0.29, 0.06$ seconds when $L = 10, 20, 50, 100, 200, 500$, respectively. These are all much shorter than that of the classical kernel regression imputation estimator ($20896.18$ seconds). However, the reduction of computing time is achieved at the cost of some potential accuracy loss. The absolute difference between the KDI estimators and the classical kernel regression imputation estimator is at least $0.016$. Computation of the classical sieve regression imputation estimator takes $78.17$ seconds, which is brought down to $9.86, 5.72, 3.89, 2.66, 2.29, 1.99$ seconds by the SDI method when $L = 10, 20, 50, 100, 200, 500$, respectively. The SDI method has little accuracy loss in this problem as the maximum absolute difference between the SDI estimators and the classical sieve regression imputation estimator is at most $0.001$.
According to the empirical evidence in simulations and real data analyses, we recommend using the SDI method with a large number of machines concerned with the estimation accuracy and computing time.

We also evaluate the performance of the proposed tuning parameter selection procedure in the MovieLens dataset. We consider the candidate set $\mathcal{C} = \{0.1$, $0.5$, $0.9$, $1.3$, $1.7$, $2.1\}$ for $c$ and the number of machine $L=100$. 
The absolute difference between the KDI estimator under the selected $c$ and the benchmark $\hat{\mu}_{\mathbb{K}}$ is $0.038$. For comparison, we calculate the the KDI estimators under different $c$'s in $\mathcal{C}$. The absolute differences between red the KDI estimator with $c = 0.1, 0.5, 0.9, 1.3, 1.7, 2.1$ and the benchmark $\hat{\mu}_{\mathbb{K}}$ are $2.072$, $2.065$, $0.351$, $0.038$, $0.046$ and $0.052$, respectively.
It can be seen that the KDI estimator with constant selected via the proposed DWCV algorithm is closest to the benchmark $\hat{\mu}_{\mathbb{K}}$ among the KDI estimators with different candidate constants. Moreover, the computing time of the KDI estimator with CV-selected constant is $3.73$ seconds. The time is slightly larger than that of the KDI estimator with a predetermined constant which is $1.11$ when $L = 100$. The extra computing time comes from the CV procedure. The computing time is still much shorter than the classical kernel regression imputation estimator even though the CV procedure is included.

The absolute differences between the SDI estimator under the selected $c$ and the benchmark $\hat{\mu}_{\mathbb{S}}$ is $0.007$. The corresponding absolute differences with $c = 0.1, 0.5, 0.9, 1.3, 1.7, 2.1$ are $0.007$, $0.004$, $0.001$, $0.010$, $0.014$ and $0.014$, respectively.
The SDI estimators with different constants have similar performance and all of them are close to the benchmark $\hat{\mu}_{\mathbb{S}}$. As we have designed, the proposed DWCV algorithm selects the smallest $c$ in the candidate set $\mathcal{C}$ in this case to reduce the computing time. The computing time of the SDI estimator with the selected $c$ is $0.78$, which is even shorter than the computing time $2.66$ of the SDI estimator with a predetermined constant when $L = 100$.

\section{Concluding Remarks}
In this paper, we propose two distributed regression imputation methods for response mean estimation, namely the KDI and SDI methods. Compared to the classical method, our distributed methods can reduce the computational burden with large-scale data significantly. The KDI method needs extremely little communication between machines. It can achieve the minimax rate and is asymptotically normal with asymptotic variance achieving the semiparametric efficiency bound provided the number of machines $L$ is not too large. The SDI method can achieve similar statistical properties while accommodating more machines compared to the KDI method at the expense of more required communication. If the communication cost is high, the KDI method is recommended and the number of local machines should be limited to avoid the additional bias of this method.
If the communication cost is low, the SDI method is recommended and more local machines can be used to reduce computing time. Moreover, according to our simulation results, the SDI method is preferred if the number of covariates is large.



\acks{We would like to gratefully acknowledge the GroupLens for making the datasets available.
	Wang's research was supported by the National Natural Science Foundation of China (General program 11871460 and program for Innovative Research Group 61621003), a grant from the Key Lab of Random Complex Structure and Data Science, CAS.}


\newpage

\appendix
\section*{Appendix A. Some Lower Bounds}
\subsection*{Minimax Lower Bound}
Here, we establish a minimax lower bound for the response mean estimation problem with missing responses. 
For constants $M_{1}, M_{3}, M_{4} > 0$, $0 < M_{2} < 1$ and $q > d$, let
\begin{equation}\label{eq: minimax set1}
	\begin{aligned}
		\mathcal{P}_{1} = \bigg\{P:\ & \text{$P$ is an observational distribution of $(\delta, Y, X)$ and under $P$,} \\
		&\text{(i) $\pi(x)$, $f(x)$, $m(x)$ and their partial derivatives up to order $q > 0$ }\\
		&\quad\text{is bounded by $M_{1}$;}\\
		&\text{(ii) $\inf_{x}\pi(x) \geq M_{2}$; (iii) $\inf_{x}f(x) \geq M_{3}$; (iv) $\sup_{x}\sigma^{2}(x) \leq M_{4}$}
		\bigg\}
	\end{aligned}
\end{equation}
be the set of observational distributions that satisfy the regularity conditions imposed in Theorem \ref{thm:kernel}. Then we have the following minimax lower bound.
\begin{proposition}\label{prop: minimax}
	There are some universal constants $C_{\mathcal{P}_{1}}$, $\xi> 0$ such that 
	\[\inf_{\hat{\mu}}\sup_{P\in \mathcal{P}_{1}} P\left(|\hat{\mu} - \mu| \geq \frac{C_{\mathcal{P}_{1}}}{\sqrt{N}}\right)\geq \xi,\]
	where the infimum is taken over all estimators for $\mu$ based on $N$ i.i.d. observations.
\end{proposition}
\begin{proof}
	We prove the result by considering a tractable parametric subclass of $\mathcal{P}_{1}$. Let 
	\[\begin{aligned}
		\mathcal{P}_{\rm par} = \bigg\{P:\ & \text{$P$ is an observational distribution of $(\delta, Y, X)$ and, under $P$,} \\
		&\text{(i) $Y\Perp X$; (ii) $\pi(x) \equiv 1$; (iii) $X \sim  U[-1,1]$; (iv) $Y\sim N(\mu,1)$}.
		\bigg\}
	\end{aligned}
	\]
	Then $\mathcal{P}_{\rm par}$ is clearly a subset of $\mathcal{P}_{1}$. Under observational distributions in $\mathcal{P}_{\rm par}$, there is no missing responses and estimating $\mu$ becomes a normal mean estimation problem. According to theorem \ref{thm: sieve} of \cite{Chen2018Robust}, there are some universal constants $C_{\mathcal{P}_{1}}$, $\xi> 0$ such that
	\[\inf_{\hat{\mu}}\sup_{P\in \mathcal{P}_{\rm par}} P\left(|\hat{\mu} - \mu| \geq \frac{C_{\mathcal{P}_{1}}}{\sqrt{N}}\right)\geq \xi.\]
	This implies the result of this proposition because 
	\[\sup_{P\in \mathcal{P}_{1}} P\left(|\hat{\mu} - \mu| \geq \frac{C_{\mathcal{P}_{1}}}{\sqrt{N}}\right) \geq \sup_{P\in \mathcal{P}_{\rm par}} P\left(|\hat{\mu} - \mu| \geq \frac{C_{\mathcal{P}_{1}}}{\sqrt{N}}\right).\]	
\end{proof}

Let the basis functions $\{v_{k}(x\}_{k=1}^{\infty})$ be the tensor products of Chebyshev polynomials.
Similarly to \eqref{eq: minimax set1}, for constants $0 < M_{1}, M_{2} < 1$, $M_{3}, M_{4}, M_{5} > 0$ and $q > d/2$, we define the  
set of observational distributions that satisfy the regularity conditions imposed in Theorem \ref{thm: sieve}
\begin{equation}\label{eq: minimax set2}
	\begin{aligned}
		\mathcal{P}_{2} = \bigg\{P:\ & \text{$P$ is an observational distribution of $(\delta, Y, X)$ and under $P$,} \\
		&\text{(i) $M_{1}^{-1} \leq \sigma_{\rm min}(\Sigma) \leq \sigma_{\rm max}(\Sigma) \leq M_{1}$; (ii) $\inf_{x}\pi(x) \geq M_{2}$;}\\
		&\text{(iii) $E[(m(X))^{2}]\leq M_{3}$; (iv) $\sup_{x}\sigma^{2}(x) \leq M_{4}$;}\\
		&\text{(v) $\forall K$, $\exists\bar{\beta},\bar{\gamma}$ such that $\E[(m(X) - V_{K}(X)^{\T}\bar{\beta})^2] \leq M_{5} K^{-\frac{2q}{d}}$} \\
		&\text{\qquad and $\E[(\pi(X)^{-1} - V_{K}(X)^{\T}\bar{\gamma})^2]\leq M_{5} K^{-\frac{2q}{d}}$}
		\bigg\}.
	\end{aligned}
\end{equation}
Then similar minimax lower bound can be obtained by the arguments of Proposition \ref{prop: minimax}.

\subsection*{Semiparametric Efficiency Bound}
Semiparametric efficiency bound is a lower bound for the asymptotic variance of a regular asymptotically linear estimator in a semiparametric problem. See \cite{bickel1982adaptive} for more introductions and rigorous definitions. For the response mean estimation problem with missing responses, the semiparametric efficiency bound can be obtained according the results of \cite{hahn1998role}. Using the notations in this paper, the semiparametric efficiency bound is
\[\E\left[\frac{\sigma^{2}(X)}{\pi(X)}\right] + \var[m(X)].\]
Straightforward calculation can verify that $\var[\psi]$ in Theorem \ref{thm:kernel} and \ref{thm: sieve} equals to the semiparametric efficiency bound presented here.
\section*{Appendix B. Proof of Theorem \ref{thm:kernel}}
\label{app:theorem}
\begin{proof}
	Throughout this proof, we always give them the superscript ``$(l)$" for quantities defined using the data on the $l$-th machine. With some abuse of notation, we use the same subscripts $\{1,\dots, n\}$ to denote the observations on the $l$-th machine for $l=1,\dots,L$. 
	Note that
	\begin{equation}\label{eq: main}
		\begin{split}
			\sqrt{N}(\tilde{\mu}_{\mathbb{K}}-\mu)
			=&\sqrt{N}\frac{1}{L}\sum_{l=1}^{L}(\hat{\mu}^{(l)}_{\mathbb{K}}-\mu).
		\end{split}
	\end{equation}
	We begin the proof by decomposing $\hat{\mu}^{(l)}_{\mathbb{K}}$. Recall that $n = N/L$. By simple algebra, we have the following decomposition for $l = 1,2,\dots,L$
	\begin{equation}\label{eq: part}
		\begin{split}
			\hat{\mu}^{(l)}_{\mathbb{K}}
			=&\frac{1}{n}\sum^{n}_{i=1}\{\delta_iY_i+(1-\delta_i)m(X_i)\}\\
			&+\frac{1}{n}\sum^{n}_{i=1}(1-\delta_i)(\hat{m}^{(l)}_{\mathbb{K}}(X_i)-m(X_i))\\
			\eqqcolon&R^{(l)}+S^{(l)},
		\end{split}
	\end{equation}
	where
	\[\hat{m}^{(l)}_{\mathbb{K}}(x)=\frac{\sum^{n}_{i= 1}K_h(X_i-x)\delta_iY_i}{\sum^{n}_{i=1}K_h(X_i-x)\delta_i},\]
	and $K_h(\cdot) = h^{-d}K(\cdot/h )$.
	Since $R^{(l)}$ is a mean of i.i.d. random variables whose theoretical property is easy to derive, the main task is to study $S^{(l)}$. Let 
	\[\hat{t}^{(l)}(x) = \frac{1}{n}\sum^{n}_{i=1}K_h(X_i-x)\delta_iY_i\quad \text{and} \quad \hat{s}^{(l)}(x) = \frac{1}{n}\sum^{n}_{i=1}K_h(X_i-x)\delta_i.\]
	Define $t(x)=m(x)\pi(x)f(x)$ and $s(x)=\pi(x)f(x)$. Then
	\begin{equation}\label{eq: part 2}
		\begin{split}
			S^{(l)}
			&=\frac{1}{n}\sum^{n}_{i=1}(1-\delta_i)\left(\frac{\hat{t}^{(l)}(X_i)}{\hat{s}^{(l)}(X_i)}-\frac{t(X_i)}{s(X_i)}\right)\\
			&=\frac{1}{n}\sum^{n}_{i=1}(1-\delta_i)\left(\frac{\hat{t}^{(l)}(X_i) - t(X_i)}{s(X_i)}\right)\\ &\quad +\frac{1}{n}\sum^{n}_{i=1}(1-\delta_i)\hat{t}^{(l)}(X_i)\left(\frac{1}{\hat{s}^{(l)}(X_i)} - \frac{1}{s(X_i)}\right)\\
			&\eqqcolon S_{1}^{(l)} + S_{2}^{(l)}.
		\end{split}
	\end{equation}
	Note that 
	\begin{equation}\label{eq: U statistic}
		\begin{aligned}
			S_{1}^{(l)}
			&= \frac{1}{n^{2}}\sum_{i,j}\frac{1-\delta_{i}}{s(X_{i})}(\delta_{j}Y_{j}K_{h}(X_{j} - X_{i}) - t(X_{i}))\\
			&= \frac{1}{n^{2}}\sum_{i\neq j}^{n}\frac{1-\delta_{i}}{s(X_{i})}(\delta_{j}Y_{j}K_{h}(X_{j} - X_{i}) - t(X_{i}))\\
			&\quad + \frac{1}{n^{2}}\sum_{i = j}^{n}\frac{1-\delta_{i}}{s(X_{i})}(\delta_{j}Y_{j}K_{h}(X_{j} - X_{i}) - t(X_{i})) \\
			&\coloneqq U_{1}^{(l)} + V_{1}^{(l)},
		\end{aligned}
	\end{equation}
	and
	\[V_{1}^{(l)} = -\frac{1}{n^{2}}\sum_{i=1}^{n}\frac{1-\delta_{i}}{s(X_{i})}t(X_{i}).\]
	We have
	\begin{equation*}
		\E[|V_{1}^{(l)}|] = O\left(\frac{1}{n}\right)
	\end{equation*}
	by straightforward calculation.
	Denote 
	\[H_{ij} = \frac{1-\delta_{i}}{s(X_{i})}(\delta_{j}Y_{j}K_{h}(X_{j} - X_{i}) - t(X_{i})).\]
	Then
	\[U^{(l)}_{1} = \frac{1}{n^2}\sum_{i\neq j}^{n}H_{ij}\]
	is a U-statistic. Let
	\begin{equation}\label{eq: U decomp}
		W^{(l)}_{1} = U^{(l)}_{1} - \check{U}^{(l)}_{1},
	\end{equation}
	where
	\begin{equation}\label{eq: U hat}
		\check{U}^{(l)}_{1} = \E[H_{ij}]+\left\{\frac{1}{n}\sum^{n}_{i=1}\E[H_{ij}\mid Z_i]-\E[H_{ij}]\right\}+\left\{\frac{1}{n}\sum^{n}_{j=1}\E[H_{ij}\mid Z_j]-\E[H_{ij}]\right\}
	\end{equation}
	and $Z_i = (\delta_{i}, \delta_{i}Y_{i},X_{i})$. By some standard calculations in literature of nonparametric kernel regression,  we have
	\begin{equation}\label{eq: Hi}
		\E[H_{ij}\mid Z_i] = O(h^{q}),
	\end{equation}
	\begin{equation}\label{eq: Hj}
		\E[H_{ij}\mid Z_j] = \frac{\delta_j(1-\pi(X_j))Y_{j}}{\pi(X_{j})} - \E[(1-\delta)Y]+O(h^{q})
	\end{equation}
	and
	\begin{equation}\label{eq: Hij}
		\E[H_{ij}]=O(h^{q}).
	\end{equation}
	\eqref{eq: U hat}, \eqref{eq: Hi}, \eqref{eq: Hj} and \eqref{eq: Hij} imply
	\begin{equation}\label{eq: U hat 2}
		\check{U}^{(l)}_{1}
		= \frac{1}{n}\sum_{i=1}^{n}\frac{\delta_{i}(1-\pi(X_{i}))Y_{i}}{\pi(X_{i})} - \E[(1-\delta)Y] + O(h^{q}).
	\end{equation}
	By \eqref{eq: U statistic}, \eqref{eq: U decomp} and \eqref{eq: U hat 2}, we conclude
	\begin{equation}\label{eq: decomp S1}
		\begin{aligned}
			S_{1}^{(l)} 
			&=\frac{1}{n}\sum_{i=1}^{n}\frac{\delta_{i}(1-\pi(X_{i}))Y_{i}}{\pi(X_{i})} - \E[(1-\delta)Y] + W_{1}^{(l)}
			+ V_{1}^{(l)} + O(h^{q}).
		\end{aligned}
	\end{equation}
	It is clear that $\E[W_{1}^{(l)}]=0$. In addition, by standard U-statistic theory \citep[Chapter 3 in][]{Shao2003stat}, we have $\E[(W_{1}^{(l)})^{2}]=O(n^{-2}h^{-d})$. 
	
	Next, we move on to the term $S_{2}^{(l)}$.
	By Taylor's expansion, we have
	\begin{equation}
		\begin{aligned}\label{eq: decomp S2}
			S^{(l)}_{2}
			&=-\frac{1}{n}\sum_{i=1}^{n}(1-\delta_{i})\frac{\hat{t}^{(l)}(X_{i})}{s^{2}(X_{i})}(\hat{s}^{(l)}(X_{i}) - s(X_{i})) + \frac{1}{n}\sum_{i=1}^{n}(1 - \delta_{i})\frac{\hat{t}^{(l)}(X_{i})}{\tilde{s}_{i}^{3}}(\hat{s}^{(l)}(X_{i}) - s(X_{i}))^{2}\\
			&\eqqcolon -S_{2,1}^{(l)} + S_{2,2}^{(l)}
		\end{aligned}
	\end{equation}
	where $\tilde{s}_{i}$ is between $s(X_{i})$ and $\hat{s}^{(l)}(X_{i})$. 
	According to (C.4) and (C.6), by the similar arguments as in the proof of \eqref{eq: decomp S1}, we can decompose $S_{2,1}^{(l)}$ as
	\begin{equation}\label{eq: decomp S21}
		S_{2,1}^{(l)} = \frac{1}{n}\sum_{i=1}^{n}\frac{\delta_{i}(1-\pi(X_{i}))m(X_{i})}{\pi(X_{i})} - \E[(1-\delta)Y] + W_{2}^{(l)} + V_{2}^{(l)} + O(h^{q})	
	\end{equation}
	where $\E[W_{2}^{(l)}]=0$, $\E[(W_{2}^{(l)})^{2}]=O(n^{-2}h^{-2d})$ and $\E[|V_{2}^{(l)}|] = O(n^{-1}h^{-d})$. The ``$O$" are all uniform in $l$ since data are independent and identically distributed across machines.
	Combining \eqref{eq: part}, \eqref{eq: part 2}, \eqref{eq: decomp S1}, \eqref{eq: decomp S2} and \eqref{eq: decomp S21}, we have
	\begin{align*}
		\frac{1}{L}\sum_{l=1}^{L}\hat{\mu}_{\mathbb{K}}^{(l)} 
		&= 	\frac{1}{N}\sum_{i=1}^{N} \left\{\frac{\delta_{i}Y_{i}}{\pi(X_{i})} + \frac{\pi(X_{i}) -\delta_{i}}{\pi(X_{i})}m(X_{i})\right\} + \frac{1}{L}\sum_{l=1}^{L} W_{1}^{(l)} + \frac{1}{L}\sum_{l=1}^{L}W_{2}^{(l)} \\
		&\quad + \frac{1}{L}\sum_{l=1}^{L}V_{1}^{(l)} + \frac{1}{L}\sum_{l=1}^{L}V_{2}^{(l)} + \frac{1}{L}\sum_{l=1}^{L}S_{2,2}^{(l)} + O(h^{q}),
	\end{align*}
	\[
	\E\left[\left(\frac{1}{L}\sum_{l=1}^{L} W_{1}^{(l)}\right)^{2}\right] = O\left(\frac{L}{N^{2}h^{d}}\right),
	\]
	\[
	\E\left[\left(\frac{1}{L}\sum_{l=1}^{L} W_{2}^{(l)}\right)^{2}\right] = O\left(\frac{L}{N^{2}h^{2d}}\right),
	\]
	\[\E\left[\left|\frac{1}{L}\sum_{l=1}^{L}V_{1}^{(l)}\right|\right] = O\left(\frac{L}{N}\right),\]
	and 
	\[\E\left[\left|\frac{1}{L}\sum_{l=1}^{L}V_{2}^{(l)}\right|\right] = O\left(\frac{L}{Nh^{d}}\right).\]
	Thus
	\begin{equation}\label{eq: med}
		\begin{aligned}
			\frac{1}{L}\sum_{l=1}^{L}\hat{\mu}_{\mathbb{K}}^{(l)} 
			&= \frac{1}{N}\sum_{i=1}^{N} \left\{\frac{\delta_{i}Y_{i}}{\pi(X_{i})} + \frac{\pi(X_{i}) -\delta_{i}}{\pi(X_{i})}m(X_{i})\right\} + \frac{1}{L}\sum_{l=1}^{L}S_{2,2}^{(l)} \\
			&\quad+ O_{P}\left(\frac{L}{Nh^{d}}\right) + O(h^{q}).
		\end{aligned}
	\end{equation}
	So it remains to deal with the term $\sum_{l=1}^{L}S_{2,2}^{(l)}/L$. Let $s_{*} = \inf_{x}s(x)$. By (C.2) and (C.3), we have $s_{*} > 0$. For the constant $\tau$ in Lemma \ref{lem: sup exp}, define the event 
	\[E = \left\{\max_{l}\sup_{x}|\hat{s}^{(l)}(x) - s(x)| \leq \tau(a_{N} + h^{q})\right\},\]
	where $a_{N} = \sqrt{L \log N /(N h^{d})}$.
	For sufficiently large $N$, we have $\tilde{s}_{i} \geq 0.5s_{*}$ on event $\mathcal{E}$. Notice that
	\[
	\begin{aligned}
		S_{2,2}^{(l)}& = \frac{1}{n}\sum_{i=1}^{n}(1 - \delta_{i})\frac{\hat{t}^{(l)}(X_{i})}{\tilde{s}_{i}^{3}}(\hat{s}^{(l)}(X_{i}) - s(X_{i}))^{2}\\
		& = \frac{1}{n}\sum_{i=1}^{n}(1 - \delta_{i})\frac{1}{\tilde{s}_{i}^{3}}(\hat{t}^{(l)}(X_{i}) - t(X_{i}))(\hat{s}^{(l)}(X_{i}) - s(X_{i}))^{2} \\
		&\quad + \frac{1}{n}\sum_{i=1}^{n}(1 - \delta_{i})\frac{t(X_{i})}{\tilde{s}_{i}^{3}}(\hat{s}^{(l)}(X_{i}) - s(X_{i}))^{2}.
	\end{aligned}
	\]
	Thus on event $\mathcal{E}$, it holds that
	\begin{align*}
		|S_{2,2}^{(l)}| & \leq \frac{8}{s_{*}^{3}}\frac{1}{n}\sum_{i = 1}^{n}|\hat{t}^{(l)}(X_{i}) - t(X_{i})|(\hat{s}^{(l)}(X_{i}) - s(X_{i}))^{2} \\
		& \quad + \frac{8}{s_{*}^{3}}\frac{1}{n}\sum_{i=1}^{n}
		|t(X_{i})|(\hat{s}^{(l)}(X_{i}) - s(X_{i}))^{2}\\
		&\leq \frac{8}{s_{*}^{3}}\tau^{2}(a_{N} + h^{q})^{2}\frac{1}{n} \sum_{i=1}^{n}|\hat{t}^{(l)}(X_{i}) - t(X_{i})| \\
		&\quad + \frac{8}{s_{*}^{3}}\frac{1}{n}\sum_{i=1}^{n}
		|t(X_{i})|(\hat{s}^{(l)}(X_{i}) - s(X_{i}))^{2}.
	\end{align*}
	This implies
	\begin{equation}\label{eq: decomp S22}
		\begin{aligned}
			\frac{1}{L}\sum_{l=1}^{L}S_{2,2}^{(l)} &\leq \frac{8}{s_{*}^{3}}\tau^{2}(a_{N} + h^{q})^{2}\frac{1}{N} \sum_{i=1}^{N}|\hat{t}^{(l)}(X_{i}) - t(X_{i})|\\
			& \quad + \frac{8}{s_{*}^{3}}\frac{1}{N}\sum_{i=1}^{N}|t(X_{i})|(\hat{s}^{(l)}(X_{i}) - s(X_{i}))^{2}.
		\end{aligned}
	\end{equation}
	on event $\mathcal{E}$.
	By Jensen's inequality, we have
	\begin{equation}\label{eq: moment bound S22 term1}
		\begin{aligned}
			\E\left[\frac{1}{N} \sum_{i=1}^{N}|\hat{t}^{(l)}(X_{i}) - t(X_{i})|\right] 
			& = \E\left[|\hat{t}^{(l)}(X_{i}) - t(X_{i})|\right] \\
			& \leq \sqrt{\E\left[(\hat{t}^{(l)}(X_{i}) - t(X_{i}))^{2}\right]}.
		\end{aligned}
	\end{equation}
	According to (C.1), (C.4) and (C.5), we have
	\begin{equation*}
		\begin{aligned}
			& \E\left[(\hat{t}^{(l)}(X_{i}) - t(X_{i}))^{2}\right] \\
			& =  \E\left[\frac{1}{n^{2}}\sum_{j_{1}\neq j_{2}}^{n}(\delta_{j_{1}}Y_{j_{1}}K_{h}(X_{j_{1}}-X_{i}) -  t(X_{i}))(\delta_{j_{2}}Y_{j_{2}}K_{h}(X_{j_{2}}-X_{i}) - t(X_{i}))\right] \\
			& \quad + \E\left[\frac{1}{n^{2}}\sum_{j =1}^{n}(\delta_{j}Y_{j}K_{h}(X_{j}-X_{i}) - t(X_{i}))^{2}\right] \\
			& = O(h^{2q}) + O\left(\frac{1}{nh^{d}}\right).
		\end{aligned}
	\end{equation*}
	Combing this with \eqref{eq: moment bound S22 term1}, we have 
	\[\frac{1}{N} \sum_{i=1}^{N}|\hat{t}^{(l)}(X_{i}) - t(X_{i})| = O_{P}\left(h^{q} + \sqrt{\frac{L}{Nh^{d}}}\right).\]
	Then (C.6) implies
	\begin{equation}\label{eq: bound S22 term1}
		(a_{N} + h^{q})^{2}\frac{1}{N} \sum_{i=1}^{N}|\hat{t}^{(l)}(X_{i}) - t(X_{i})| = O_{P}\left(\frac{L}{Nh^{d}} + h^{2q}\right)
	\end{equation}
	By (C.1). we have $\sup_{x}|t(x)| \leq \infty$.
	Straightforward calculation can show
	\[
	\begin{aligned}
		&\E\left[\frac{1}{N}\sum_{i=1}^{N}|t(X_{i})|(\hat{s}^{(l)}(X_{i}) - s(X_{i}))^{2}\right] \\
		& =  \E\left[|t(X_{i})|(\hat{s}^{(l)}(X_{i}) - s(X_{i}))^{2}\right]\\
		& \leq \inf_{x}|t(x)| \E\left[(\hat{s}^{(l)}(X_{i}) - s(X_{i}))^{2}\right]\\
		& = O\left(\frac{L}{Nh^{d}} + h^{2q}\right).
	\end{aligned}
	\]
	Combing this with \eqref{eq: decomp S22} and \eqref{eq: bound S22 term1}, we have
	\begin{equation}\label{eq: bound S22}
		\left|\frac{1}{L}\sum_{l=1}^{L}S_{2,2}^{(l)}\right| \leq O_{P}\left(\frac{L}{Nh^{d}} + h^{2q}\right)
	\end{equation}
	on the event $\mathcal{E}$. In Lemma \ref{lem: sup exp}, we prove that $P(\mathcal{E}) \to 1$ under (C.1)--(C.6). Thus \eqref{eq: bound S22} implies
	\begin{equation}\label{eq: reminder}
		\frac{1}{L}\sum_{l=1}^{L}S_{2,2}^{(l)} = O_{P}\left(\frac{L}{Nh^{d}} + h^{2q}\right). 
	\end{equation}
	Combining \eqref{eq: med} and \eqref{eq: reminder}, we get
	\[\frac{1}{L}\sum_{l=1}^{L}\hat{\mu}_{\mathbb{K}}^{(l)} - \mu = \psi_{N} + O_{P}\left(\frac{L}{Nh^{d}}+ h^{q}\right).\]
	This equality implies \eqref{res: KDI AN} under Condition (C.7).
	The proof of Theorem \ref{thm:kernel} is then completed.
\end{proof}	

\begin{lemma}\label{lem: sup exp}
	Under (C.1)--(C.6), let $a_{N} = \sqrt{L \log N /(N h^{d})}$, then there is some universal constant $\tau > 0$ such that 
	\begin{equation*}
		P\left(\max_{l}\sup_{x}|\hat{s}^{(l)}(x) - s(x)| > \tau(a_{N} + h^{q}) \right) \to 0.
	\end{equation*}
\end{lemma}
\begin{proof}
	As in the main text, we use $C$ to denote a generic positive constant that may be different in different places.
	Recall that $\hat{s}^{(l)}(x) = \frac{1}{n}\sum^{n}_{i=1}K_h(X_{i}-x)\delta_{i}$.
	By some standard calculations in literature of nonparametric kernel regression, we have $\max_{l}\sup_{x}|\E[\hat{s}^{(l)}(x)] - s(x)| \leq \tau h^{q}$ for $\tau$ above a certain threshold. Hence
	\begin{equation}\label{eq: centralize}
		\left\{\max_{l}\sup_{x}|\hat{s}^{(l)}(x) - s(x)| > \tau(a_{N} + h^{q})\right\} \subset
		\left\{\max_{l}\sup_{x}|\hat{s}^{(l)}(x) - \E[\hat{s}^{(l)}(x)]| > \tau a_{N}\right\}.
	\end{equation}
	Under (C.5), we have $|K(x) - K(y)| \leq M\|x-y\|$ for some positive constant $M$. Hence $|\hat{s}^{(l)}(x) - \hat{s}^{(l)}(x)|\leq Mh^{-d-1}\|x - y\|$ and $|\E[\hat{s}^{(l)}(x)] - \E[\hat{s}^{(l)}(x)]|\leq Mh^{-d-1}\|x - y\|$. 
	By (C.3), $X$ has a bounded support $\mathscr{X}$. By standard results on the covering number of the bounded set, there exist $x_{1},\dots, x_{B}$ where $B\leq C\tau^{-d}a_{N}^{-d}h^{-d(d+1)}$ such that $\forall x \in \mathscr{X}$, $\exists b\in\{1,\dots,B\}$, $\|x - x_{b}\|\leq \tau a_{N}h^{d+1}/(3M)$. Hence 
	\[
	\left\{\max_{l}\sup_{x}|\hat{s}^{(l)}(x) - \E[\hat{s}^{(l)}(x)]| > \tau a_{N}\right\}
	\subset
	\left\{\max_{l}\max_{b}|\hat{s}^{(l)}(x_{b}) - \E[\hat{s}^{(l)}(x_{b})]| > \frac{\tau a_{N}}{3}\right\}.
	\]
	Combining this with \eqref{eq: centralize}, we have
	\begin{equation}\label{eq: induced set}
		\left\{\max_{l}\sup_{x}|\hat{s}^{(l)}(x) - s(x)| > \tau(a_{N} + h^{q})\right\} \subset
		\left\{\max_{l}\max_{b}|\hat{s}^{(l)}(x_{b}) - \E[\hat{s}^{(l)}(x_{b})]| > \frac{\tau a_{N}}{3}\right\}
	\end{equation}
	for sufficiently large $N$. By (C.5), we have $\bar{K}= \sup_{x}K(x) < \infty$. Thus $K_h(X_{i}-x)\delta_{i} \leq \bar{K} h^{-d}$. Moreover, $\var[K_h(X_{i}-x)\delta_{i}] \leq \E[(K_h(X_{i}-x)\delta_{i})^{2}] \leq C h^{-d}$. By Bernstein inequality \citep{wainwright2019} for bounded random variables and the union bound, we have
	\begin{equation}\label{eq: prob bound}
		\begin{aligned}
			&P\left(\max_{l}\max_{b}|\hat{s}^{(l)}(x_{b}) - \E[\hat{s}^{(l)}(x_{b})]| > \frac{\tau a_{N}}{3}\right)\\
			&\leq 2L B\exp\left(-\frac{\tau^{2}a_{N}^{2}nh^{d}}{18C}\right)
			\leq  2\exp\left(-\frac{\tau^{2}\log N}{18C}  + \log L + \log B \right).
		\end{aligned}
	\end{equation}
	for sufficiently large $N$ under (C.6).
	It is not hard to verify that the rightmost side of the inequality \eqref{eq: prob bound} converges to zero under (C.6) if we take $\tau$ to be a sufficiently large constant. This completes the proof.
\end{proof}

\section*{Appendix C. Proof of Theorem \ref{thm: sieve}}
\begin{proof}
	We define $\beta^{*} = \Sigma^{-1}\E [\delta V_{K}(X)Y]$ and $\gamma^{*} = \Sigma^{-1}\eta$ where
	$\eta = \E[\delta V_{K}(X)/\pi(X)] = \E[V_{K}(X)]$. Let $m_{\mathbb{S}}^{*}(x) = V_{K}(x)^{\T}\beta^{*}$ and $w_{\mathbb{S}}^{*}(x) = V_{K}(x)^{\T}\gamma^{*}$. Clearly, $\E[\delta(m(X) - m_{\mathbb{S}}^{*}(X))^{2}] = \min_{\beta}\E[\delta(m(X) - V_{K}(X)^{\T}\beta)^2]$ and $\E[\delta(1/\pi(X) - w_{\mathbb{S}}^{*}(X))^{2}] = \min_{\gamma}\E[\delta(1/\pi(X) - V_{K}(X)^{\T}\gamma)^2]$. Then by Condition (C.9)(ii), we have 
	\begin{equation} \label{eq: obs error m}
		\begin{aligned}
			\E[\delta(m(X) - m_{\mathbb{S}}^{*}(X))^{2}]& = \min_{\beta}\E[\delta(m(X) - V_{K}(X)^{\T}\beta)^2]\\
			&\leq \min_{\beta}\E[(m(X) - V_{K}(X)^{\T}\beta)^2]\\
			&\leq \E[(m(X) - V_{K}(X)^{\T}\bar{\beta})^2]\\
			&\leq C K^{-2r},
		\end{aligned}
	\end{equation} and 
	\begin{equation}\label{eq: obs error pi}
		\begin{aligned}
			\E\left[\delta\left(\frac{1}{\pi(X)} - w_{\mathbb{S}}^{*}(X)\right)^{2}\right] &= \min_{\gamma}\E\left[\delta\left(\frac{1}{\pi(X)} - V_{K}(X)^{\T}\gamma\right)^2\right]\\
			&\leq \min_{\gamma}\E\left[\left(\frac{1}{\pi(X)} - V_{K}(X)^{\T}\gamma\right)^2\right]\\
			&\leq \E\left[\left(\frac{1}{\pi(X)} - V_{K}(X)^{\T}\bar{\gamma}\right)^2\right]\\
			&\leq C K^{-2r}.
		\end{aligned}
	\end{equation}
	
	Let $\hat{m}_{\mathbb{S}}(x) = V_{K}(x)^{\T}\hat{\beta}$ and $\tilde{m}_{\mathbb{S}}(x) = V_{K}(x)^{\T}\tilde{\beta}$. By the definition of $\hat{\beta}$ and $v_{1}(x)\equiv 1$, we have $(\sum_{i=1}^{N}\delta_{i}Y_{i} - \sum_{i=1}^{N}\delta_{i}V_{K}(X_i)^{\T}\hat{\beta}) / N = 0$. Thus
	\[
	\hat{\mu}_{\mathbb{S}} = \frac{1}{N}\sum_{i=1}^{N}(\delta_{i}Y_{i} + (1 - \delta_{i})\hat{m}_{\mathbb{S}}(X_i))
	= \frac{1}{N}\sum_{i=1}^{N}\hat{m}_{\mathbb{S}}(X_i)
	\]
	and 
	\[\tilde{\mu}_{\mathbb{S}} - \hat{\mu}_{\mathbb{S}} = \frac{1}{N}\sum_{i=1}^{N}(1 - \delta_{i})V_{K}(X_{i})^{\T}(\hat{\beta} - \tilde{\beta}) \leq \zeta_{K}\|\hat{\beta} - \tilde{\beta}\|.\]
	Thus to prove the asymptotic normality in theorem \ref{thm: sieve}, it suffices to prove 
	\begin{equation}\label{eq: res1}
		\sqrt{N}(\hat{\mu}_{\mathbb{S}}-\mu)\stackrel{d}{\rightarrow}N(0,\var[\psi])
	\end{equation} 
	and 
	\begin{equation}\label{eq: res2}
		\|\hat{\beta} - \tilde{\beta}\| = o_{\P}\left((\zeta_{K}\sqrt{N})^{-1}\right).
	\end{equation}
	The convergence rate result in Theorem \ref{thm: sieve} can be obtained in the proof of the asymptotic normality result.
	
	\subsection*{Proof of \eqref{eq: res1}}
	Let $F(\cdot)$ be the distribution function of $X$. Note that $\mu = \int m(x)dF(x)$, thus
	we have the following decomposition of $\sqrt{N}(\hat{\mu}_{\mathbb{S}}-\mu)$,
	\begin{align*}
		\sqrt{N}(\hat{\mu}_{\mathbb{S}}-\mu) 
		&=\sqrt{N}\left(\frac{1}{N}\sum_{i=1}^{N}(\hat{m}_{\mathbb{S}}(X_{i}) - m_{\mathbb{S}}^{*}(X_{i})) - \int(\hat{m}_{\mathbb{S}}(x) - m_{\mathbb{S}}^{*}(x))dF(x)\right) \\
		&+ \sqrt{N}\left(\frac{1}{N}\sum_{i=1}^{N}(m_{\mathbb{S}}^{*}(X_{i}) - m(X_{i})) - \int(m_{\mathbb{S}}^*(x) - m(x))dF(x)\right) \\
		&+ \sqrt{N}\int (\hat{m}_{\mathbb{S}}(x) - m_{\mathbb{S}}^{*}(x))dF(x) + \sqrt{N}\int (m_{\mathbb{S}}^{*}(x) - m(x))dF(x) \\
		&+ \sqrt{N}\left(\frac{1}{N}\sum_{i=1}^{N}m(X_{i}) - \mu\right)\\
		&\eqqcolon R_1 + R_2 + R_3 + R_4 + R_{5}.
	\end{align*}
	
	First, by Lemma \ref{lem: est error} and Condition (C.11)(i), we have
	\begin{align*}
		|R_{1}| & = \left|\sqrt{N}\left(\frac{1}{N}\sum_{i=1}^{N}V_{K}(X_{i})^{\T} - \eta^{\T}\right)(\hat{\beta} - \beta^{*})\right|  \\
		& \leq \left\|\frac{1}{\sqrt{N}}\sum_{i=1}^{N}V_{K}(X_{i}) - \eta\right\|\|\hat{\beta} - \beta^{*}\| \\
		& = O_{\P}\left(\sqrt{\E[\|V_{K}(X_{i}) - \eta\|^2]}\|\hat{\beta} - \beta^{*}\|\right)\\
		& = O_{\P}(\zeta_{K}\|\hat{\beta} - \beta^{*}\|)\\
		& = O_{\P}\left(\sqrt{\frac{\zeta_{K}^4}{N}}\right)\\
		& = o_{\P}(1).
	\end{align*}
	By Condition (C.2) and \eqref{eq: obs error m}, we have
	\begin{align*}
		E[R_{2}^{2}] &= \var[m_{\mathbb{S}}^{*}(X) - m(X)] \\
		&\leq \E[(m_{\mathbb{S}}^{*}(X) - m(X))^2] \\
		&= \E\left[\frac{\delta}{\pi(X)}(m_{\mathbb{S}}^{*}(X) - m(X))^2\right] \\
		& \leq C\E\left[\delta(m_{\mathbb{S}}^{*}(X) - m(X))^2\right] \\
		&  \leq C^{2}K^{-2r}.
	\end{align*}
	This implies $R_{2} = O_{P}(K^{-r}) = o_{P}(1)$. 
	
	For $R_{3}$, we have
	\begin{align*}
		R_{3} = &\sqrt{N}\int (\hat{m}_{\mathbb{S}}(x) - m_{\mathbb{S}}^{*}(x))dF(x)\\
		& = \sqrt{N}\eta^{\T}(\hat{\beta} - \beta^{*}) \\
		& = \frac{1}{\sqrt{N}}\sum_{i=1}^{N}\eta^{\T}\hat{\Sigma}^{-1} \delta_{i}V_{K}(X_{i})(Y_{i} - m_{\mathbb{S}}^{*}(X_{i})).
	\end{align*}
	Note that
	\begin{equation}\label{eq: decomposition coef}
		\begin{aligned}
			\|\hat{\Sigma}^{-1}\eta - \Sigma^{-1}\eta\|
			& =  \|\hat{\Sigma}^{-1}(\eta - \hat{\Sigma}\Sigma^{-1}\eta)\| \\
			& \leq \|\hat{\Sigma}^{-1}\|\|\eta - \hat{\Sigma}\Sigma^{-1}\eta\| 
		\end{aligned}
	\end{equation}
	and
	\begin{equation}\label{eq: linear coef}
		\begin{aligned}
			\eta - \hat{\Sigma}\Sigma^{-1}\eta & = \frac{1}{N}\sum_{i=1}^{N}(\eta - \delta_{i}V_{K}(X_{i})V_{K}(X_{i})^{\T}\Sigma^{-1}\eta).
		\end{aligned}
	\end{equation}
	Straightforward calculation can show
	\begin{equation}\label{eq: coef expect}
		\E\left[\eta - \delta_{i}V_{K}(X_{i})V_{K}(X_{i})^{\T}\Sigma^{-1}\eta\right] = 0
	\end{equation}
	and 
	\[
	\begin{aligned}
		\E\left[\left\|\eta - \delta_{i}V_{K}(X_{i})V_{K}(X_{i})^{\T}\Sigma^{-1}\eta\right\|^{2}\right] &\leq \E\left[\left\|
		\delta_{i}V_{K}(X_{i})V_{K}(X_{i})^{\T}\Sigma^{-1}\eta\right\|^{2}\right] \\
		& = \E\left[\delta \|V_{K}(X)\|^{2}\eta^{\T}\Sigma^{-1}V_{K}(X)V_{K}(X)^{\T}\Sigma^{-1}\eta\right] \\
		& \leq \zeta_{K}^{2}\eta^{\T}\Sigma^{-1}\eta \\
		& = \zeta_{K}^{2}\gamma^{*\T}\Sigma\gamma^{*} \\
		& = \zeta_{K}^{2}E[(w_{\mathbb{S}}^{*}(X))^{2}] \\
		& \leq \zeta_{K}^{2}E[(\pi(X))^{-2}].
	\end{aligned}
	\]
	Then Condition (C.2) implies $E[(\pi(X))^{-2}]$ is bounded and hence 
	\[
	\E\left[\left\|\eta - \delta_{i}V_{K}(X_{i})V_{K}(X_{i})^{\T}\Sigma^{-1}\eta\right\|^{2}\right] = O(\zeta_{K}^{2}).
	\]
	Combining this with \eqref{eq: coef expect} and \eqref{eq: linear coef}, we have
	\begin{equation}\label{eq: bound coef}
		E[\|\eta - \hat{\Sigma}\Sigma^{-1}\eta\|^{2}] = O\left(\frac{\zeta_{K}^{2}}{N}\right)
	\end{equation}
	Lemma \ref{lem: est error} and Conditions (C.8), (C.10) implies $\hat{\Sigma}^{-1} = O_{P}(1)$. This together with \eqref{eq: decomposition coef} and \eqref{eq: bound coef} implies $ \|\hat{\Sigma}^{-1}\eta - \Sigma^{-1}\eta\| = O_{P}\left(\sqrt{\zeta_{K}^{2}/N}\right)$.
	Then because
	\[\left\|\frac{1}{\sqrt{N}}\sum_{i=1}^{N}\delta_{i}V_{K}(X_{i})(Y_{i} - m_{\mathbb{S}}^{*}(X_{i}))\right\|
	= O_{\P}\left(\sqrt{\E[\delta\|V_{K}(X)\|^{2}(Y - m_{\mathbb{S}}^{*}(X))^{2}]}\right) = O_{\P}(\zeta_{K}),\]
	we have
	\begin{align*}
		&\left|R_{3} -  \frac{1}{\sqrt{N}}\sum_{i=1}^{N}\eta^{\T}\Sigma^{-1} \delta_{i}V_{K}(X_{i})(Y_{i} - m_{\mathbb{S}}^{*}(X_{i}))\right|\\
		&\leq \|\hat{\Sigma}^{-1}\eta - \Sigma^{-1}\eta\|
		\left\|\frac{1}{\sqrt{N}}\sum_{i=1}^{N}\delta_{i}V_{K}(X_{i})(Y_{i} - m_{\mathbb{S}}^{*}(X_{i}))\right\| \\
		& = O_{\P}\left(\sqrt{\frac{\zeta_{K}^4}{N}}\right)\\
		& = o_{\P}(1)
	\end{align*}
	according to Condition (C.11)(i).
	Moreover, by Conditions (C.4), (C.9)(ii) and (C.11)(i), we have 
	\[
	\zeta_{K}^{2}K^{-2r} \leq \frac{1}{2}\left(\frac{\zeta_{K}^{4}}{N} + NK^{-4r} \right) \to 0,
	\]
	and hence
	\begin{align*}
		&\E\left[\left(\frac{1}{\sqrt{N}}\sum_{i=1}^{N}\eta^{\T}\Sigma^{-1} \delta_{i}V_{K}(X_{i})(Y_{i} - m_{\mathbb{S}}^{*}(X_{i})) - \frac{1}{\sqrt{N}}\sum_{i=1}^{N}\frac{\delta_i}{\pi(X_{i})}(Y_{i} - m(X_{i}))\right)^{2}\right]\\
		& = \E\left[\left(\delta w_{\mathbb{S}}^{*}(X)(Y - m_{\mathbb{S}}^{*}(X)) - \frac{\delta}{\pi(X)}(Y - m(X))\right)^{2}\right]\\
		& = \E\left[\left\{\left(\delta w_{\mathbb{S}}^{*}(X) - \frac{\delta}{\pi(X)}\right)(Y - m(X)) + \delta w_{\mathbb{S}}^{*}(X)(m(X) - m_{\mathbb{S}}^{*}(X))\right\}^{2}\right]\\
		&\leq 2 \E\left[\left(\delta w_{\mathbb{S}}^{*}(X) - \frac{\delta}{\pi(X)}\right)^{2}(Y - m(X))^{2}\right] + 2\E\left[\left\{\delta w_{\mathbb{S}}^{*}(X)(m(X) - m_{\mathbb{S}}^{*}(X))\right\}^{2}\right] \\
		&\leq C\left\{ \E\left[\left(\delta w_{\mathbb{S}}^{*}(X) - \frac{\delta}{\pi(X)}\right)^{2}\right] + \zeta_{K}^{2}\E\left[(m(X) - m_{\mathbb{S}}^{*}(X))^{2}\right]\right\} \\
		&\leq C K^{-2r} + C\zeta_{K}^{2}K^{-2r} \to 0.
	\end{align*}
	Thus
	\begin{align*}
		&\frac{1}{\sqrt{N}}\sum_{i=1}^{N}\eta^{\T}\Sigma^{-1} \delta_{i}V_{K}(X_{i})(Y_{i} - m_{\mathbb{S}}^{*}(X_{i})) - \frac{1}{\sqrt{N}}\sum_{i=1}^{N}\frac{\delta_i}{\pi(X_{i})}(Y_{i} - m(X_{i})) \\
		&= O_{\P}\left(\sqrt{\E\left[\left(\frac{1}{\sqrt{N}}\sum_{i=1}^{N}\eta^{\T}\Sigma^{-1} \delta_{i}V_{K}(X_{i})(Y_{i} - m_{\mathbb{S}}^{*}(X_{i})) - \frac{1}{\sqrt{N}}\sum_{i=1}^{N}\frac{\delta_i}{\pi(X_{i})}(Y_{i} - m(X_{i}))\right)^2\right]}\right)\\
		&= O_{P}(\zeta_{K}K^{-r}) \\
		& = o_{\P}(1)
	\end{align*}
	according to Condition (C.11)(i).
	
	For the term $R_{4}$, we have
	\begin{align*}
		R_{4}& = \sqrt{N}\int (m_{\mathbb{S}}^{*}(x) - m(x))dF(x)\\
		& = \sqrt{N}\E \left[\frac{\delta}{\pi(X)}(m_{\mathbb{S}}^{*}(X) - m(X))\right] \\
		& = \sqrt{N}\E\left[\delta\left(\frac{1}{\pi(X)} - w_{\mathbb{S}}^{*}(X)\right)(m_{\mathbb{S}}^{*}(X) - m(X))\right] + \sqrt{N}\E[\delta w_{\mathbb{S}}^{*}(X)(m_{\mathbb{S}}{*}(X) - m(X))].
	\end{align*}
	By the definition of $\beta^{*}$ and $w_{\mathbb{S}}^{*}(x)$, we have
	\[\E[\delta w_{\mathbb{S}}^{*}(X)(m_{\mathbb{S}}{*}(X) - m(X))] = \E[\delta \gamma^{*\T}V_{K}(X)(V_{K}(X)^{\T}\beta^{*} - Y)] = 0.\]
	Thus by Conditions (C.9)(ii) and (C.11)(i), we have
	\begin{align*}
		|R_{4}| & = \left|\sqrt{N}\E\left[\delta\left(\frac{1}{\pi(X)} - w_{\mathbb{S}}^{*}(X)\right)(m_{\mathbb{S}}^{*}(X) - m(X))\right]\right| \\ 
		& \leq \sqrt{N}\sqrt{\E\left[\delta\left(\frac{1}{\pi(X)} - w_{\mathbb{S}}^{*}(X)\right)^{2}\right]}\sqrt{\E[\delta(m_{\mathbb{S}}^{*}(X) - m(X))^{2}]} \\
		& \leq \sqrt{N}K^{-2r} \to 0.
	\end{align*} 
	So far, we have proved 
	\[
	\begin{aligned}
		R_{1} + R_{2} + R_{3} + R_{4} &= \frac{1}{\sqrt{N}}\sum_{i=1}^{N}\left(\frac{\delta_i}{\pi(X_i)}(Y_i-m(X_i)) + m(X_i) - \mu\right) \\
		&\quad + O_{P}\left(\sqrt{\frac{\zeta_{K}^4}{N}} + \zeta_{K}K^{-r}+\sqrt{N}K^{-2r}\right) \\
		& = \sqrt{N}\psi_{N} + O_{P}\left(\sqrt{\frac{\zeta_{K}^4}{N}}+\sqrt{N}K^{-2r}\right),
	\end{aligned}
	\]
	where the second equality is due to $2\zeta_{K}K^{-r} \leq \sqrt{\zeta_{K}^4/N}+\sqrt{N}K^{-2r}$.
	Then under Condition (C.11)(i), \eqref{eq: res1} follows from the central limit theorem.
	\subsection*{Proof of \eqref{eq: res2}}
	For $t = 1,\dots, T$,
	\begin{align*}
		\beta_{t} 
		& = \beta_{t-1} + (\tilde{\Sigma} + \alpha I)^{-1}(\hat{\Gamma} - \hat{\Sigma}\beta_{t-1})\\
		& = \beta_{t-1} + (\tilde{\Sigma} + \alpha I)^{-1}\hat{\Sigma}(\hat{\beta} - \beta_{t-1}).
	\end{align*}
	Then we have
	\begin{align*}
		\beta_{t} - \hat{\beta} 
		& = (I - (\tilde{\Sigma} + \alpha I)^{-1}\hat{\Sigma})(\beta_{t-1} - \hat{\beta}) \\
		& = (I - (\tilde{\Sigma} - \hat{\Sigma} + \hat{\Sigma} + \alpha I)^{-1}\hat{\Sigma})(\beta_{t-1} - \hat{\beta}) \\
		& = (I - ((\hat{\Sigma} + \alpha I)^{-1}(\tilde{\Sigma} - \hat{\Sigma}) + I)^{-1} (I + \alpha\hat{\Sigma}^{-1})^{-1})(\beta_{t-1} - \hat{\beta}).
	\end{align*}
	
	Let $A = (\hat{\Sigma} + \alpha I)^{-1}(\tilde{\Sigma} - \hat{\Sigma})$ and $H = 
	((\hat{\Sigma} + \alpha I)^{-1}(\tilde{\Sigma} - \hat{\Sigma}) + I)^{-1} (I + \alpha\hat{\Sigma}^{-1})^{-1}$, then 
	\[
	\|\beta_{t} - \hat{\beta}\| \leq \max\{1 - \lambda_{\rm min}(H), \lambda_{\rm max}(H) - 1\}\|\beta_{t-1} - \hat{\beta}\|.
	\]
	
	If $\|A\| < 1/2$, by the relationships
	\[(I + A)^{-1} = I - A + (I + A)^{-1}A^2\]
	and
	\[\|(I + A)^{-1}\| \leq (1 - \|A\|)^{-1},\]
	we have
	\[\|(I+A)^{-1} - I\| \leq \|A\| + \|A\|^{2}(1 - \|A\|)^{-1} \eqqcolon h(\|A\|) < 1,\]
	\[(1 - h(\|A\|)) I \preceq (I + A)^{-1} \preceq (1 + h(\|A\|))I,\]
	and
	\[(1 - h(\|A\|)) (I + \alpha \hat{\Sigma}^{-1})^{-1} \preceq H \preceq (1 + h(\|A\|))(I + \alpha \hat{\Sigma}^{-1})^{-1}.\]
	Here for two symmetric matrices $A$ and $B$, $A \preceq B$ means $B - A$ is positive semi-definite. Then
	\[\lambda_{\rm min}(H) \geq \frac{1 - h(\|A\|)}{1 + \alpha \lambda_{\rm min}^{-1}(\hat{\Sigma})}, \qquad \lambda_{\rm max}(H) \leq \frac{1 + h(\|A\|)}{1 + \alpha\lambda_{\rm max}^{-1}(\hat{\Sigma})}.\]
	Thus
	\[\max\{1 - \lambda_{\rm min}(H), \lambda_{\rm max}(H) - 1\} \leq \max\left\{\frac{\alpha\lambda_{\rm min}^{-1}(\hat{\Sigma}) + h(\|A\|)}{1 + \alpha \lambda_{\rm min}^{-1}(\hat{\Sigma})}, \frac{h(\|A\|) - \alpha\lambda_{\rm max}^{-1}(\hat{\Sigma})}{1 + \alpha \lambda_{\rm max}^{-1}(\hat{\Sigma})}\right\}.\]
	By Lemma \ref{lem: est error}, we have
	\[\|\tilde{\Sigma} - \Sigma\| = O_{\P}\left(\sqrt{\frac{\zeta_{K}^{2}\log K}{n}}\right)\]
	and
	\[\|\hat{\Sigma} - \Sigma\| = O_{\P}\left(\sqrt{\frac{\zeta_{K}^{2}\log K}{N}}\right).\]
	Because $\alpha \asymp \log^2 K \sqrt{\frac{\zeta_{K}^{2}}{n}}$, we have
	\begin{align*}
		\|A\| & \leq \frac{\|\tilde{\Sigma} - \hat{\Sigma}\|}{\alpha + \lambda_{\rm min}(\hat{\Sigma})}\\
		& \leq \frac{\|\tilde{\Sigma} - \Sigma\| + \|\hat{\Sigma} - \Sigma\|}{\alpha + \lambda_{\rm min}(\Sigma) - \|\hat{\Sigma} - \Sigma\|}\\
		& \leq \frac{\|\tilde{\Sigma} - \Sigma\| + \|\hat{\Sigma} - \Sigma\|}{\alpha + C_{\rm L} - \|\hat{\Sigma} - \Sigma\|} \\
		& = o_{\P}\left(\min\left\{\alpha, \frac{1}{\log K}\right\}\right)
	\end{align*}
	and $h(\|A\|) = o_{\P}(\min\{\alpha, 1/\log K\})$.
	Under Condition (C.10), by Lemma \ref{lem: est error}, $\lambda_{\rm min}(\hat{\Sigma}) \geq 1/2 \lambda_{\rm min}(\Sigma)$ and $\lambda_{\rm max}(\hat{\Sigma}) \leq 2 \lambda_{\rm max}(\Sigma)$ with probability tending to 1. Thus the inequalities
	\[\lambda_{\rm max}(H) - 1 \leq 0\]
	and
	\begin{equation}\label{eq: contraction}
		\max\{1 - \lambda_{\rm min}(H), \lambda_{\rm max}(H) - 1\} \leq 1 - \frac{1 - \min\left\{\alpha, \frac{1}{\log K}\right\}}{1 + 2\alpha\lambda_{\rm min}(\Sigma)} \leq C_{N,K}
	\end{equation}
	hold with probability tending to 1.
	By Condition (C.8) and the definition of $\beta^{*}$, we have
	\begin{equation}\label{eq: bound beta*}
		\begin{aligned}
			\|\beta^{*}\|^{2}
			&\leq C_{\rm L}^{-1}\beta^{*\T}\Sigma\beta^{*} \\
			& = C_{\rm L}^{-1}\E[\delta\beta^{*\T} V_{K}(X)V_{K}(X)^{\T}\beta^{*}]\\
			& = C_{\rm L}^{-1}\E[\delta(m_{\mathbb{S}}^{*}(X))^2]\\
			& \leq C_{\rm L}^{-1}\E[\delta(m(X))^2] \\
			& \leq  C_{\rm L}^{-1}\E[(m(X))^2].
		\end{aligned}
	\end{equation}
	Hence  Condition (C.9)(i) implies $\|\beta^{*}\|$ is bounded and we have $\|\hat{\beta}\| = O_{\P}(1)$ according to Condition (C.10) and Lemma \ref{lem: est error}. Combining this with \eqref{eq: contraction}, we have $\|\tilde{\beta} - \hat{\beta}\| = O_{\P}(C_{N,K}^{T})$.
	So far, the result \eqref{res: sieve rate} in Theorem \ref{thm: sieve} has been established. Under Condition (C.11)(ii), we have $0.5\log N + \log \zeta_{K} -T\log C_{N,K}\to -\infty$ and hence $\|\tilde{\beta} - \hat{\beta}\| = o_{\P}\left((\zeta_{K}\sqrt{N})^{-1}\right)$. This proves \eqref{eq: res2}, which completes the proof.
\end{proof}

\begin{lemma}\label{lem: est error}
	Under Conditions (C.8) and (C.10), we have
	\[
	\|\hat{\Sigma} - \Sigma\| = O_{\P}\left(\sqrt{\frac{\zeta_{K}^{2}\log K}{N}}\right)
	\]
	and
	\[
	\|\hat{\beta} - \beta^{*}\| = O_{\P}\left(\sqrt{\frac{\zeta_{K}^2}{N}}\right).
	\]
\end{lemma}
\begin{proof}
	Under Conditions (C.8) and (C.10), we have $E[\|\hat{\Sigma} - \Sigma\|] = O\left(\sqrt{\zeta_{K}^{2}\log K/N}\right)$ according to Lemma 6.2 in \cite{belloni2015some}. This implies
	\[
	\|\hat{\Sigma} - \Sigma\| = O_{\P}\left(\sqrt{\frac{\zeta_{K}^{2}\log K}{N}}\right).
	\]
	The second statement of the lemma follows from
	\begin{align*}
		\|\hat{\beta} - \beta^{*}\| &= \left\|\hat{\Sigma}^{-1}\frac{1}{N}\sum_{i=1}^{N}\delta_{i}V_{K}(X_{i})(Y_{i} - m_{\mathbb{S}}^{*}(X_{i}))\right\|\\
		& \leq \|\hat{\Sigma}^{-1}\|O_{\P}\left(\sqrt{\frac{1}{N}\E[\|V_{K}(X_{i})\|^{2}(Y - m_{\mathbb{S}}(X))^{2}]}\right) \\
		& = O_{\P}\left(\sqrt{\frac{\zeta_{K}^2}{N}}\right).
	\end{align*}
\end{proof}

\section*{Appendix C. Additional Simulation Results}
\subsection*{Ablation Study of the DWCV Algorithm}
In this section, we conducted an ablative study to evaluate the effect of weights in the DWCV algorithm. For simplicity, we use ``DCV" to indicate that the weight in the DWCV algorithm is replaced by the uniform weight.
We adopt the simulation settings in Section \ref{subsec: sim tuning}, and compare the RMSE and computing time of the proposed estimators with the constant selected by the DCV and DWCV algorithms.
The simulation results are presented in Table \ref{table: ablation}.

\centerline{[Insert Table \ref{table: ablation} about here.]}

Table \ref{table: ablation} shows that the additional time due to the calculation of the weights is almost negligible for all estimators in all settings. The accuracy of the estimators based on the DWCV algorithm is never worse than that based on the DCV algorithm. The DWCV algorithm can improve the accuracy of the KDI estimators compared to the DCV algorithm. The RMSE of the KDI estimator based on the DCV algorithm is more than $3.5$ times larger than that of the KDI estimator based on the DWCV algorithm when $d=5$ in the linear setting or $d=15$ in the nonlinear setting.



%
%


\vskip 0.2in
\bibliography{dc}
\newpage

\begin{figure}[H]
	\centering
	\subfigure[Kernel-based SGM estimators with $d=5$.]{
		\includegraphics[scale = 0.45]{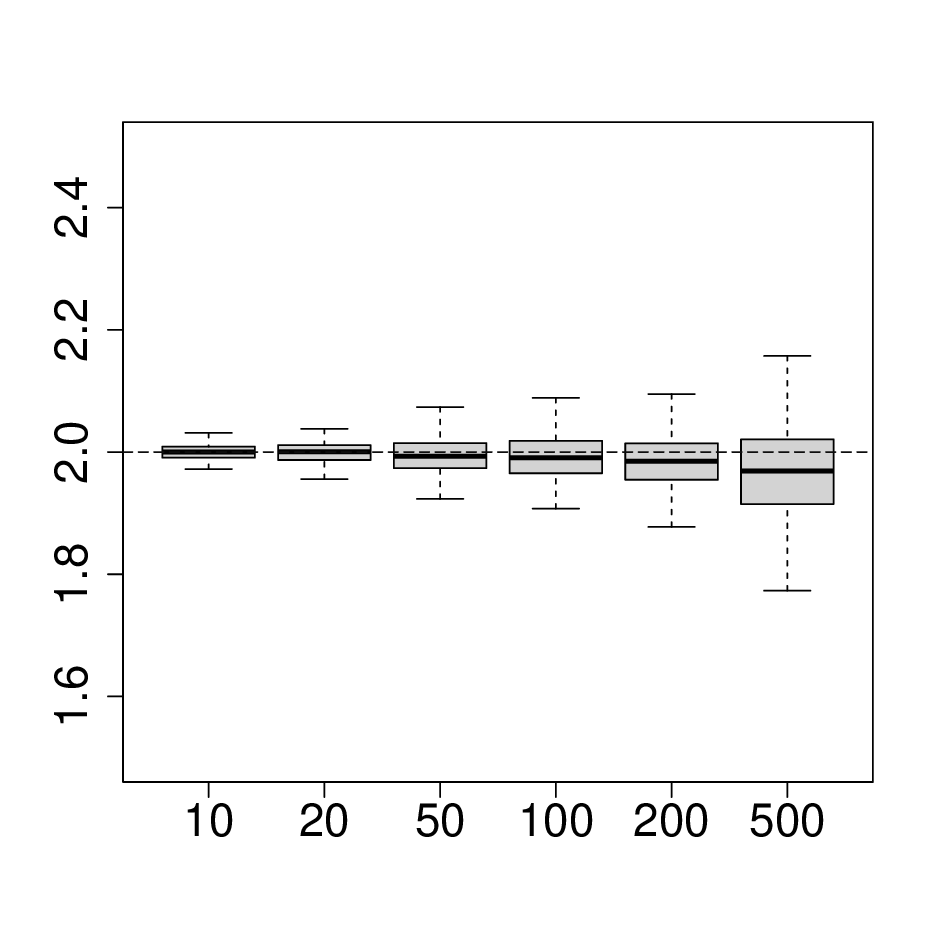}
	}
	\subfigure[KDI estimators with $d=5$.]{
		\includegraphics[scale = 0.45]{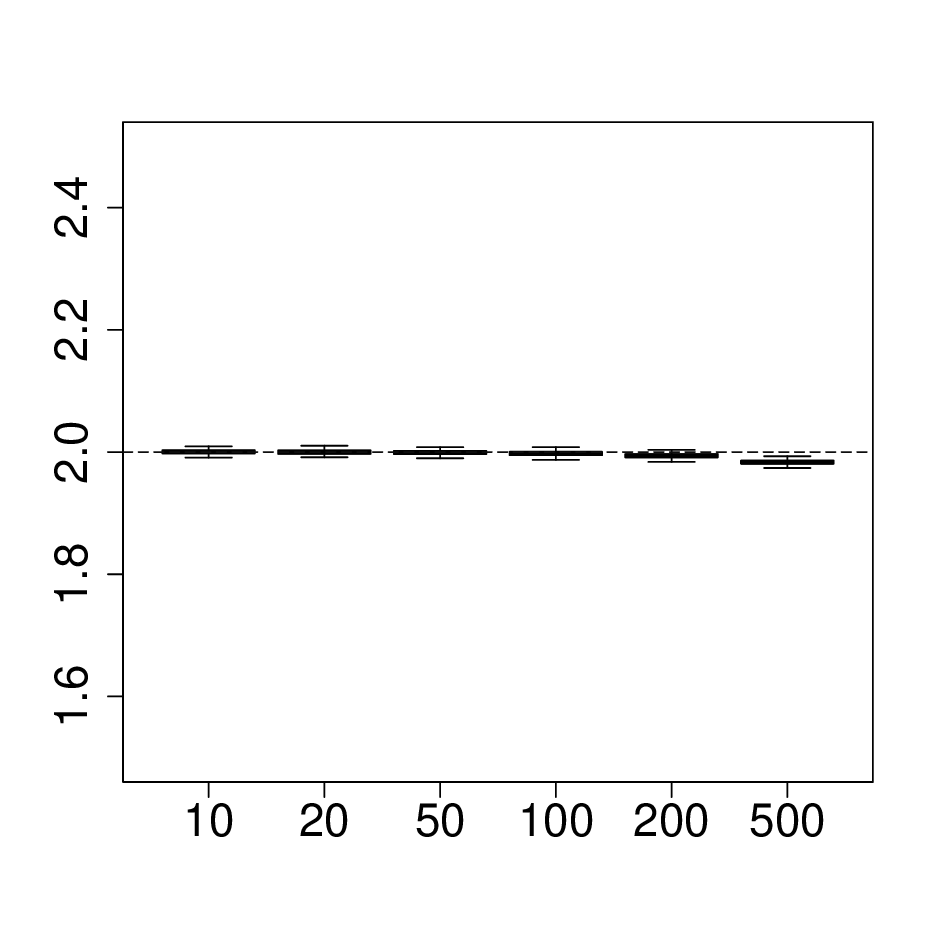}
	}
	\\
	\subfigure[Kernel-based SGM estimators with $d=15$.]{
		\includegraphics[scale = 0.45]{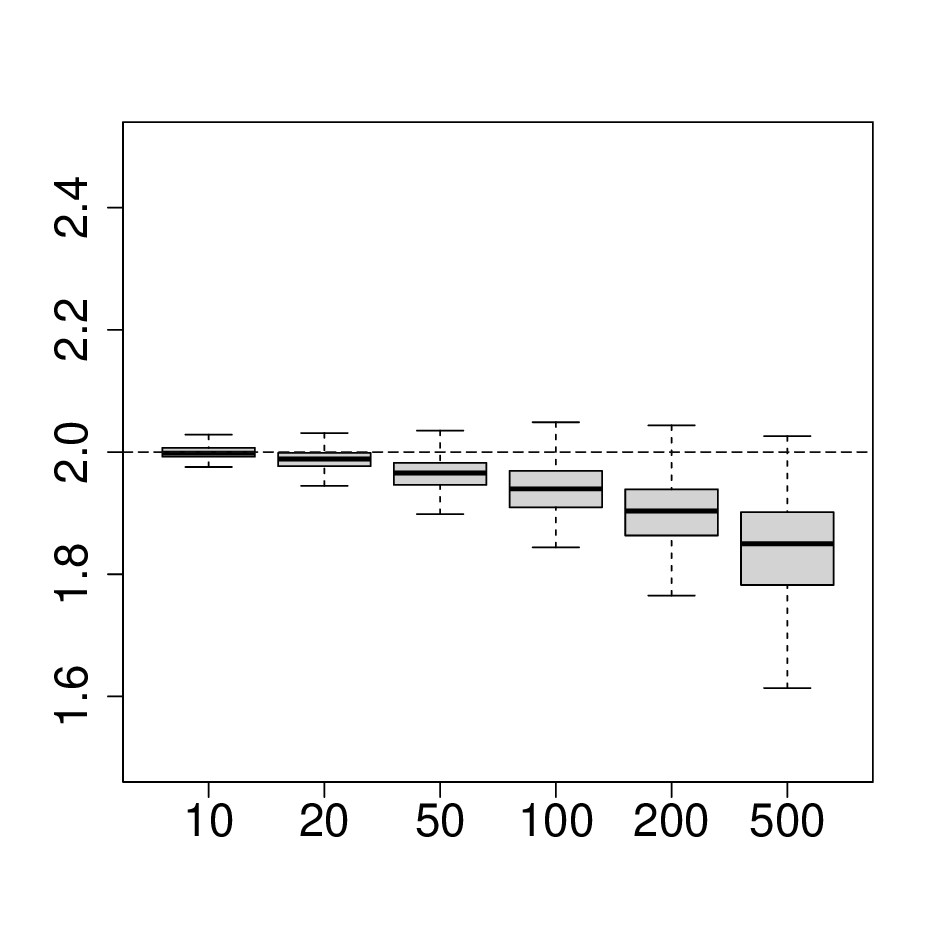}
	}
	\subfigure[KDI estimators with $d=15$.]{
		\includegraphics[scale = 0.45]{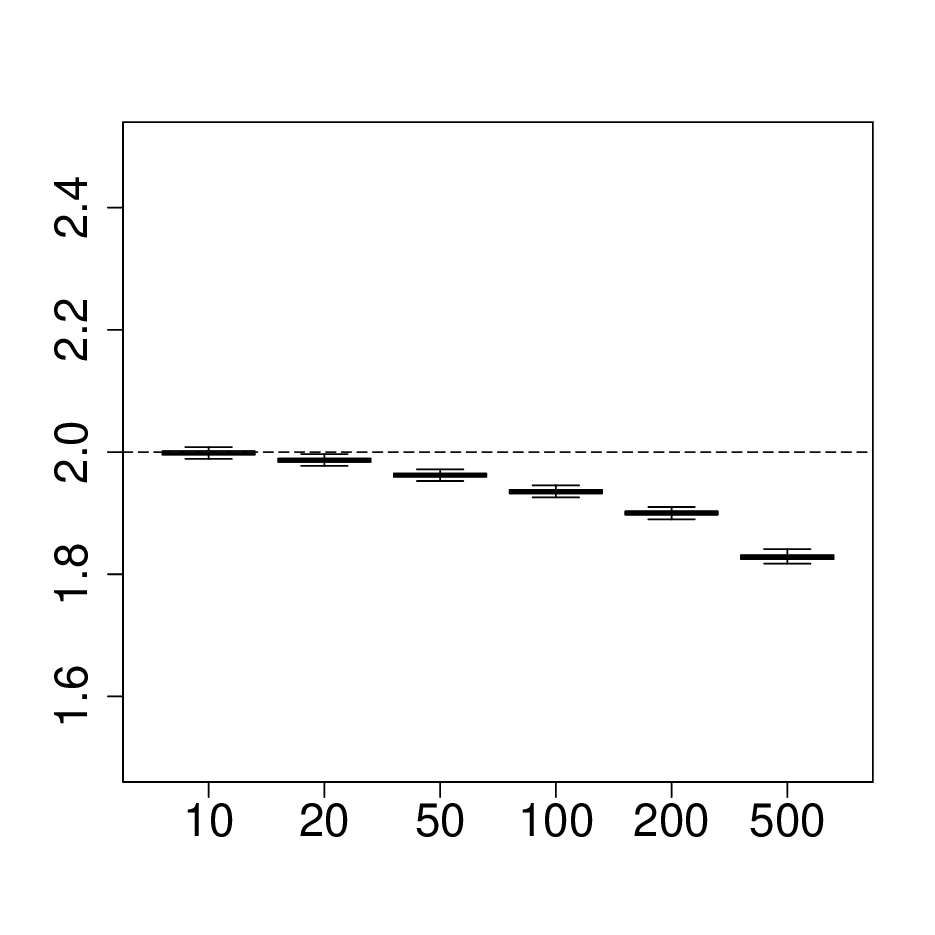}
	}
	\caption{Box-plots of the kernel-based estimators. Dashed line: the true parameter $\mu = 2$.}\label{fig: Linear_boxplot_kernel}
\end{figure}

\begin{table}[H]
	\centering
	\begin{tabular}{cccccc}
		\toprule
		& & \multicolumn{2}{c}{SGM} & \multicolumn{2}{c}{KDI}\\
		\cmidrule{3-6}
		$d$ &$L$ &\small{Bias} & \small{SE} & \small{Bias} & \small{SE}\\
		\midrule
		\multirow{6}{*}{5}
		&10  &0.000 &0.013  &0.000 &0.004\\
		&20 &-0.001 &0.017  &0.000 &0.004\\
		&50 &-0.006 &0.028 &-0.001 &0.004\\
		&100 &-0.008 &0.038 &-0.002 &0.004\\
		&200 &-0.016 &0.047 &-0.006 &0.004\\
		&500 &-0.035 &0.080 &-0.017 &0.004\\
		\midrule
		\multirow{6}{*}{15}
		&10 &-0.001 &0.011 &-0.001 &0.004\\
		&20 &-0.013 &0.015 &-0.013 &0.004\\
		&50 &-0.035 &0.025 &-0.038 &0.004\\
		&100 &-0.059 &0.039 &-0.065 &0.004\\
		&200 &-0.098 &0.054 &-0.100 &0.004\\
		&500 &-0.156 &0.086 &-0.172 &0.004\\
		\bottomrule
	\end{tabular}
	\caption{\label{table: Linear_kernel}The bias and SE of the kernel-based estimators.}
\end{table}

\begin{figure}[H]
	\centering
	\subfigure[Sieve-based SGM estimators with $d=5$.]{
		\includegraphics[scale = 0.45]{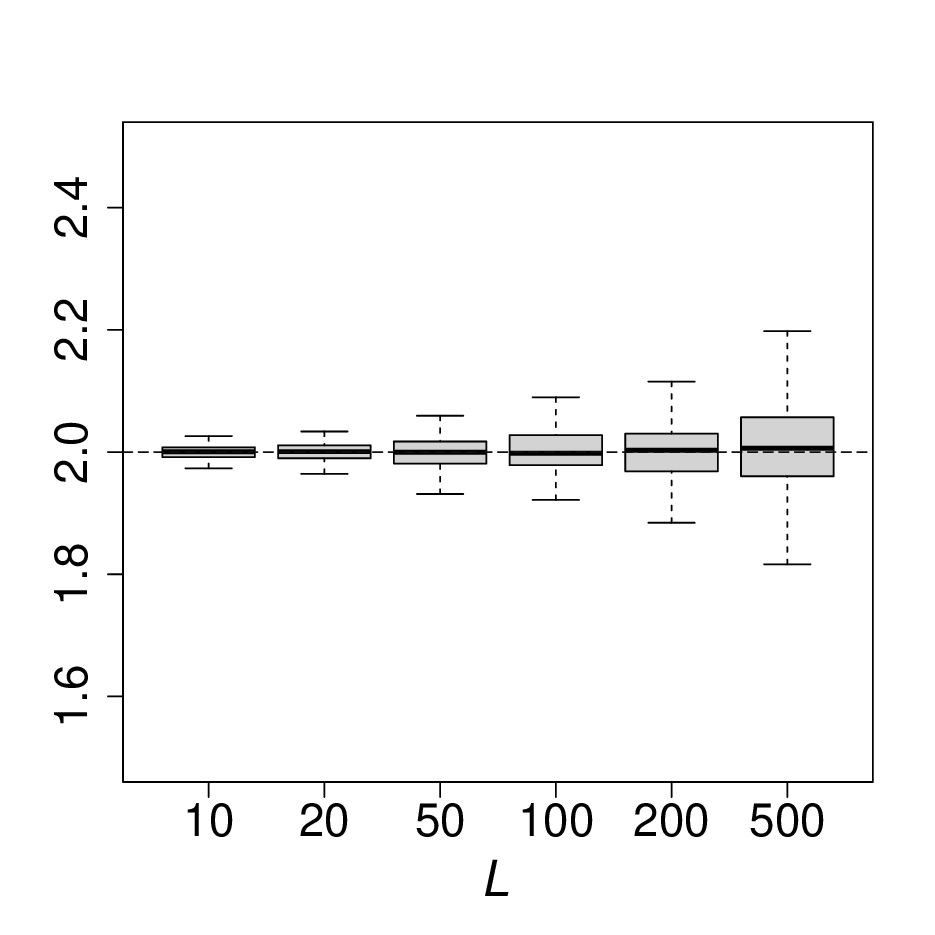}
	}
	\subfigure[SDI estimators with $d=5$.]{
		\includegraphics[scale = 0.45]{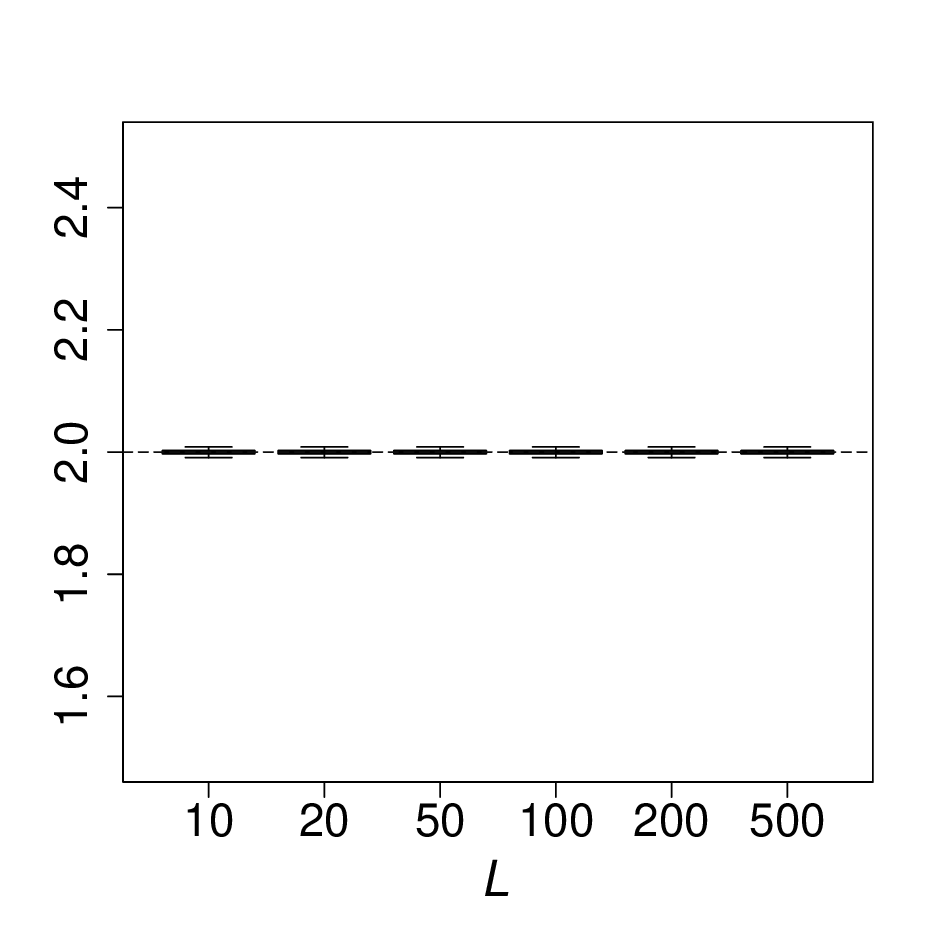}
	}
	\\
	\subfigure[Sieve-based SGM estimators with $d=15$.]{
		\includegraphics[scale = 0.45]{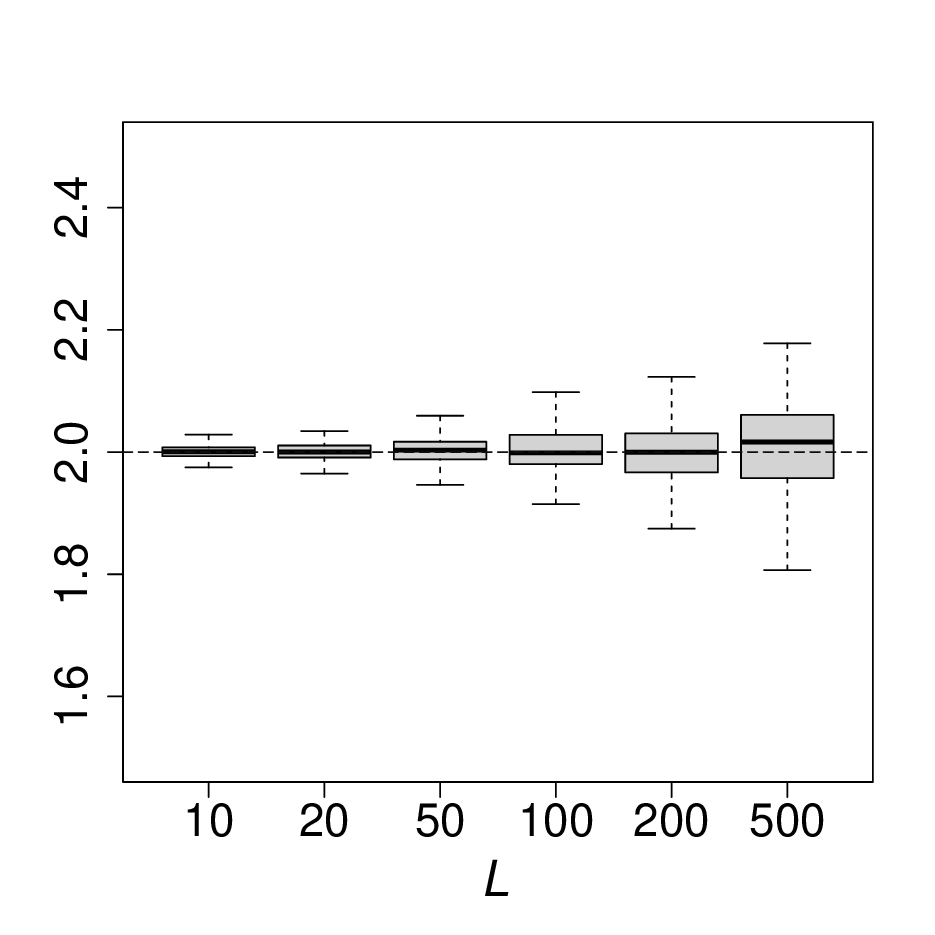}
	}
	\subfigure[SDI estimators with $d=15$.]{
		\includegraphics[scale = 0.45]{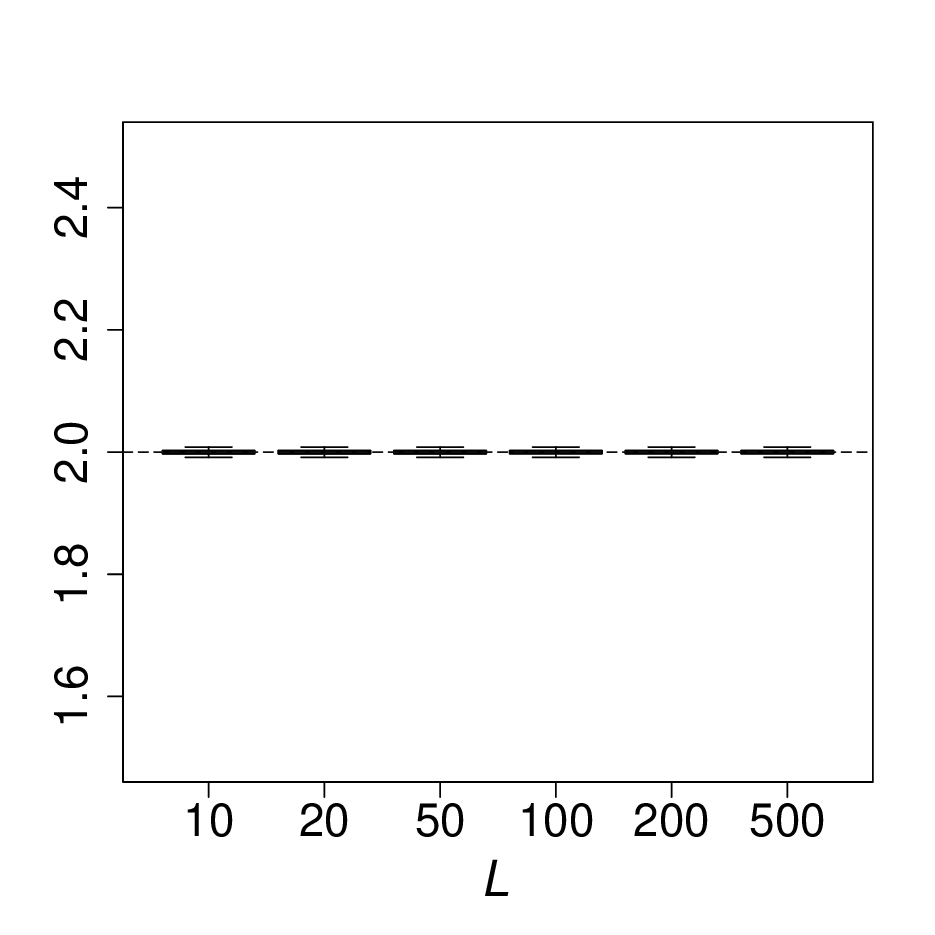}
	}\caption{Box-plots of the sieve-based estimators. Dashed line: the true parameter $\mu = 2$.}\label{fig: Linear_boxplot_sieve}
\end{figure}

\begin{table}[H]
	\centering
	\begin{tabular}{cccccc}
		\toprule
		& & \multicolumn{2}{c}{SGM} & \multicolumn{2}{c}{SDI}\\
		\cmidrule{3-6}
		$d$ &$L$ &\small{Bias} & \small{SE} & \small{Bias} & \small{SE}\\
		\midrule
		\multirow{6}{*}{5}
		&10  &0.000 &0.011    &0.000 &0.004\\
		&20  &0.000 &0.015    &0.000 &0.004\\
		&50 &-0.001 &0.026    &0.000 &0.004\\
		&100  &0.001 &0.036    &0.000 &0.004\\
		&200  &0.000 &0.046    &0.000 &0.004\\
		&500  &0.007 &0.074    &0.000 &0.004\\
		\midrule
		\multirow{6}{*}{15}
		&10  &0.000 &0.011    &0.000 &0.003\\
		&20  &0.001 &0.014    &0.000 &0.003\\
		&50  &0.002 &0.023    &0.000 &0.003\\
		&100  &0.004 &0.036    &0.000 &0.003\\
		&200 &-0.002 &0.049    &0.000 &0.003\\
		&500  &0.009 &0.077    &0.000 &0.003\\
		\bottomrule
	\end{tabular}
	\caption{\label{table: Linear_sieve}The bias and SE of the sieve-based estimators.}
\end{table}

\begin{figure}[H]
	\centering
	\subfigure[Kernel-based SGM estimators with $d=5$.]{
		\includegraphics[scale = 0.45]{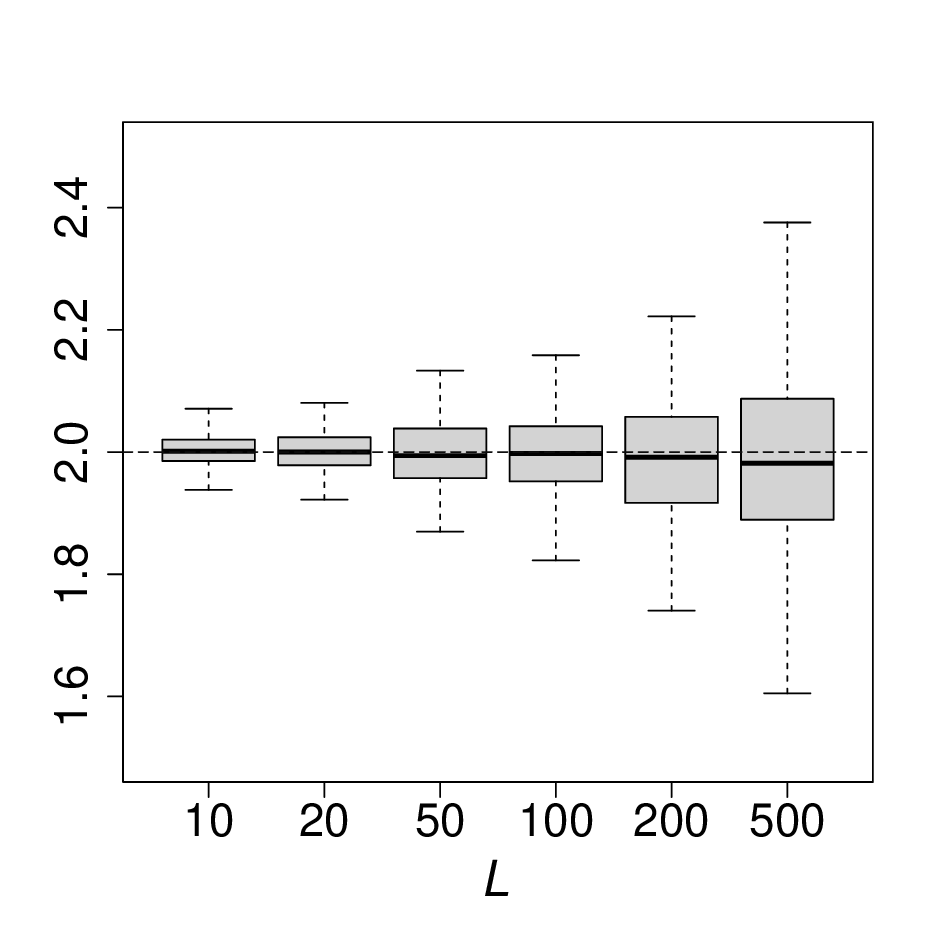}
	}
	\subfigure[KDI estimators with $d=5$.]{
		\includegraphics[scale = 0.45]{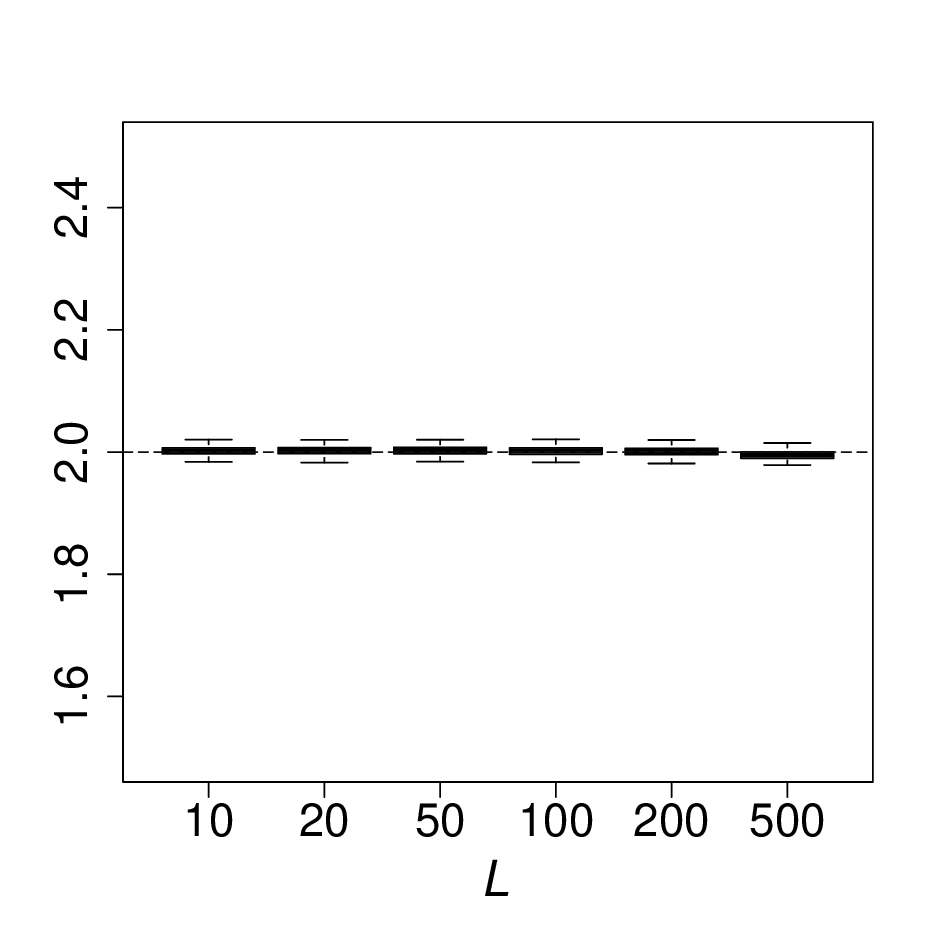}
	}
	\\
	\subfigure[Kernel-based SGM estimators with $d=15$.]{
		\includegraphics[scale = 0.45]{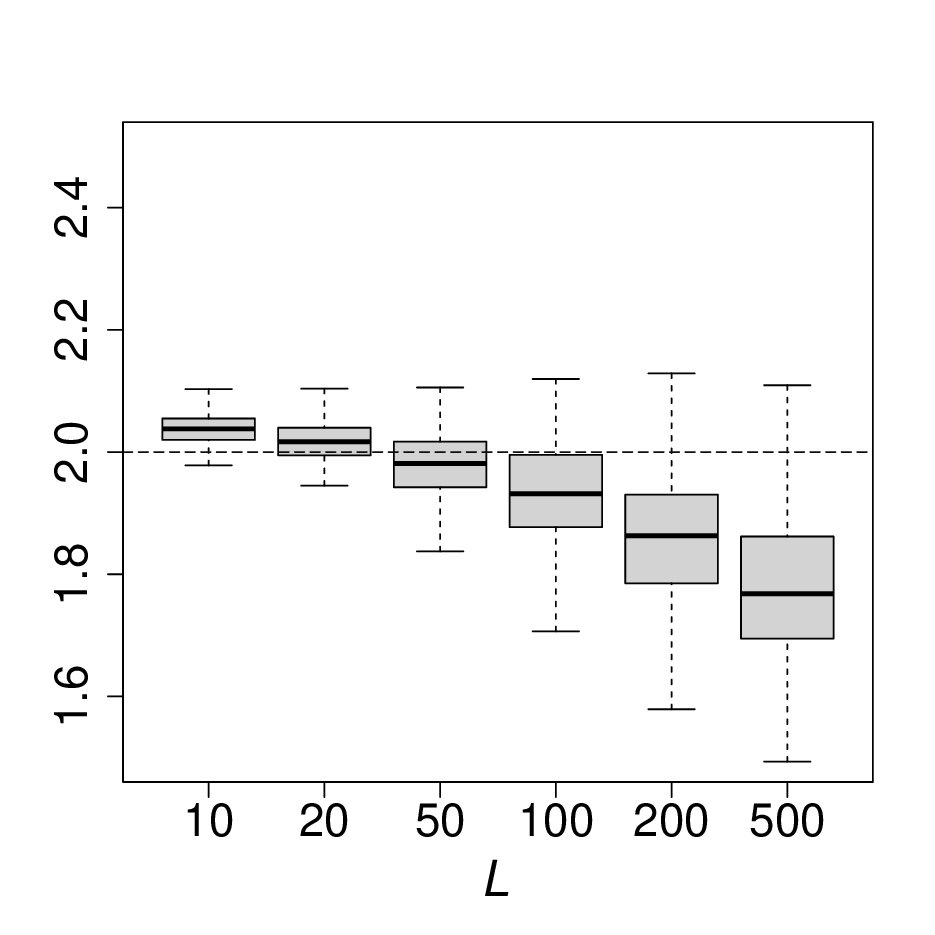}
	}
	\subfigure[KDI estimators with $d=15$.]{
		\includegraphics[scale = 0.45]{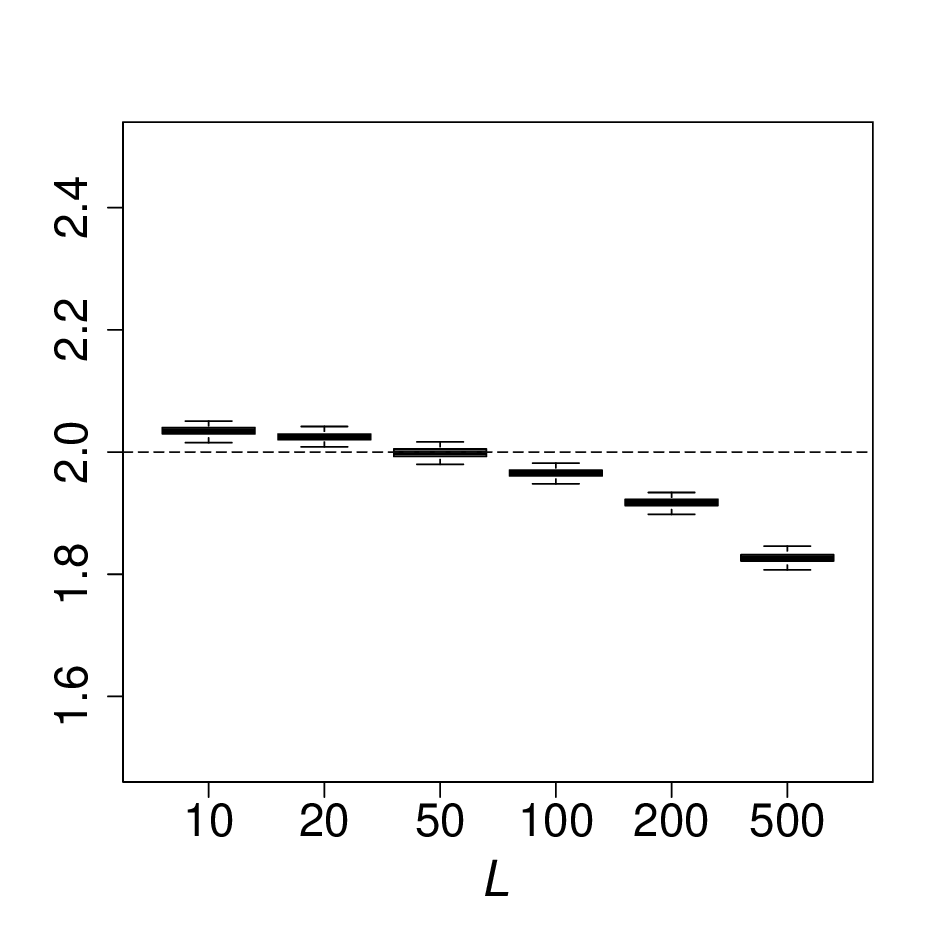}
	}
	\caption{Box-plots of the kernel-based estimators in the nonlinear setting. Dashed line: the true parameter $\mu = 2$.}\label{fig: Nonlinear_boxplot_kernel}
\end{figure}

\begin{table}[H]
	\centering
	\begin{tabular}{cccccc}
		\toprule
		& & \multicolumn{2}{c}{SGM} & \multicolumn{2}{c}{KDI}\\
		\cmidrule{3-6}
		$d$ &$L$ &\small{Bias} & \small{SE} & \small{Bias} & \small{SE}\\
		\midrule
		\multirow{6}{*}{5}
		&10  &0.002 &0.024  &0.002 &0.007\\
		&20  &0.001 &0.031  &0.002 &0.007\\
		&50 &-0.003 &0.055  &0.002 &0.007\\
		&100 &-0.004 &0.071  &0.002 &0.007\\
		&200 &-0.010 &0.097  &0.001 &0.007\\
		&500 &-0.009 &0.157 &-0.005 &0.007\\
		\midrule
		\multirow{6}{*}{15}
		&10  &0.036 &0.024  &0.035 &0.007\\
		&20  &0.020 &0.033  &0.025 &0.007\\
		&50 &-0.018 &0.055 &-0.001 &0.008\\
		&100 &-0.068 &0.080 &-0.034 &0.007\\
		&200 &-0.143 &0.106 &-0.083 &0.007\\
		&500 &-0.217 &0.133 &-0.173 &0.007 \\
		\bottomrule
	\end{tabular}
	\caption{\label{table: Nonlinear_kernel}The bias and SE of the kernel-based estimators in the nonlinear setting.}
\end{table}

\begin{figure}[H]
	\centering
	\subfigure[Sieve-based SGM estimators with $d=5$.]{
		\includegraphics[scale = 0.45]{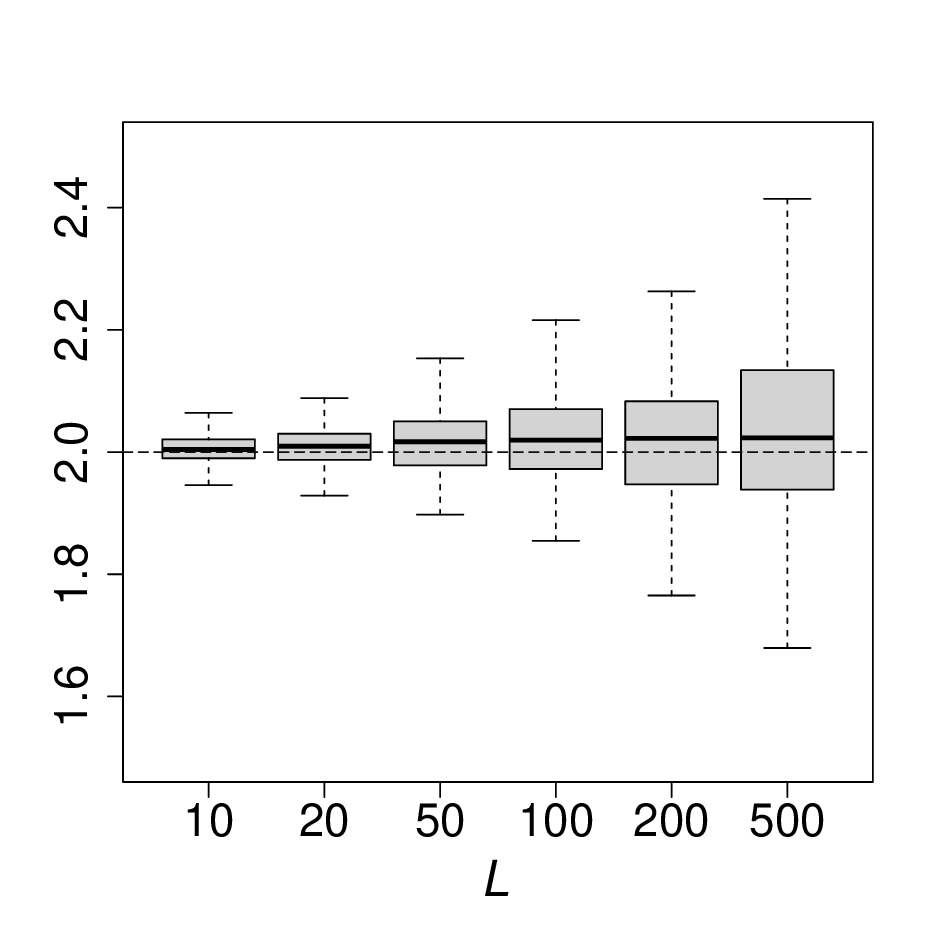}
	}
	\subfigure[SDI estimators with $d=5$.]{
		\includegraphics[scale = 0.45]{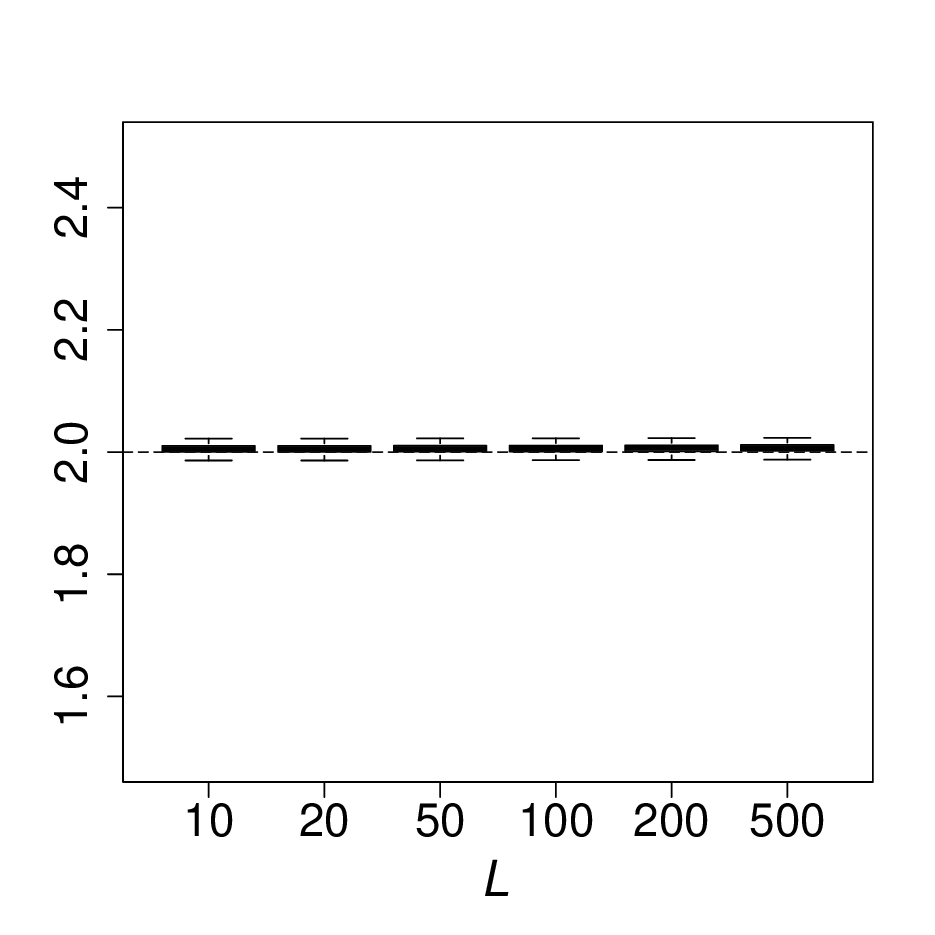}
	}
	\\
	\subfigure[Sieve-based SGM estimators with $d=15$.]{
		\includegraphics[scale = 0.45]{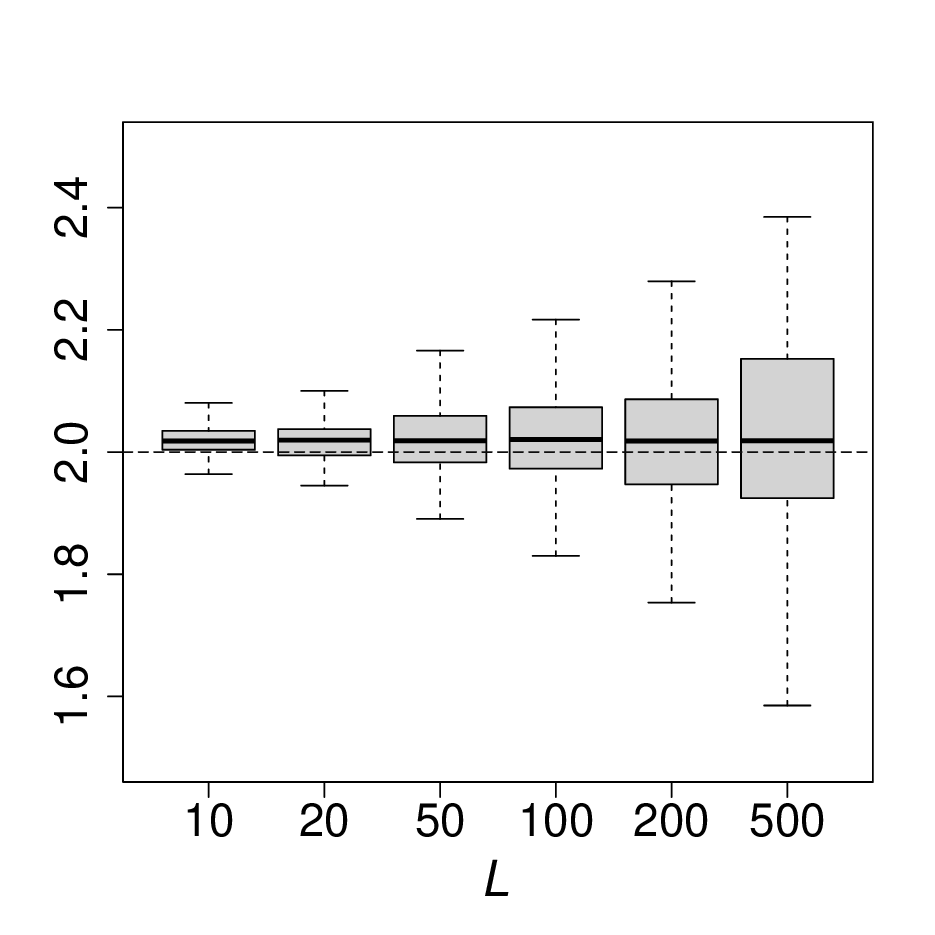}
	}
	\subfigure[SDI estimators with $d=15$.]{
		\includegraphics[scale = 0.45]{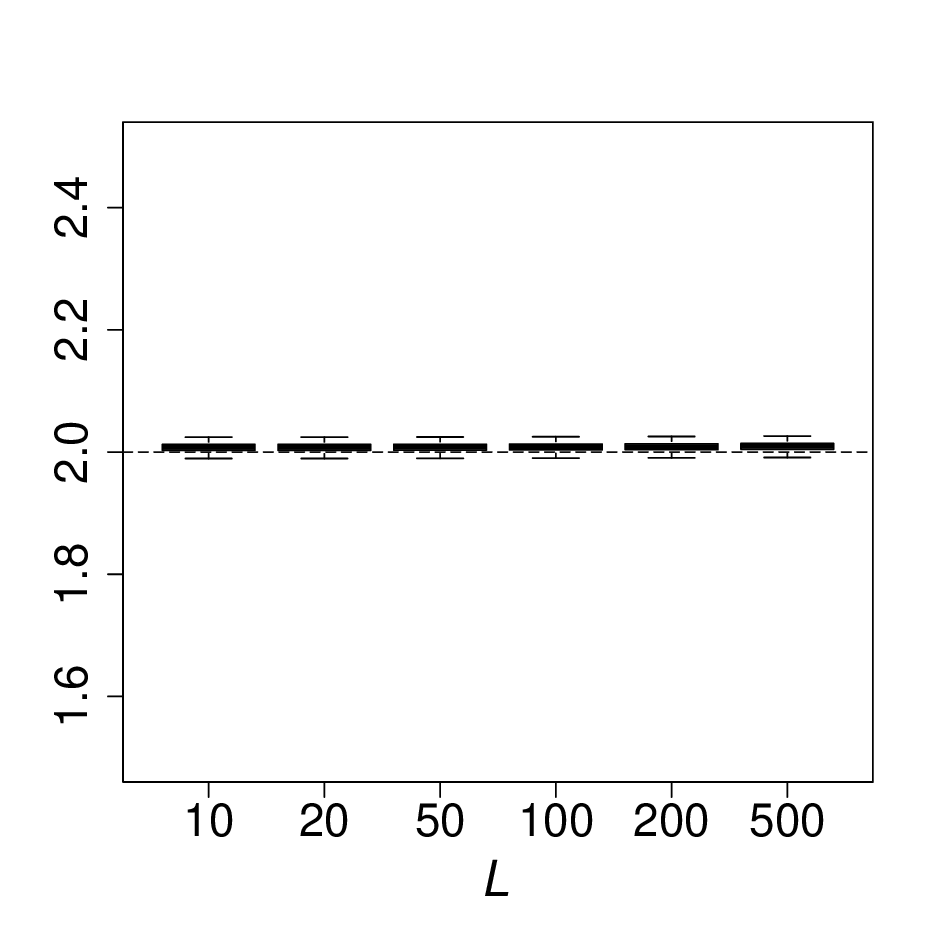}
	}\caption{Box-plots of the sieve-based estimators in the nonlinear setting. Dashed line: the true parameter $\mu = 2$.}\label{fig: Nonlinear_boxplot_sieve}
\end{figure}

\begin{table}[H]
	\centering
	\begin{tabular}{cccccc}
		\toprule
		& & \multicolumn{2}{c}{SGM} & \multicolumn{2}{c}{SDI}\\
		\cmidrule{3-6}
		$d$ &$L$ &\small{Bias} & \small{SE} & \small{Bias} & \small{SE}\\
		\midrule
		\multirow{6}{*}{5}
		&10  &0.032 &0.153 &0.007 &0.007\\
		&20 &0.016 &0.095 &0.006 &0.007\\
		&50 &0.023 &0.071 &0.006 &0.007\\
		&100 &0.016 &0.051 &0.005 &0.007\\
		&200 &0.009 &0.031 &0.005 &0.007\\
		&500 &0.004 &0.023 &0.005 &0.007\\
		\midrule
		\multirow{6}{*}{15}
		&10  &0.033 &0.156 &0.009 &0.007\\
		&20 &0.016 &0.099 &0.009 &0.007\\
		&50 &0.023 &0.073 &0.008 &0.007\\
		&100 &0.021 &0.054 &0.008 &0.007\\
		&200 &0.018 &0.031 &0.008 &0.007\\
		&500 &0.018 &0.023 &0.008 &0.007\\
		\bottomrule
	\end{tabular}
	\caption{\label{table: Nonlinear_sieve}The bias and SE of the sieve-based estimators in the nonlinear setting.}
\end{table}

\begin{table}[H]
	
	\begin{center}
		\begin{tabular}{*{6}{c}}
			\toprule
			&&\multicolumn{2}{c}{Linear} & \multicolumn{2}{c}{Noninear}\\
			\cmidrule{3-6} 
			&\small{$L$}& \small{$d=5$} & \small{$d=15$} & \small{$d=5$} & \small{$d=15$} \\
			\midrule
			\multirow{5}{20pt}{KDI} 
			&$10$&       310.10     &  1037.62     &   220.99    &     995.49\\
			&$20$&        57.44     &   236.55     &    37.33    &     229.99\\
			&$50$&         6.58     &    18.58     &     4.28    &      18.30\\
			&$100$&         1.68    &      3.16    &      1.11   &        3.18\\
			&$200$&         0.45    &      0.79    &      0.30   &        0.80\\
			&$500$&         0.09   &       0.14    &      0.06   &        0.14\\
			\specialrule{0em}{-3pt}{-3pt}\\
			\multicolumn{2}{l}{Non-Dist} 
			&42582.79 &131884.11  &45299.34 &132325.88 \\
			\midrule
			\multirow{5}{20pt}{SDI} 
			&$10$&         10.80   &   34.20   &    10.82 &       30.97\\
			&$20$&       7.025  &     19.08   &      7.01   &      19.66\\
			&$50$&        3.30   &     10.35   &     3.32    &       10.25 \\
			&$100$&        2.10     &    6.37    &     2.04    &      6.54\\
			&$200$&       1.15   &      3.89   &     1.16  &       3.69\\
			&$500$&         0.57     &      2.47  &     0.54    &     2.31 \\
			\specialrule{0em}{-3pt}{-3pt}\\
			\multicolumn{2}{l}{Non-Dist} 
			&   50.42 & 122.88  & 49.99 & 122.22\\
			\bottomrule
		\end{tabular}
	\end{center}
	\caption{CPU times of the classical non-distributed kernel/sieve estimators and the KDI and SDI estimators with different $L$.}\label{table: simtime}
\end{table}

\begin{table}[H]
	\begin{center}
		\begin{tabular}{*{7}{c}}
			\toprule
			\small{$c$}& \small{$0.1$} & \small{$0.5$} & \small{$0.9$} & \small{$1.3$} & \small{$1.7$} & \small{$2.1$}\\
			\midrule
			KDI&    1.54 &1.54 &1.55 &1.56 &1.56 &1.51\\
			SDI&   0.32 & 1.79 & 3.45&  6.01 & 9.51 &12.56\\
			\bottomrule
		\end{tabular}
	\end{center}
	\caption{CPU times of the KDI/SDI estimators in the nonlinear setting with $d=5$, $L = 100$ and different $c$'s in $\mathcal{C}$.}\label{table: time tuning parameter}
\end{table}

 ~\\
\newpage

\begin{figure}[H]
	\centering
	\subfigure[KDI estimator with $d=5$.]{
		\includegraphics[scale = 0.4]{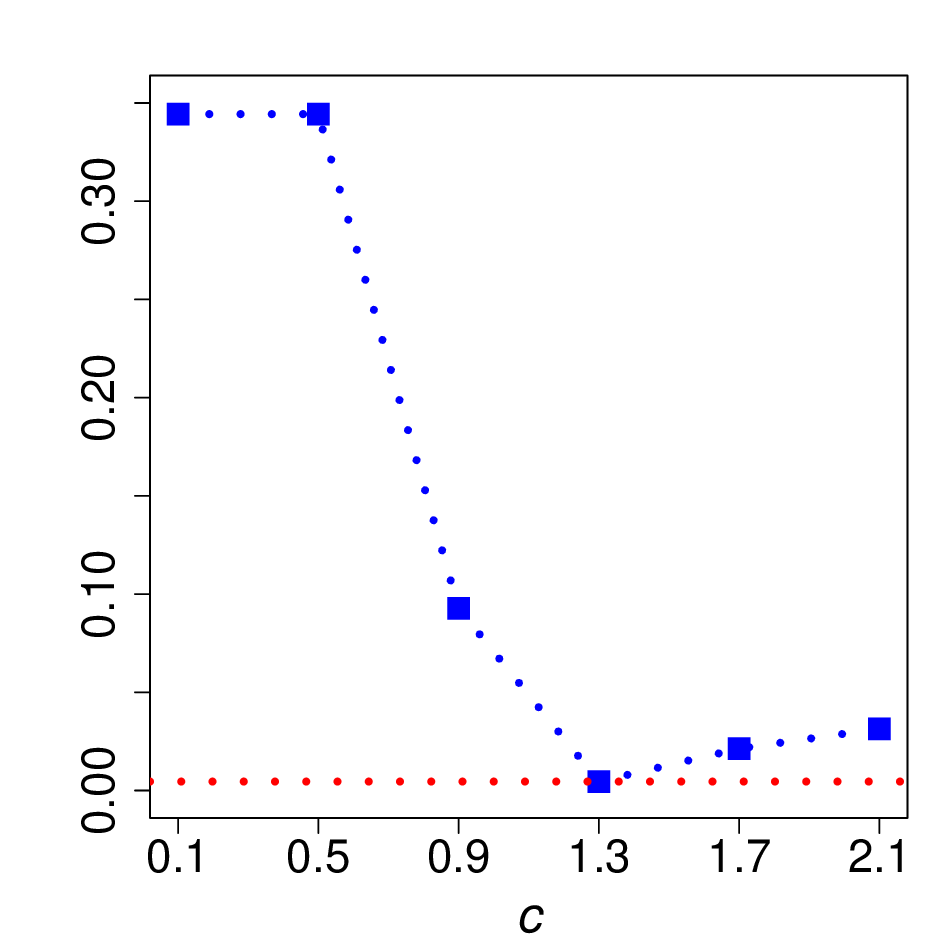}
	}
	\subfigure[KDI estimator with $d=15$.]{
		\includegraphics[scale = 0.4]{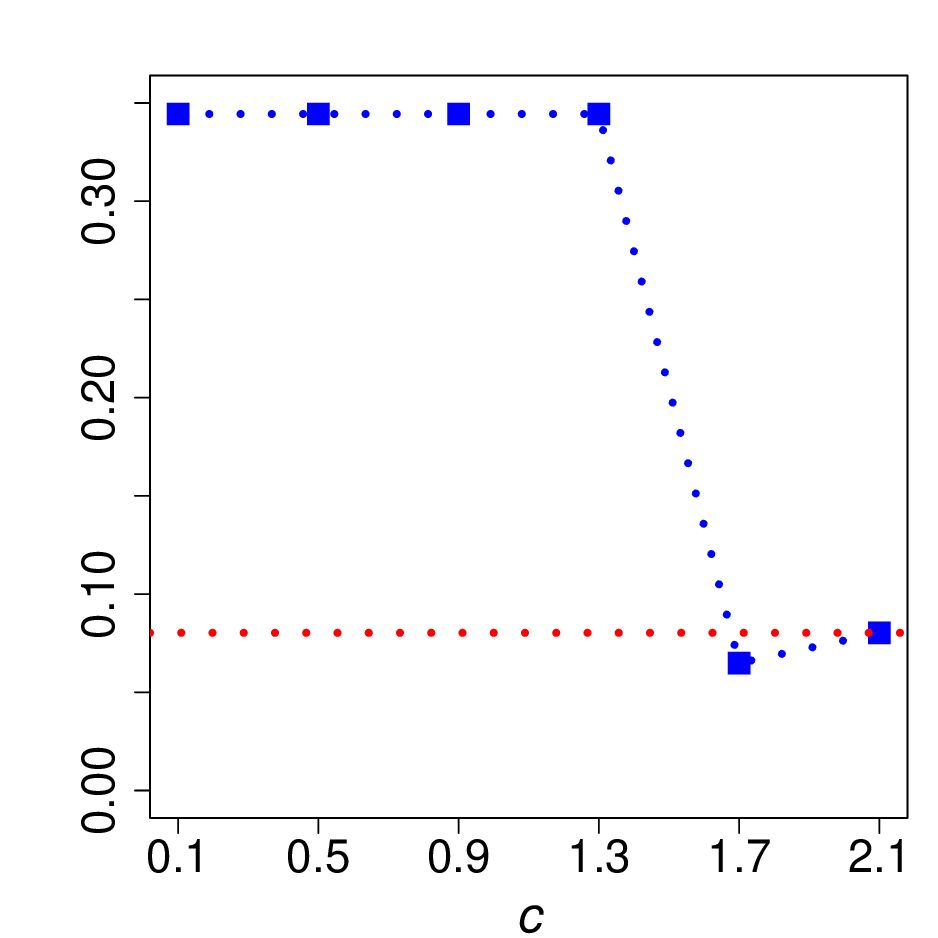}
	}
	\subfigure[SDI estimator with $d=5$.]{
		\includegraphics[scale = 0.4]{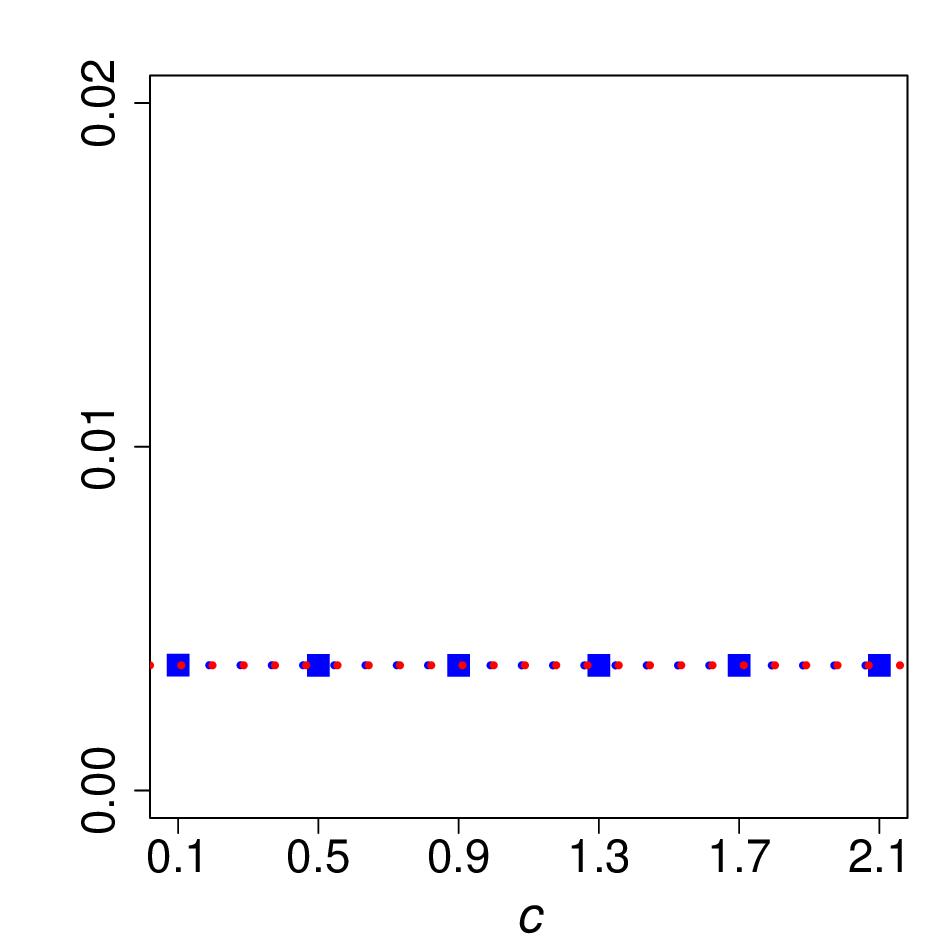}
	}
	\subfigure[SDI estimator with $d=15$.]{
		\includegraphics[scale = 0.4]{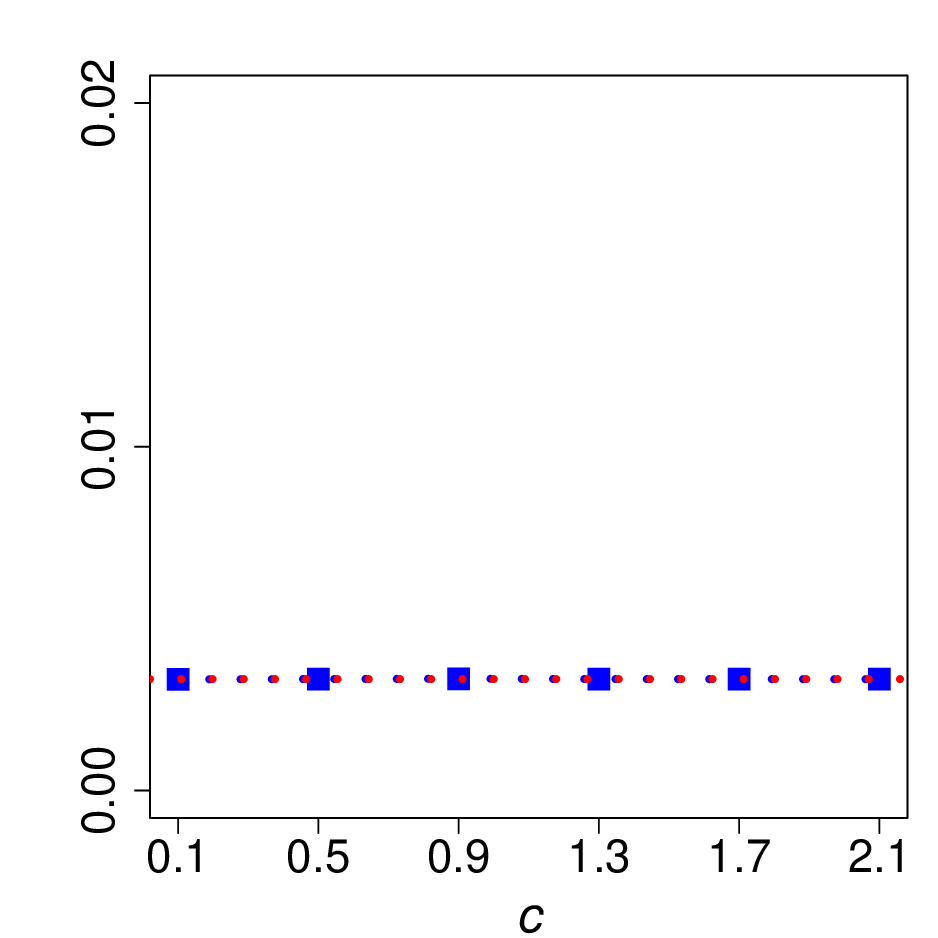}
	}
	\caption{Comparison of RMSE in the linear setting with $L = 100$. Solid line: the RMSE of KDI and SDI estimators with different $c$ respectively in the linear setting. Dashed line: the RMSE of KDI and SDI estimators estimator with selected bandwidth by the proposed DWCV algorithm.}\label{fig: Linear_bws}
\end{figure}

\begin{figure}[H]
	\centering
	\subfigure[KDI estimator with $d=5$.]{
		\includegraphics[scale = 0.4]{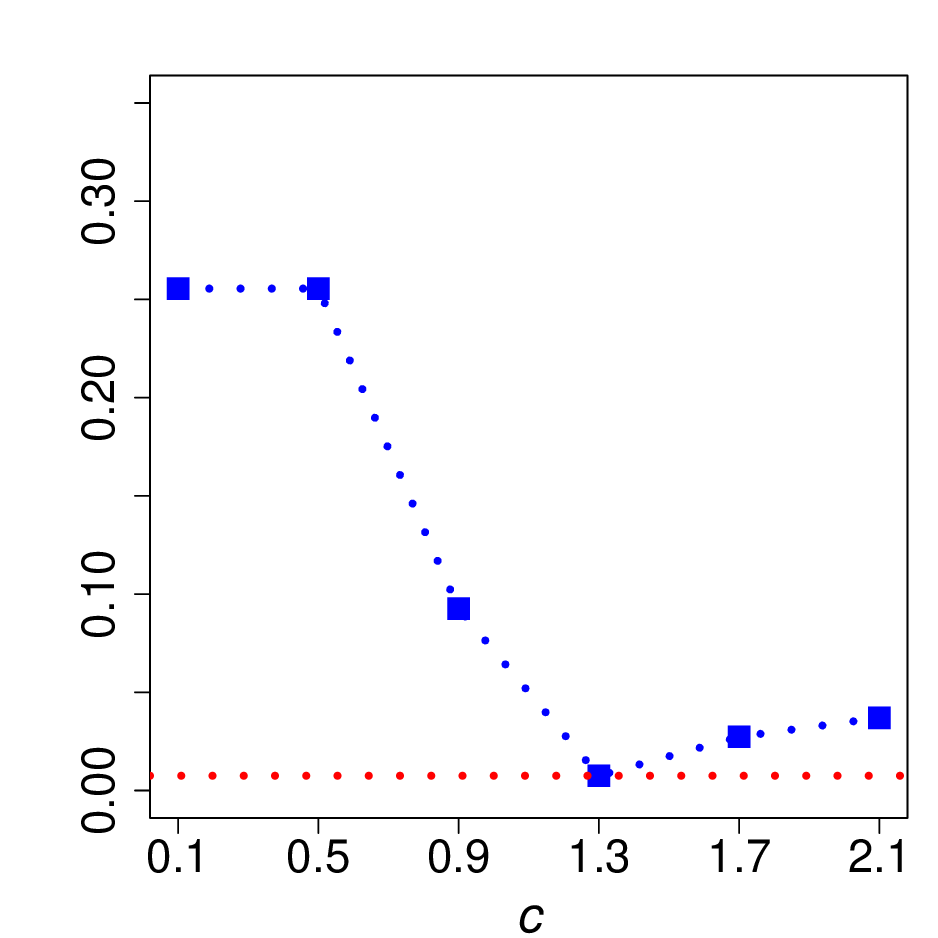}
	}
	\subfigure[KDI estimator with $d=15$.]{
		\includegraphics[scale = 0.4]{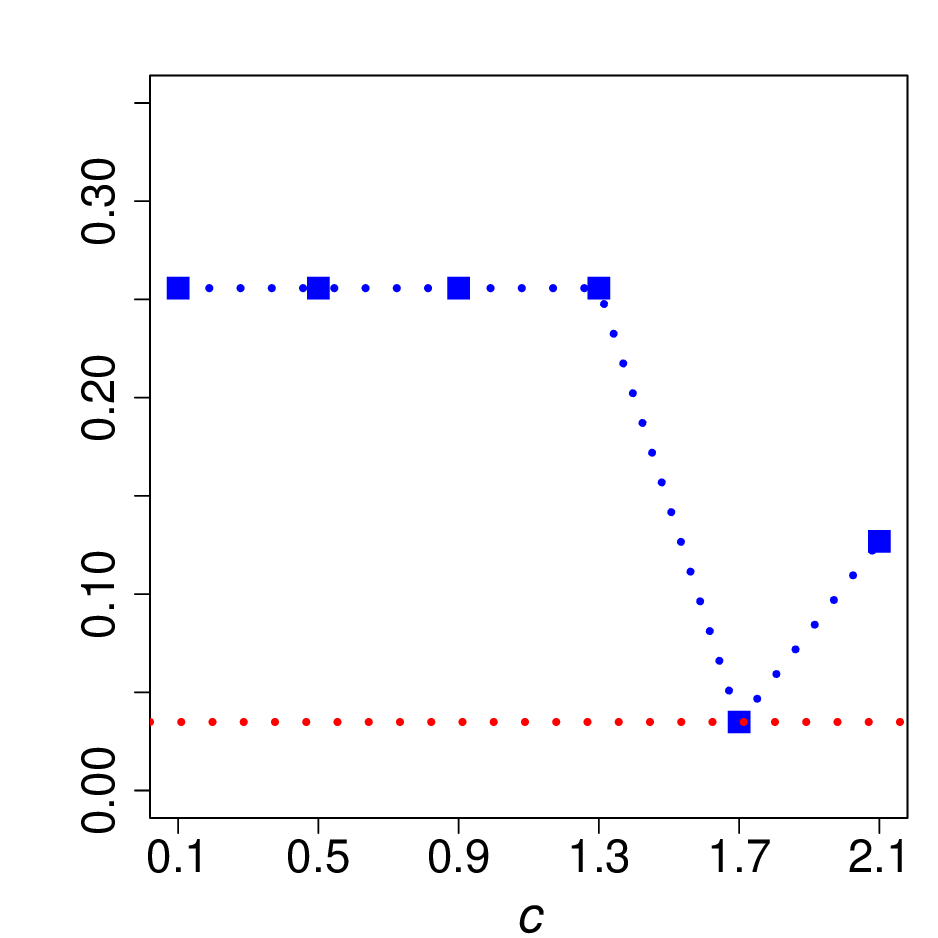}
	}
	\subfigure[SDI estimator with $d=5$.]{
		\includegraphics[scale = 0.4]{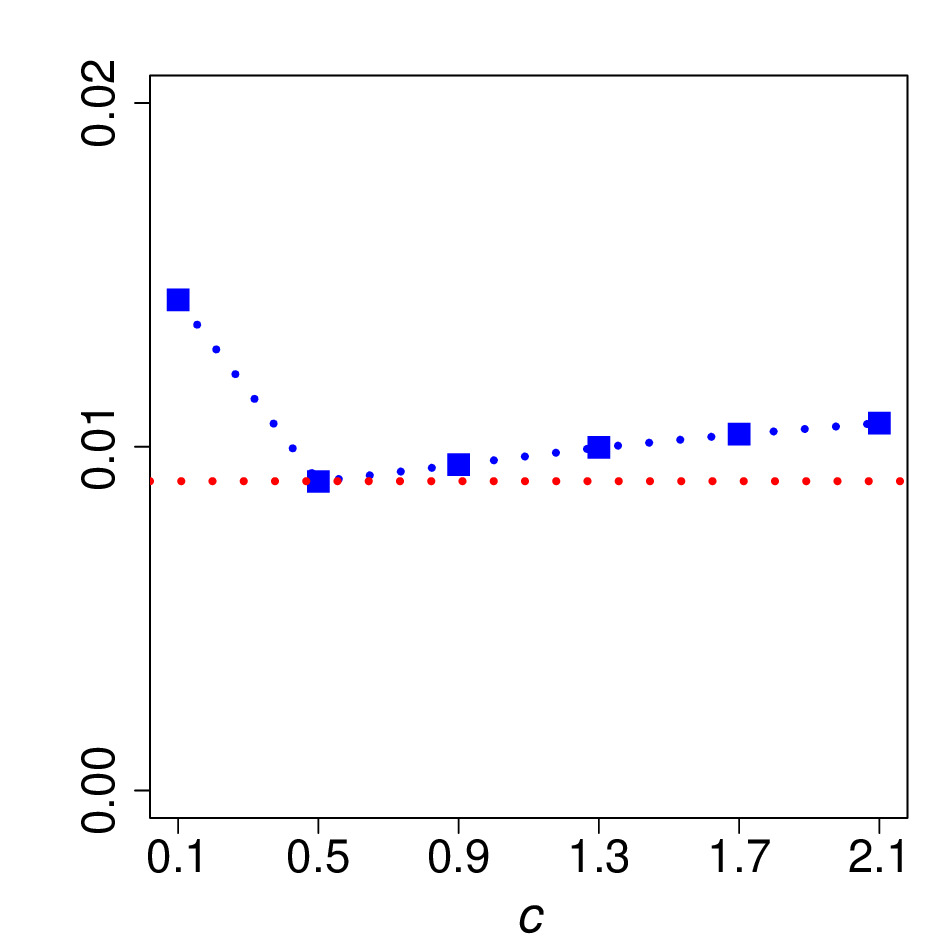}
	}
	\subfigure[SDI estimator with $d=15$.]{
		\includegraphics[scale = 0.4]{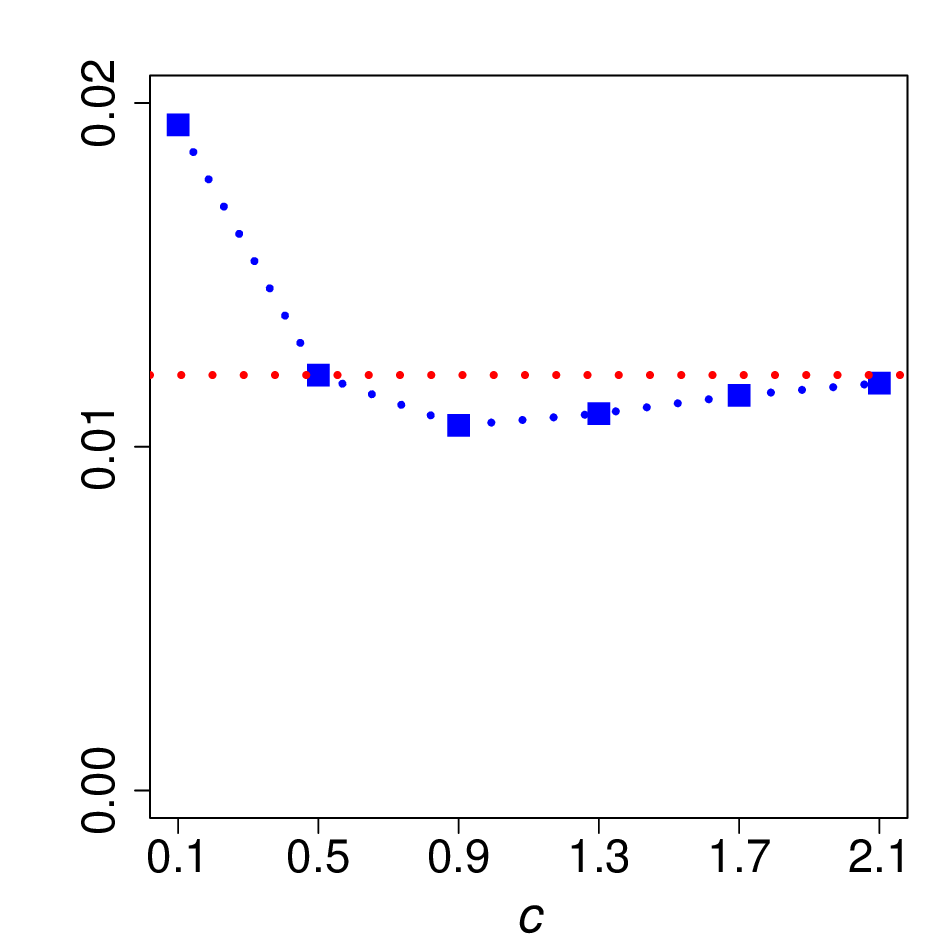}
	}
	\caption{Comparison of RMSE in the nonlinear setting with $L = 100$. Solid line: the RMSE of KDI and SDI estimators with different $c$ in the nonlinear setting. Dashed line: the RMSE of KDI and SDI estimators with selected bandwidth by the proposed DWCV algorithm.}\label{fig: Nonlinear_bws}
\end{figure}

\begin{table}[H]
	\centering
	\begin{tabular}{cccccc}
		\toprule
		&\multicolumn{2}{c}{$d=5$} & \multicolumn{2}{c}{$d=15$}\\
		\cmidrule{2-5}
		$L$ &\small{Bias} & \small{SE}  &\small{Bias} & \small{SE}\\
		\midrule
		10  &0.000 &0.010 &0.002 &0.010\\
		20  &0.000 &0.010 &0.002 &0.010\\
		50  &0.000 &0.010 &0.003 &0.010\\
		100 &0.000 &0.010 &0.004 &0.010\\
		200 &0.000 &0.010 &0.005 &0.010\\
		500 &0.001 &0.010 &0.007 &0.010\\
		\bottomrule
	\end{tabular}
	\caption{\label{table: br}The bias and SE of the SDI method when $1/\pi(x)$ can be approximated well.}
\end{table}

\begin{figure}[h]
	\centering
	\subfigure[Kernel-based estimators]{
		\includegraphics[scale = 0.45]{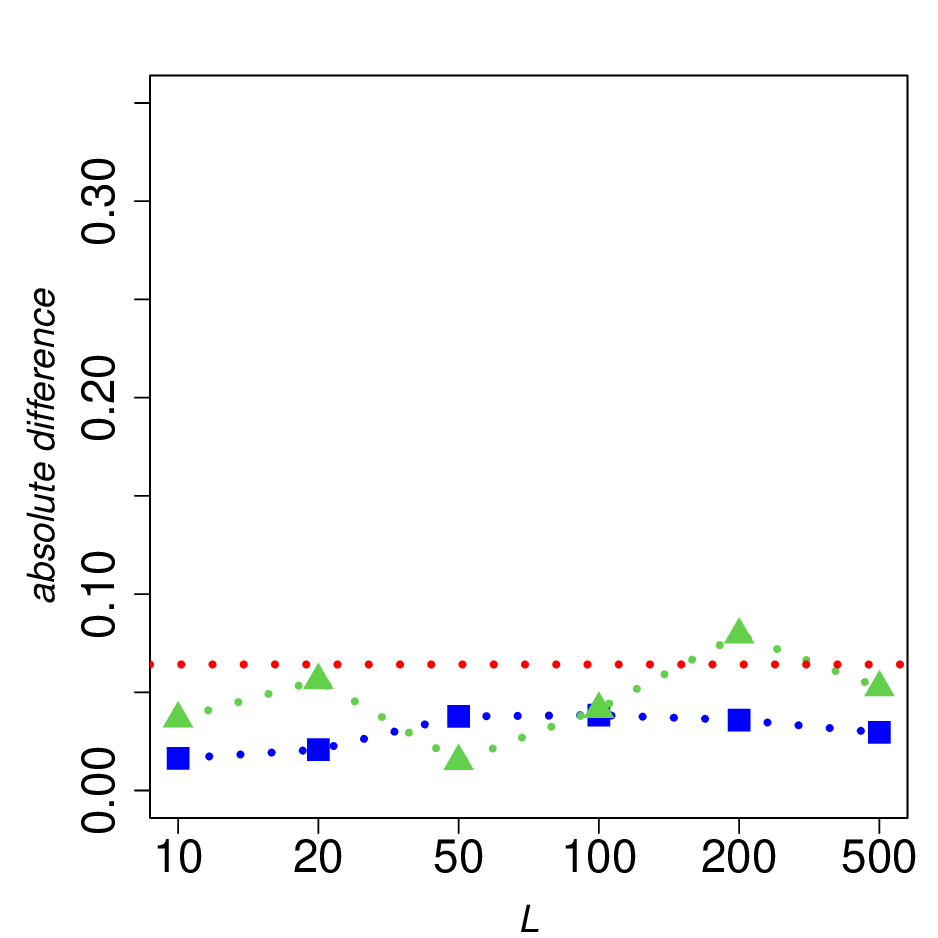}
	}
	\subfigure[Sieve-based estimators]{
		\includegraphics[scale = 0.45]{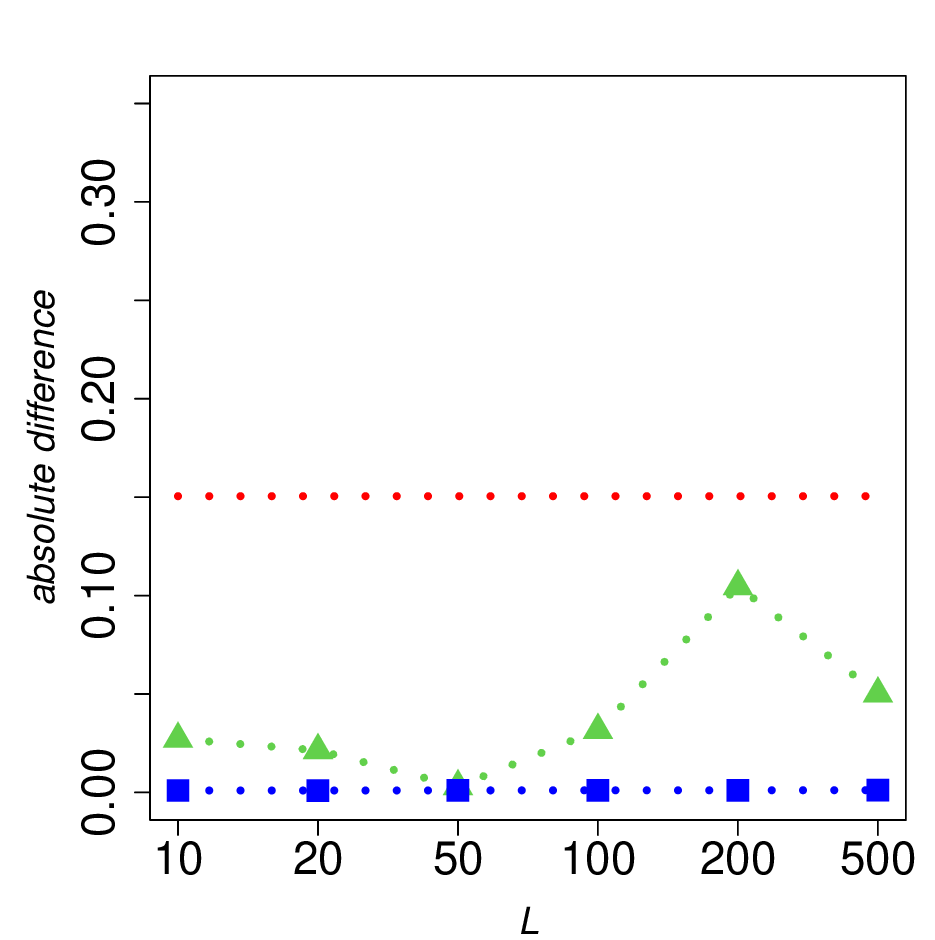}
	}
	\caption{The absolute value curves of the differences 
		between the estimators and their benchmarks. Red dotted lines are for the complete-case estimator, dotted lines with green 
		triangle are for SGM estimators and dotted lines with blue square are for distributed imputation estimators.
		The benchmark is $\hat{\mu}_{\mathbb{K}}$ in (a) and $\hat{\mu}_{\mathbb{S}}$ in (b), respectively.}\label{fig: movie}
\end{figure}


\begin{table}[H]
	\centering
	\begin{tabular}{ccccccc}
		\toprule
		& & & \multicolumn{2}{c}{DCV} & \multicolumn{2}{c}{DWCV}\\
		\cmidrule{4-5}\cmidrule{6-7}
		Method & Setting &$d$ &\small{RMSE} & \small{time} & \small{RMSE} & \small{time}\\
		\midrule
		\multirow{4}{*}{KDI} &\multirow{2}{*}{Linear}
		&5 &0.021 &8.45 &0.005 &8.49\\
		& &15 &0.080 &24.38 &0.077 &24.68\\
		\cmidrule{2-7}
		& \multirow{2}{*}{Nonlinear}
		&5  &0.008 &8.46  &0.008 &8.49 \\
		& &15 &0.127 &24.63  &0.036 &24.78\\
		\midrule
		\multirow{4}{*}{SDI} & \multirow{2}{*}{Linear}
		&5  &0.004 &3.50  &0.004 &2.34 \\
		& &15 &0.004 &4.74  &0.004 &4.34\\
		\cmidrule{2-7}
		& \multirow{2}{*}{Nonlinear}
		&5 &0.009 &2.03 &0.009 &2.33\\
		& &15 &0.012 &3.58 &0.012 &4.32\\
		\bottomrule
	\end{tabular}
	\caption{\label{table: ablation} RMSE and computing time of the KDI and SDI estimators.}
\end{table}

\end{document}